\documentclass{article}
\usepackage[utf8]{inputenc}
\usepackage{tikz}
\usepackage{caption}
\usepackage[utf8]{inputenc}
\title{CDS directions}
\usetikzlibrary{calc}
\usepackage{todonotes}
\usepackage{amsmath,amssymb,amsthm,graphicx,caption}
\usepackage{fullpage}
\usepackage{hyperref}
\newcommand{\abs}[1]{\lvert #1 \rvert}

\newtheorem{theorem}{Theorem}
\newtheorem{definition}[]{Definition}
\newtheorem{corollary}[]{Corollary}
\newtheorem{example}{Example}
\newtheorem{remark}{Remark}
\newtheorem{proposition}{Proposition}
\newtheorem{observation}{Observation}
\newtheorem{lemma}{Lemma}
\newcommand{\problem}{\textsc{cds-clearing}}
\allowdisplaybreaks
\usetikzlibrary{positioning}
\usetikzlibrary{automata} 
\usetikzlibrary[automata]
\usetikzlibrary{snakes}

\title{Strong Approximations and Irrationality in Financial Networks with Financial Derivatives}

\author{Author names omitted for review}
\author{Stavros D. Ioannidis \and Bart de Keijzer \and  Carmine Ventre}
\date{Department of Informatics\\
King's College London}

\begin{document}\maketitle
\begin{abstract}
Financial networks model a set of financial institutions (firms) interconnected by obligations. Recent work has introduced to this model a class of obligations called credit default swaps, a certain kind of financial derivatives. The main computational challenge for such systems is known as the clearing problem, which is to determine which firms are in default and to compute their exposure to systemic risk, technically known as their recovery rates. It is known that the recovery rates form the set of fixed points of a simple function, and that these fixed points can be irrational. Furthermore, Schuldenzucker et al. (2016) have shown that finding a weakly (or ``almost") approximate (rational) fixed point is $\textsf{PPAD}$-complete.

We further study the clearing problem from the point of view of irrationality and approximation strength. Firstly, we observe that weakly approximate solutions may misrepresent 
the actual financial state of an institution. 
On this basis, we study the complexity of finding a strongly (or ``near") approximate solution, and show $\textsf{FIXP}$-completeness. 
We then study the structural properties required for irrationality, and we give necessary conditions for irrational solutions to emerge: The presence of certain types of cycles in a financial network forces the recovery rates to take the form of roots of non-linear 
polynomials. In the absence of a large subclass of such cycles, we study the complexity of finding an exact fixed point, which we show to be a problem close to, albeit outside of, \textsf{PPAD}.
\end{abstract}

\section{Introduction}
The International Monetary Fund says that the global financial crisis (GFC) of 2007 has had long lasting consequences, including loss of growth, large public debt and even a decline of fertility rates, see \cite{IMFBlogs}. Consequently, the need to assess the \emph{systemic risk} of the financial network cannot be overstated. For example, if banks at risk of defaults could be easily identified in the complex network of financial obligations, then spread could be preemptively avoided with appropriate countermeasures such as bailouts from central banks or regulators.

In this context, the \emph{clearing problem} introduced in \cite{eisenberg2001systemic} plays a central role. We are given a  so-called financial network, that is, a graph where vertices are banks (or, more generally, financial institutions) and weighted arcs $(u,v)$ model direct liabilities from bank $u$ to bank $v$. Each bank has also some assets external to the network, that can be used to pay its liabilities. The question is to compute a \emph{clearing recovery rate vector}, that is, the ratio between money available (coming from assets and payments from others) over liabilities for each bank. If this ratio is bigger than $1$ for a bank, then it will be able pay its dues -- in this case, we simply set its rate to $1$. The banks that are in default have recovery rates smaller than $1$. 
The problem of computing clearing recovery rates (which we will also refer to as the \emph{clearing problem}) is well understood in simple settings: When there are only simple debt contracts in the network, then clearing recovery rate vectors always exist, are unique, and can be computed in polynomial time \cite{eisenberg2001systemic}.

However, Eisenberg and Noe's model in \cite{eisenberg2001systemic} ignores the issue of  financial derivatives that may be present in the system. The deregulation allowing banks to invest in these products is considered by many as one of the triggers of the GFC. The introduction of financial derivatives to financial networks is due to \cite{schuldenzucker2016clearing}, where the focus is on a simple and yet widely used class of conditional obligations known as Credit Default Swaps (CDSes), the idea being to ``swap'' or offset a bank's credit risk with that of another institution. More specifically, a CDS has three entities: a creditor $v$, a debtor $u$ and a reference bank $z$ --- $u$ agrees to pay $v$ a certain amount  whenever $z$ defaults. Whilst CDSes were conceived in the early 1990s as a way to protect $v$ from the insolvency of $z$ for direct liabilities (i.e., a $(v,z)$-arc in the network), they quickly became a speculative tool to bet against the creditworthiness of the reference entity and have in fact been widely used both as a hedging strategy against the infamous collateralised debt obligations, whose collapse contributed to the GFC, and pure speculation during the subsequent eurozone crisis. The clearing problem in the presence of these financial derivatives is somewhat less well-characterised: it is known that the clearing recovery rate vectors can be expressed as the fixed points of a certain function, and existence of solutions is then guaranteed via a fixed-point argument \cite{schuldenzucker2016clearing}. On the other hand, these fixed points can be irrational, and the computational problem is $\textsf{PPAD}$-complete \cite{schuldenzucker2017finding} as long as one is interested in only a \emph{weak} approximation of a recovery rate vector. 

\subsection{Our Contributions}
In this paper we deepen the study of the clearing problem for financial networks with CDSes from two complementary viewpoints. Firstly, we argue that weak approximations can be misleading in this domain, as the objective under the weak approximation criterion is to find an ``almost'' fixed point (i.e., a point which is not too far removed from its image under the function). The risk estimate provided by this concept might be very far off the actual rate, thus changing the amount of bailout needed or even whether a bank needs rescue in the first place (see, e.g., our example in Appendix \ref{apx:A}). A more useful (but more difficult) objective is to obtain a strong approximation, that is, a point that is geometrically close to an actual fixed point of the function. Such a risk estimate would be actionable for a regulator, as the error could be measured in terms of irrelevant decimal places. Furthermore, the banks themselves would accept the rate when the strong approximation guarantee is negligible, whereas a weak approximation could significantly misrepresent their income and are subject to be challenged, legally or otherwise. 

As our first contribution, we settle the computational complexity of computing strong approximations to the clearing problem in terms of $\textsf{FIXP}$ \cite{etessami2010complexity}, by showing that the clearing problem is complete for this class.
In our reduction, we provide a series of financial network gadgets that are able to compute opportune arithmetic operations over recovery rates. Interestingly, not that many $\textsf{FIXP}$-complete problems are known, although there are a few important natural such problems (three-or-more-player Nash equilibria being a notable example). Hardness reductions for this class tend to be rather technically involved and not straightforward (see Section \ref{sec:related} for an overview of related work). The hardness reduction that we provide here indeed has some technical obstacles as well, although our reduction is quite natural at a high level, and could inspire further developments in the area. Our result complements the current state of the art and completes the picture about the computational complexity of the clearing problem with financial derivatives. It shows that computing strongly approximate fixed points is harder than computing weakly approximate fixed points, which holds due to $\textsf{PPAD}$ being equal to the class $\textsf{Linear-FIXP}$, which is a restriction of $\textsf{FIXP}$, and this makes $\textsf{PPAD}$ (indirectly) a subclass of $\textsf{FIXP}$.
\begin{quote}
    \textbf{Main Theorem 1 (informal).} \emph{Computing a strong approximation to the   clearing recovery rates in a financial network with CDSes is $\textsf{FIXP}$-complete.} 
\end{quote}
The $\textsf{FIXP}$-hardness of the strong approximation problem indicates that there is an additional numerical aspect contributing to the hardness of the problem, which is not present in the weak approximation problem (where the hardness is of a combinatorial nature, due to the reducibility to the end-of-the-line problem which is canonical to \textsf{PPAD}). For the strong approximation problem, the nature of the underlying function for which we want to find the fixed points requires, in particular, the multiplication operation, which ultimately accounts for irrationality and super-polynomial numerical precision 
being necessary in order to derive whether a given point is a strong approximation to a clearing vector. 


We then turn our attention towards irrational solutions with the goal to determine the source of irrationality and understand when it is possible to compute the clearing recovery rate vector exactly in the form of rational numbers. We identify a structural property of cycles in an opportunely enriched network that leads to unique irrational solutions. This property exactly differentiates the CDSes that produce and propagate irrationality of the recovery rates, that we call ``switched on'', from those that do not, termed ``switched off''. We prove the following close-to-tight characterisation of irrationality:
\begin{quote}
    \textbf{Main Theorem 2 (informal).} \emph{If the financial network has only ``switched on'' CDSes in a cycle \emph{and} the cycle cannot be shortcut with paths of length at most three\footnote{The length-at-most-three condition is restated in the form of a more refined condition in the respective technical sections that lead to this result.} then there exist rational values for debt and asset values for which the recovery rate vector is unique and irrational. Conversely, if every cycle of the financial network does not have any  ``switched on'' CDSes then we can compute rational recovery rates in a polynomial number of operations, provided that we have oracle access to $\textsf{PPAD}$.} 
\end{quote}
Our proof of irrationality uses a type of graph ``algebra'' (that is, a set of network fragments and an operation on them) that is able to generate all the possible cycle structures with the property above, which uncovers a connection between the network structure of the clearing problem and the roots of non-linear equations. 
For the opposite direction, we provide an algorithm that exploits the acyclic structure of financial networks with solely ``switched off'' CDSes. This algorithm iteratively computes the recovery rates of each strongly connected component of the network. We show that even for the simpler topologies of the financial system under consideration, the problem remains \textsf{PPAD}-hard, hence the need for the oracle access to \textsf{PPAD}.



\smallskip \noindent
\textbf{Significance of Our Results.} We see our complexity and irrationality results as important analytical tools that legislators can use to regulate financial derivatives. For example, our results contribute to the ongoing debate in the US and Europe about whether speculative uses of CDSes should be banned. In particular, they support, from a computational point of view, the call to ban so-called ``naked'' CDSes (as already done by the EU for sovereign debt  in the wake of the Eurozone crisis, see \cite{EUBusiness}). A naked CDS is purely speculative since its creditor and debtor have no direct liabilities with the reference entity. It turns out that these CDSes add arcs between potentially unconnected nodes, thus possibly adding more of the cycles that lead to irrationality and, given that strong approximations are out of scope due to our $\textsf{FIXP}$-hardness, it is not only combinatorially but also numerically intractable to gain insight in the systemic risk of such financial networks. A mechanism to monitor the topology of a financial network might be useful to avoid the construction of cyclical structures that include CDSes.

\smallskip \noindent \textbf{Technical and Conceptual Innovations.} Both of our main results introduce significant novel technical and conceptual innovations to the field. 

As mentioned above, our reduction for the $\textsf{FIXP}$-hardness is somewhat more direct than in previous work we are aware of. Our reduction is direct, in the sense that it starts from the algebraic circuit defined by an arbitrary problem in $\textsf{FIXP}$. The reduction employs two main steps: We firstly force the outputs of all gates in the circuit to be in the unit cube, by essentially borrowing arguments from \cite{etessami2010complexity}, after which we  produce a series of network gadgets that preserve gate-wise the computations of said circuit; this makes the reduction conceptually straightforward in its setup.  

It is worth highlighting a specific technical challenge that we overcome in our proof, as we think it sheds further light on $\textsf{FIXP}$, and in particular, on the operator basis of the algebraic circuits that are used to define the class. It is known that the circuit of problems in $\textsf{FIXP}$ can be restricted without loss of generality to be built on the arithmetic basis $\{\max, + , *\}$ \cite{etessami2010complexity}, whereas restricting the internal signals of the circuit to the unit cube (with the toolkit developed in \cite{etessami2010complexity}) needs some further operators, including $/$. 
For our optimisation problem to be in $\textsf{FIXP}$, we need the rather mild and realistic assumption that our instances are \emph{non-degenerate} as defined in  \cite{schuldenzucker2017finding}. 
The function of which the fixed points define the recovery rates of non-degenerate instances is well defined, where the non-degeneracy is needed to avoid a division by $0$. It turns out that non-degeneracy is incompatible with division being part of the \textsf{FIXP} operator basis, i.e., it seems difficult to build such a financial network that in any sense simulates a division of two signals in an algebraic circuit. To bypass this problem, our proof shows that it is possible to substitute $/$ in the basis with the square root operator, $\sqrt{\cdot}$, whilst keeping the function well defined. This substitution can be used to simulate division with constant large powers of $2$, and this turns out to be sufficient to omit the $/$-operator (i.e., arbitrary division). This novel observation might be useful for other problems where division is problematic to either define the fixed point function, or the reduction.  

Our second result indirectly aims at characterising the ``rational fragment'' of $\textsf{FIXP}$: To the best of our knowledge this is the first study in this direction. A couple of observations can be drawn from our attempt. Firstly, our sufficiency conditions for irrationality suggest that any such characterisation needs to fully capture the connection between the fixed point condition and the rational root theorem; our proof currently exploits the cyclical structure of networks with ``switched on'' CDSes to define one particular quadratic equation with irrational roots. Whilst this captures a large class of instances, more work is needed to give a complete characterisation (see Section \ref{sec:conclusions} for a discussion). Secondly, our sufficiency conditions for rational solutions highlight a potential issue with their representation. Due to the operations in the arithmetic basis, most notably multiplication, these solutions can grow exponentially large (even though each call to the \textsf{PPAD} oracle returns solutions of size polynomial in their input). This observation establishes a novel connection between the Blum-Shub-Smale computational model \cite{blum1998complexity} (wherein the size to store any real number is assumed to be unitary and standard arithmetic operations are executed in one time unit), the rational part of \textsf{FIXP}, and \textsf{PPAD}. Our result paves way to further research on the subject. 

\smallskip \noindent \textbf{Outline.} Our paper is organised as follows. Section \ref{sec:related} overviews work in areas of interest. We give our preliminaries in Section \ref{sec:preli} where we introduce our modeling of financial networks, a gentle introduction to \textsf{FIXP}, and begin discussing the irrationality of clearing recovery rates. Section \ref{sec:fixp} contains our proof of completeness for \textsf{FIXP}, whereas Sections \ref{sec:irrationality} and \ref{sec:rationality} contain the two directions of our second main result. We draw some conclusions in Section \ref{sec:conclusions}.

\section{Related Work}\label{sec:related}
Systemic risk and contagion in financial networks have been studied extensively in the literature \cite{acemoglu2015systemic,elliott2014financial,glasserman2015likely,heise2012derivatives,hemenway2016sensitivity,hu2012network,cifuentes2005liquidity}. Previous work models a financial network as a setup of interconnected nodes, representing economic firms, in an arc-weighted graph where arcs represent debt obligations from one firm toward another. Among the first papers on systemic risk in financial networks, Eisenberg and Noe \cite{eisenberg2001systemic}, study the problem of finding a clearing payment vector for a financial system that admits only direct liabilities. They prove, applying Tarski's fixed point theorem, that such payment vectors always exist and provide a polynomial time algorithm for computing one. A variation of the original model by Eisenberg and Noe is presented by Rogers and Veraart in \cite{rogers2013failure}. 

Financial systems admitting both direct liabilities and CDS contracts are introduced by Schuldenzucker et al. in \cite{schuldenzucker2016clearing}. In this paper, the authors are interested in computing the clearing recovery rate vector of the system. They establish that such a vector always exists when studying models in \cite{eisenberg2001systemic} whereas that's not he case for models in \cite{rogers2013failure} 
and deciding whether one does 
is \textsf{NP}-hard 
\cite{DBLP:books/daglib/0072413}. Subsequently, Schuldenzucker et al. in \cite{schuldenzucker2017finding} study the problem of computing a clearing recovery rate vector in models where existence is guaranteed. Early in their paper, they construct a simple instance of a financial network whose solution is proved to be irrational, thus making the problem of computing a clearing recovery rate vector an approximation problem. Their main result is that almost-approximating the clearing recovery rate vector is \textsf{PPAD}-complete, meaning that no PTAS 
exists unless $\textsf{P}=\textsf{PPAD}$.

In absence of approximation algorithms for computing an exact solution, much work has been done on the computational complexity of deciding the proper clearing recovery rate that satisfies specific objectives for the system. Papp and Wattenhofer in \cite{papp2020default}, establish that if a regulator can characterize some solution, it is still \textsf{NP}-hard to decide the clearing recovery rate that minimizes the number of defaulting banks and the unpaid debt in the system. Moreover they prove that choosing  the solution that is most preferable by the largest set of banks, the solution that is preferred by a specific bank as well as the solution with the best equity distribution is also \textsf{NP}-hard to approximate within some constant factor.

Despite the existence of a clearing recovery rate vector being guaranteed, irrationality ensures the lack of efficient algorithms for computing one. Our present paper originates from this absence, leading us to consider approximation versions of the problem. It is known from \cite{schuldenzucker2016clearing}, that the problem of computing a clearing recovery rate vector can be cast as a fixed point total search problem. A complexity class regarding fixed point computations of total search problems, is the \textsf{FIXP} class defined by Etessami and Yannakakis in \cite{etessami2010complexity}, where the authors define and introduce the \textsf{FIXP} complexity class and prove that the problem of computing a 3-player (or more) Nash equilibrium is complete for this class. Not many \textsf{FIXP}-complete problems are known in the literature, but some example we are aware of can be found in \cite{goldberg2021hairy,filos2021complexity,DBLP:conf/stoc/GargMVY17,filos2021fixp}. 

Work has also been done on the incentives of banks in financial networks. The authors of \cite{BertschingerHS20} study price of anarchy and stability in games where banks can strategically decide how to repay debts when insolvent. The strategic aspects of modifying the structure of the network (by, e.g., writing off a debt) to a bank's advantage are considered in \cite{PappW20a}.

\section{Model and Preliminaries}\label{sec:preli}

\subsection{Financial Systems}
We denote by $N = \{1,\ldots,n\}$ a set of $n$ financial institutions, which we will call \emph{banks}, for convenience. Each bank possesses a-priori some amount of  \emph{external assets}, denoted by $e_i \in \mathbb{Q}_{\geq 0}$ for a bank $i \in N$. We let $e = (e_1, \ldots ,e_n)$ be the vector of external assets.
Each bank $i$ has a certain \emph{debt} (or \emph{liability}) to each other bank in $N$. We consider two types of debts bank $i$ can have toward other banks:  \emph{debt contracts} and \emph{credit default swaps (CDSes)}.

A \emph{debt contract} is a contract between two banks that requires one of the banks, named the \emph{debtor} (or \emph{writer}), to pay a certain amount to the other bank, named the \emph{creditor} (or \emph{holder}). The value that needs to be paid under a debt contract from debtor $i \in N$ towards creditor $j \in N$ is denoted by $c_{i,j}^\emptyset \in \mathbb{Q}_{\geq 0}$. For convenience, we sometimes write $c_{i,j}$ instead of   $c_{i,j}^\emptyset$.
We denote by $\mathcal{DC}$ the set of all pairs of banks participating in a debt contract, That is, if $(i,j) \in \mathcal{DC}$, then there exists a debt contract where $i$ is the debtor and $j$ is the creditor.

A \emph{Credit Default Swap (CDS)} is a contract that involves three banks. In a CDS, a debtor owes money to a creditor, similarly to a debt contract. However, the amount of money owed is dependent on whether a third bank called the \emph{reference bank} is in default, and is directly proportional to what extent it is in default. A bank is in default when it has insufficient assets to pay its total amount of debts in full. More formally, bank $i$'s recovery rate $r_i \in [0,1]$ is the fraction of $i$'s debts that the bank is able to pay off, given its total assets (defined as the sum of $e_i$ and the payments it receives from other banks). Bank $i$ is said to be in default iff $r_i < 1$. In case a reference bank $R \in N$ of a certain CDS is in default, the debtor $i \in N$ of that CDS is obliged to pay the creditor $j \in N$ an amount of $(1-r_R)c_{i,j}^R$, where $c_{i,j}^R \in \mathbb{Q}_{\geq 0}$ is a prespecified amount associated to the particular CDS between banks $i,j$, with reference $R$. We denote by $\mathcal{CDS}$ the set of all triplets of banks participating in a CDS contract. That is, if $(i,j,R) \in \mathcal{CDS}$, then there exists a CDS contract where $i$ is the debtor bank, $j$ is the creditor, and $R$ is the reference bank.

The value $c_{i,j}$ (or $c_{i,j}^R$) of a debt contract (or CDS) is also referred to as the \emph{notional} of the contract.

Note that the non-existence of a debt contract or CDS between a pair or triplet of banks is equivalent to having a contract present where the corresponding notional is $0$, and that multiple contracts between identical pairs or triplets of banks can be merged into a single contract with the sum of the individual contract's notionals.
Therefore, we may assume that we are given a three-dimensional ($n \times n \times n$) matrix $c$ that contains all contract notionals for direct liabilities and CDS contracts. We do not allow any bank to have a debt contract with itself, and assume that all three banks in any CDS are distinct, which we represent in $c$ by setting $c_{i,i} = 0$ for all $i \in N$ and setting $c_{i,j}^k = 0$ whenever $i,j$ and $k$ are not all distinct. This brings us to the definition of a financial system.

\begin{definition}
A financial system is a triplet $(N,e,c)$, where $N = \{1,..,n\}$ is a set of banks, $e = (e_1,..,e_n) \in \mathbb{Q}^n_{\geq 0}$ is the vector of external assets, and $c \in \mathbb{Q}^{n \times n \times n}_{\geq 0}$ is a three-dimensional matrix of all contract notionals. 
\end{definition}

It is often reasonable to assume that the set of non-zero debt contracts and non-zero CDSes in a given financial system is rather sparse, in which case it is natural to view a given financial system $I = (N,e,c)$ as a coloured directed graph-like structure, which we call the \emph{contract graph} as follows: We first construct a directed multigraph $G = (V,A)$, where $V = N$ and $A$ is the multiset-union of the set $A_{0} = \{(i,j) \mid c_{i,j} \neq 0\}$ and the sets $A_k = \{(i,j) \mid c_{i,j}^k \neq 0\}$ for all $k \in N$. Furthermore, we colour each arc through a function $t : E \rightarrow \{\text{blue},\text{orange}\}$, where $t(e) = \text{blue}$ iff $e \in A_0$ and $t(e) = \text{orange}$ otherwise. Instead of writing $t(e)$ we sometimes write $t_{i,j}$ for convenience, where $e = (i,j)$. For all $(i,j,R) \in \mathcal{CDS}$ we draw a dotted orange line from node $R$ to arc $(i,j) \in A_R$, to denote that $R$ is the reference bank of the CDS between $i$ and $j$. Note that we are using the terms \emph{bank} and \emph{node} interchangeably.

In the resulting graphical notation, we label an arc with the notional of the corresponding contract, and we will label a node with the external assets of the corresponding bank, in green font. This notation is illustrated in the following example, depicted in Figure \ref{fig:1}.

\begin{example}\label{ex:1}
Let $(N,e,c)$ be a financial system where $N = \{1,2,3,4,5,6\}$ is the set of banks and $e = (1,0,0,1,0,1)$ the external assets vector (for example, the external assets of bank 1 are $e_1 = 1$ and for bank 3 we have $e_3 = 0$). 

The financial network is shown in Figure \ref{fig:1}. Node 1 is the debtor of two debt contracts, one towards node 2 and one towards node 3, represented by blue arcs. The notionals of these contracts are $c_{1,2} = 1$, $c_{1,3} = \frac{1}{2}$. Node 3 is the debtor in one debt contract towards node 5, where $c_{3,5} = \frac{1}{2}$. Node 4 is only a debtor towards node 6, where $c_{4,6} = \frac{1}{2}$.

There are two CDSes in the system, one from node $2$ to node $4$ in reference to node $3$ and one from node $5$ to node $6$ in reference to node $4$. Formally $\mathcal{CDS} = \{(2,4,3),(5,6,4)\}$. For $(5,6,4)$ the notional is $c_{5,6}^4 = 1$ and for $(2,4,3)$ the notional is $c_{2,4}^3 = \frac{2}{3}$. 

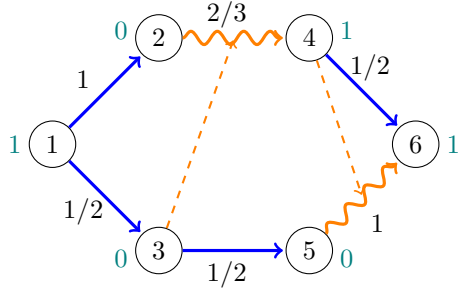
\begin{figure}[htbp!]
\centering
\begin{tikzpicture}
[shorten >=1pt,node distance=2cm,initial text=]
\tikzstyle{every state}=[draw=black!50,very thick]
\tikzset{every state/.style={minimum size=0pt}}
\node[state] (1) {$1$}; 
\node[state] (2) [above right of=1] {$2$};
\node[state] (3) [below right of=1] {$3$};
\draw[blue,very thick,->] (1)--node[midway,black,xshift = -3mm,yshift = -1mm,above] {$1$}(2);
\draw[blue,very thick,->] (1)--node[midway,black,xshift = -3mm,yshift = -5mm,above]{$1/2$}(3);
\node[state] (4) [right of=2] {$4$};
\node[state] (5) [right  of=3] {$5$};
\node[state] (6) [below right of= 4] {$6$};
\node[teal,left of=1,xshift=1.5cm]{1};
\node[teal,left of=3,xshift=1.5cm,yshift=-1mm]{0};
\node[teal,left of=2,xshift=1.5cm,yshift=1mm]{0};
\node[teal,right of=4,xshift=-1.5cm,yshift=1mm]{1};
\node[teal,right of=5,xshift=-1.5cm,yshift=-1mm]{0};
\node[teal,right of=6,xshift=-1.5cm]{1};
\draw[orange,very thick,->,snake=snake] (2)--node[midway,black,xshift = -1mm,above]{$2/3$}(4);
\draw[blue,very thick,->] (3)--node[midway,black,xshift = -1mm,below]{$1/2$}(5);
\draw[blue,very thick,->] (4)--node[midway,black,xshift = 1mm,above]{$1/2$}(6);
\draw[orange,very thick,->,snake=snake] (5)--node[midway,black,xshift = 2mm,yshift = -6mm ,above]{$1$}(6);
\path [orange,->,draw,dashed,thick] (4) -- ($ (5) !.5! (6) $);
\path [orange,->,draw,dashed,thick] (3) -- ($ (2) !.5! (4) $);
\end{tikzpicture}
\caption{A financial system with six banks and two CDSes. This illustrates the graphical notation we use to visualise a financial system.}
    \label{fig:1}
\end{figure}
\end{example}

Our study focuses on the task of computing for any given financial system, for each bank, the proportion of its liabilities that it is able to pay. This proportion is captured by the \emph{recovery rate} of a bank, introduced earlier in our definition of CDS contracts. For each bank $i \in N$ we associate a variable $r_i \in [0,1]$, that indicates the proportion of liabilities that bank $i$ can pay, where $r_i = 1$, indicates that bank $i$ can fully pay its liabilities. Having $r_i < 1$ indicates that $i$ is in default. When a bank $i$ is in default, it pays off its debts proportionally to the liabilities stemming from each of the contracts in which $i$ is the debtor. The focus of this paper is on the problem of computing a clearing recovery rate vector $r = (r_1, \cdots, r_n)$. 

Recall that the amount a debtor $i$ has to pay to creditor $j$ in a CDS contract with reference bank $R$ is given by $(1-r_R)c_{i,j}^R$. Consequently, to compute a bank's liabilities, knowledge is required of the recovery rates of other banks (in particular in case $i$ is involved in one or more CDSes as the debtor). That is, to compute a bank's recovery rate, we must be presented with other banks' recovery rates a-priori.  

Given a candidate vector $r \in [0,1]^n$ of recovery rates, we can formally define a bank's liabilities, payments, and assets in a financial system as follows.
\begin{description}
\item[Liabilities.] The liability of a bank $i \in N$ to a bank $j \in N$ under $r$ is denoted by 
\begin{equation*}
l_{i,j}(r) = c_{i,j}^\emptyset + \sum_{k \in N}(1-r_k)c_{i,j}^k.
\end{equation*}
That is we sum up the liabilities from the debt contract and all CDS contracts between $i$ and $j$. We denote by $l_i(r)$ the total liabilities of $i$:
\begin{equation}\label{eq:liabilities} 
l_i(r) = \sum_{j \in N}l_{i,j}(r) 
\end{equation}
\item[Payments.] The payment bank $i$ makes to bank $j$ under $r$ is denoted by $p_{i,j}(r)$ and it holds that:
\begin{equation*}
    p_{i,j}(r) = r_i \cdot l_{i,j}(r)
\end{equation*}
\item[Assets.] The assets of a bank $i$ under $r$ are the total amount it possesses through its external assets and incoming payments made all by other banks. 
\begin{equation}\label{eq:assets}
    a_i(r) = e_i + \sum_{j \in N}p_{j,i}(r)
\end{equation}
\end{description}

The above definitions for liabilities, payments and assets require a vector recovery rates to be given. While the definitions are valid with respect to any arbitrary vector $r$ in the intervals $[0,1]^n$, the interpretation we give to the notion of a recovery rate is that it reflects the fraction of a bank's liabilities that it is able to pay using its assets. That is, for a given financial system, we are interested in the specific recovery rate vectors such that $r_i$ is $1$ if $i$'s assets exceed its liabilities, and otherwise is equal to the ratio of its assets and liabilities. Such vectors of recovery rates are called \emph{clearing vectors}.

\begin{definition}[Clearing recovery rate vector] Given a financial system $(N,e,c)$, a recovery rate vector $r$ is called clearing if and only if for all banks $i \in N$, 
\begin{equation}\label{eq:clearing}
    r_i = \min \left\{1,\frac{a_i(r)}{l_i(r)}\right\},
\end{equation}
 if $l_i(r) > 0$, and $r_i = 1$ if $l_i(r) = 0$. 
\end{definition}

The requirement that a recovery rate vector be clearing forces an interdependence between the assets, liabilities, and clearing recovery rates of the banks in a financial system. This can be seen, in particular, in the case of cycles that contain a mixture of debt contracts and CDSes, in which case the recovery rate of each bank in the cycle is dependent on the recovery rates of all other banks in the cycle. This complex interrelationship among recovery rates makes computation of a clearing vector a non-trivial computational problem: In first instance, it is not even clear whether a vector of recovery rates always exists, or whether there can be multiple clearing vectors. Both these questions have been answered in the affirmative in \cite{schuldenzucker2016clearing}.


\addtocounter{example}{-1}
\begin{example}[continued]
We compute the clearing recovery rate vector for the example depicted in Figure 1. For node $1$ it holds that $a_1 = e_1 = 1$ and $l_1 = c_{1,2}^\emptyset + c_{1,3}^\emptyset = 1 + 1/2 = 3/2$, thus from (\ref{eq:clearing}) we get that $r_1 = \min\{1,2/3\} = 2/3$. For node $2$ we have that $a_2 = p_{1,2} + e_2 = r_1l_{1,2} + e_2 = 2/3 \cdot 1 + 0 = 2/3$. Also $l_2 = c_{2,4}^3(1-r_3) = (2/3)(1-r_3)$. So to compute $r_2$ we first need to compute $r_3$. For node 3 it holds that $a_3 = p_{1,3} + e_3 = r_1l_{1,3} + e_3 = (2/3)(1/2) + 0 = 1/3$, and $l_3 = c_{3,5} = 1/2$, finally $r_3 = \min\{1,a_3/l_3\} = \min\{1,2/3\} = 2/3$. Since node 3 is in default, the $(2,4,3)$ CDS is activated and counts towards the liabilities of node 2, so that $l_2 = (2/3)(1-2/3) = 2/9$. Finally $r_2 = \min\{1,a_2/l_2\} = \min\{1,3\} = 1$, i.e., node 2 can fully pay its liabilities. For node 4 it holds that $a_4 = p_{2,4} + e_4 = r_2c_{2,4}^3(1-r_3) + 1 = 1 \cdot 2/3 \cdot (1-2/3) + 1 = 11/9$, $l_4 = 1/2$ and $r_4 = \min\{1,a_4/l_4\} = 1$. Node 4 is not in default, so the CDS $(5,6,4)$ is not activated and node 5 has no liability towards node 6.
\end{example}

We refer to the computational task of finding a clearing vector in a given financial system as \problem, where ``\textsc{cds}'' signifies that the financial system under consideration may contain credit default swaps.\footnote{The version of the problem where there are only debt contracts and no CDSes, has been analysed in \cite{eisenberg2001systemic} and turns out to be computationally easy, see Section \ref{sec:related}.}
We refer to any clearing recovery rate vector of an instance $I \in \problem$ as a \emph{solution} of $I$. For an instance $I \in \problem$, we denote the set of solutions of $I$ by $\text{Sol}(I)$.

It is important to realise that the solutions of \problem\ are essentially the fixed points of the function expressed at the right hand side of (\ref{eq:clearing}): Let $I \in \problem$ and 
consider the function $f_I:[0,1]^n \mapsto [0,1]^n$ defined at each coordinate $i \in [n]$ by
\begin{equation}\label{eq:f}
f_I(r)_i= \frac{a_i(r)}{\max\{l_i(r),a_i(r)\}} .
\end{equation}
Function $f=f_I$ takes as input a recovery rate vector $r$ and outputs a potentially different recovery rate vector $r' = f(r)$. It holds that $r' = r$ if and only if $r$ is clearing and hence a solution of $I$. Thus, \problem\ asks, in essence, for computing a fixed point of a specific function. The fact that there exists at least one fixed point for every instance of \problem\ is not immediate, and was proved in \cite{schuldenzucker2016clearing} through an application of Kakutani's fixed point theorem.\footnote{There is a small difference between the way we define the recovery rates and the definition in \cite{schuldenzucker2016clearing}: We define the recovery rate of nodes without liabilities to be equal to 1, while \cite{schuldenzucker2016clearing} leaves the recovery rates of such nodes unconstrained. We make this assumption solely to be in line with the interpretation that a recovery rate of 1 would indicate that a bank is not in default, and it is in essence only a cosmetic difference. It does not introduce any problems with respect to the application of Kakutani's fixed point theorem in our case: We may simply exclude recovery rates of nodes without liabilities as variables from the function on which we apply Kakutani's theorem. The excluded variables correspond to nodes that are not reference banks in any CDS, nor debtors in any CDS or debt contract (where we assume that the non-degeneracy condition holds, see Definition \ref{def:nondegenerate}).}


\begin{theorem}[\cite{schuldenzucker2016clearing}]Every financial system $(N,e,c)$ has a clearing recovery rate vector.
\end{theorem}


\subsection{Irrationality}
Solving \problem\ comes down to computing the fixed points of a specific real valued function. Unfortunately, an obstacle is immediately encountered if one intends to solve \problem\ under traditionally accepted models of computation, since there exist instances of \problem\ in which all clearing vectors have irrational values, and can thus not be encoded conveniently in binary. Below we present such an instance:
\begin{example}
Consider the following instance  $I = (N,e,c)$ with $N= \{1,\cdots,8\}$ and $e_2 = e_7 =1/2$, $e_i = 0$, for every other node. Assume all direct and CDS contract notionals are equal to 1. The graph of the financial system is given in Figure \ref{fig:2}.

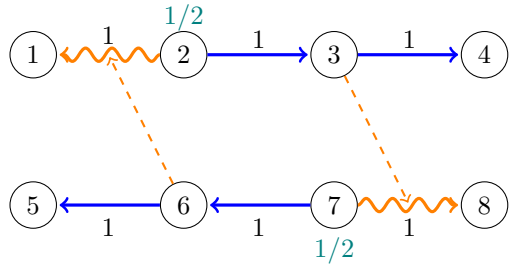
\begin{figure}[htbp!]
    \centering
\begin{tikzpicture}
[shorten >=1pt,node distance=2cm,initial text=]
\tikzstyle{every state}=[draw=black!50,very thick]
\tikzset{every state/.style={minimum size=0pt}}
\node[state] (1) {$1$}; 
\node[state] (2) [right of=1] {$2$};
\node[state] (3) [right of=2] {$3$};
\node[state] (4) [right of=3] {$4$};
\draw[orange,very thick,->,snake=snake] (2)--node[midway,black,yshift=2.5mm]{1}(1);
\draw[blue,very thick,->] (2)--node[midway,black,yshift=2mm]{1}(3);
\draw[blue,very thick,->] (3)--node[midway,black,yshift=2mm]{1}(4);
\node[state] (5) [below of=1] {$5$};
\node[state] (6) [below of=2] {$6$};
\node[state] (7) [below of=3] {$7$};
\node[state] (8) [below of=4] {$8$};
\draw[blue,very thick,->] (6)--node[midway,black,yshift=-3mm]{1}(5);
\draw[blue,very thick,->] (7)--node[midway,black,yshift=-3mm]{1}(6);
\draw[orange,very thick,->,snake=snake] (7)--node[midway,black,yshift=-3mm]{1}(8);
\node[teal,right of=2,xshift=-2cm,yshift=5mm]{1/2};
\node[teal,right of=7,xshift=-2cm,yshift=-6mm]{1/2};
\path [orange,->,draw,dashed,thick] (6) -- ($ (2) !.5! (1) $);
\path [orange,->,draw,dashed,thick] (3) -- ($ (7) !.5! (8) $);
\end{tikzpicture}
 \caption{An instance with irrational solution.}
    \label{fig:2}
\end{figure}

Any solution for this instance must satisfy the equations
\begin{equation*}
r_2 = \min\left\{1,\frac{1/2}{2-r_6}\right\}, \quad
r_6 = r_7, \quad
r_7 = \min\left\{1,\frac{1/2}{2-r_3}\right\}, \quad
r_3 = r_2.
\end{equation*}
If $r_2 = 1$ then $r_3 =1$ $\Rightarrow r_7 = 1/2 \Rightarrow r_6 = 1/2 \Rightarrow r_2 =1/3$, which is a contradiction. The same holds for $r_7$, which means that $r_2 = 1/(2(2-r_7))$ and $r_7 = 1/(2(2-r_2))$. Solving this system of equations comes down to solving $2r_2^2 - 4r_2 + 1 = 0$ which yields $r_2 = 1 - \sqrt2/2$, from which we obtain the recovery rate of all the remaining nodes, as  $r_2 = r_3 = r_7 =r_6$.
\end{example}

Instances with irrational solutions are in a sense ``common'', and certainly not difficult to find. Another example can be found in \cite{schuldenzucker2017finding}. 
Hence, while instances in \problem\ always have solutions, they do not generally have rational solutions, which makes computing the \emph{exact} solutions to instances of \problem\ inconvenient for study under a standard model of computation.\footnote{One could however, quite naturally study this problem under a real valued computation model like \cite{blum1998complexity}, under which any real number can be stored in a single unit of memory, and the basic arithmetic operation take a single time unit to execute.} Instead, it makes more sense to consider \emph{approximate} solutions for \problem\ instances.

\subsection{Approximation and Complexity}
In \cite{etessami2010complexity}, a useful framework for studying the complexity of fixed point computation problems is defined, which considers both exact computation and approximate computation of the solutions to such problems. The authors introduce the complexity class $\textsf{FIXP}$, show that many fixed point problems can be shown to belong to this class, and establish that there exist $\textsf{FIXP}$-complete problems under a suitable notion of reducibility. \problem\ is, as we have seen, a fixed point problem, hence we consider the problem within the framework of \cite{etessami2010complexity}. We introduce in this section some fundamental notions related to exact and approximate computation of fixed points, and define the complexity class $\textsf{FIXP}$ and some related classes.

Let $F$ be a continuous function that maps a compact convex set to itself (and note that $f_I$, for instances $I \in \problem$ is such a function), and let $\epsilon > 0$ be a small constant.

A \emph{weak $\epsilon$-approximate fixed point} of $F$ is a point $x$ such that its image is within a distance $\epsilon$ of $x$, i.e., $\|x - F(x)\|_{\infty} < \epsilon$. 
A \emph{strong $\epsilon$-approximate fixed point} of $F$ is a point $x$ that is within a distance $\epsilon$ near a fixed point of $F$, i.e., $\exists x' : F(x') = x' \wedge \|x'-x\|_{\infty} < \epsilon$.
A weak $\epsilon$-approximate is generally not located close to an actual fixed point, and an example of this in the case of \problem\ can be found in Appendix \ref{apx:A}. Moreover, under a mild condition on the fixed point problem under consideration, known as \emph{polynomial continuity}, a strong approximation is also a weak approximation (see Proposition \ref{prop:pc} and the definition below), which explains the use of the terms ``strong'' and ``weak''.

Formally a \emph{fixed point problem $\Pi$} is defined as a search problem such that for every instance $I \in \Pi$ there is an associated continuous function $F_I : D_I \rightarrow D_I$ where $D_I \subseteq \mathbb{R}^n$ (for some $n \in \mathbb{N}$) is compact and convex, such that the solutions of $I$ are the fixed points of $F_I$. The problem $\Pi$ is said to be \emph{polynomially computable} if there is a polynomial $q$ such that (i.) $D_I$ is a convex polytope described by a set of at most $q(|I|)$ linear inequalities, each with coefficients of a size at most $q(|I|)$, and (ii.) For each $x$ in $D_I \cap \mathbb{Q}^n$, the value $F_I(x)$ can be computed in time $q(|I| + \text{size}(x))$. Here, the ``size'' of a rational number means the number of bits needed to represent the numerator and the denominator in binary. Furthermore $\Pi$ is said to be \emph{polynomially continuous} if there is a polynomial $q$ such that for each $I \in \Pi$, and rational $\epsilon > 0$, there is a rational $\delta$ of size $q(|I| + \text{size}(\epsilon))$ satisfying the following: for all $x,y \in D_I$ with $\|x-y\|_{\infty} < \delta$ it holds that $\|F_I(x) - F_I(y)\|_{\infty} < \epsilon$.
Note that this condition defining polynomial continuity is the formally correct way of stating that the distance between any two points in the domain does not increase a lot (at most by a ``polynomial amount'') when taking their images under $F_I$.\footnote{The above collection of definitions is stated in a more extensive and refined way in \cite{etessami2010complexity} in order to capture a more versatile collection of search problems. However, for the sake of our discourse, this simplified presentation suffices.}

A fixed point problem $\Pi$ has an associated \emph{weak} and \emph{strong} approximation versions: In the weak approximation version we are given an instance $I$ of $\Pi$ and a rational $\epsilon > 0$, and we need to compute a weak $\epsilon$-approximate fixed point for $F_I$. The strong approximation version is defined analogously. 
\begin{proposition}[\cite{etessami2010complexity}, Proposition 2]
\label{prop:pc}
Let $\Pi$ be a fixed point problem.
\begin{enumerate}
    \item If $\Pi$ is polynomially continuous, then the weak approximation version of $\Pi$ polynomial-time reduces to the strong approximation version of $\Pi$.
    \item If $\Pi$ is polynomially continuous and polynomially computable, then the weak approximation version of $\Pi$ is in $\textsf{PPAD}$.
\footnote{The reader might be familiar with  $\textsf{PPAD}$, which is a complexity class introduced in \cite{papadimitriou}. Further below, we provide an indirect definition of $\textsf{PPAD}$ in Proposition \ref{prop:linearfixpppad}, through defining the complexity class $\textsf{Linear-FIXP}$.}
\end{enumerate}
\end{proposition}

With regard to \problem, it is straightforward to verify that \problem\ is polynomially computable. Furthermore, \cite{schuldenzucker2017finding} establishes implicitly that \problem\ is polynomially continuous under a (very) mild assumption that the authors call \emph{non-degeneracy}. 

\begin{definition}\label{def:nondegenerate}
A financial system is \emph{non-degenerate} iff
\begin{enumerate} 
\item Every debtor in a CDS either has positive external assets or is the debtor in at least one debt contract with a positive notional. 
\item Every bank that acts as a reference bank in some CDS is the debtor of at least one debt contract with a positive notional.
\end{enumerate}
\end{definition}

We will make this assumption as well throughout our paper, both for the sake of compatibility with \cite{schuldenzucker2017finding} and for the analytical convenience that non-degeneracy provides us (note that in particular, a division by 0 never occurs in (\ref{eq:f}) when considering a non-degenerate instance).\footnote{We note that point 2 of our definition of non-degeneracy was not included in the definition in \cite{schuldenzucker2017finding}. However, the assumption we denote by Point 2 is made throughout their paper, but was left unnamed. We chose to include this assumption in the non-degeneracy notion for convenience, rather than leaving it without a name.}
Hence, we define \problem\ to contain only non-degenerate financial networks.

In \cite{schuldenzucker2017finding}, the authors furthermore show that the weak approximation version of \problem\ is $\textsf{PPAD}$-hard, which yields $\textsf{PPAD}$-completeness of the weak approxmation version of \problem. The polynomial continuity of \problem\ shows that the strong approximation version of \problem\ is at least as hard as its weak approximation version.

Despite that computing weak approximations for \problem\ is already known to be $\textsf{PPAD}$-hard, and that computing strong approximations must therefore be at least as hard due to Point 1 of Proposition \ref{prop:pc}, we are still interested in understanding the complexity of computing such strong approximations due to the fact that, as noted above, weakly approximate clearing vectors are difficult to justify from the point of view of a regulator who has the responsibility of actually realising the clearing of a financial system: Weakly approximate clearing vectors are potentially very far removed from the actual exact clearing vectors of a financial system. Clearing a system by means of a weakly approximate vector may result in defaults of banks that would not be in default under the true clearing vectors of the system, and non-defaulting banks may similarly end up with vastly different surplus assets after clearing than they would have under a true clearing vector. Strong approximations do not suffer from this problem, as they guarantee that the correct set of banks default and that the equity of all banks after clearing is indeed accurate (of course, up to $\epsilon$ accuracy). We demonstrate and prove the claims of the above discussion through an illustrative example in Appendix \ref{apx:A}. Our example supports and motivates the research of strong approximations.

To study the strong approximation and exact versions of fixed point problems, the authors of \cite{etessami2010complexity} introduce the class $\textsf{FIXP}$, and a few derived classes, among which are $\textsf{Linear-FIXP}$ and $\textsf{FIXP}_a$.

\begin{definition}[\textsf{FIXP}]
The class $\textsf{FIXP}$ consists of all fixed point problems $\Pi$ that are polynomially computable, and for which for all $I \in \Pi$ the function $F_I : D_I \rightarrow D_I$ can be represented by an algebraic circuit $C_I$ over the basis $A = \{+,-,*,\max,\min\}$, using rational constants, such that $C_I$ computes $F_I$, and $C_I$ can be constructed from $I$ in time polynomial in $|I|$.
In more detail, the circuit $C_I$ is an acyclic directed graph over a set of gates, say $g_1,\ldots,g_m$ (in topological order), with the following properties. There are $n$ gates in the first (bottom) layer of the circuit, called the \emph{input} gates, where $n$ is the number of arguments of $F_I$. Similarly there are $n$ output gates in the last (top) layer of the circuit, called the \emph{output gates}. The remaining nodes of the circuit each represent rational constants and arithmetic operations from the set $A$. Rational constant nodes have no incoming arcs, and each arithmetic operation node has two incoming arcs. The number of outgoing arcs from each node is not bounded. We may now label $C_I$'s input nodes with an input $x \in D_I$ (coordinate-wise), and define the output of $C_I$ by propagating the outputs of each of the nodes through the network via the arcs outgoing from each of the nodes. Naturally, an input node propagates its labeled value, a rational constant node propagates the associated rational number, and an arithmetic operation node takes the two values propagated through its incoming arcs as operands, and applies the corresponding operation on its operands, which it then propagates through its outgoing arcs. The circuit $C_I$ is then said to compute $F_I$ iff the output of $C_I$ on input $x$ equals $F_I(x)$ for all $x \in D_I$.

The class $\textsf{FIXP}_a$ is defined as the class of search problems that are the strong approximation version of some fixed point problem that belongs to $\textsf{FIXP}$.

The class $\textsf{Linear-FIXP}$ is defined analogously to $\textsf{FIXP}$, but under the smaller arithmetic basis where only the gates $\{+,-\max,\min\}$ and multiplication by rational constants are used (thus, arbitrary multiplication 
is removed from the basis).
\end{definition}

The classes $\textsf{FIXP}$, $\textsf{Linear-FIXP}$, and $\textsf{FIXP}_a$ admit complete problems. Hardness of a search problem $\Pi$ for $\textsf{FIXP}$ (resp. $\textsf{Linear-FIXP}$ and $\textsf{FIXP}_a$) is defined through the existence of a polynomial time computable function $\rho : \Pi' \rightarrow \Pi$, for all $\Pi \in \textsf{FIXP}$ (resp. $\textsf{FIXP}_a$), such that the solutions of $I$ can be obtained from the solutions of $\rho(I)$ by applying a (polynomial-time computable) linear transformation on a subset of $\rho(I)$'s coordinates. This type of reduction is known as a \emph{polynomial time SL-reduction}.

We mention a few important facts about these complexity classes that are relevant for understanding the results in our present work. The following proposition states that strong approximations to problems in $\textsf{FIXP}$ can be found in polynomial space.
\begin{proposition}[\cite{etessami2010complexity}, Proposition 17]\label{prop:pspace}
$\textsf{FIXP}_a \subseteq \textsf{PSPACE}$.
\end{proposition}
The following result shows that finding exact fixed points of fixed point functions expressible through only the operations $\{+,-,\max,\min\}$ and multiplication by rational constants, is equivalent to solving problems in $\textsf{PPAD}$.
\begin{theorem}[\cite{etessami2010complexity}, Theorem 26]\label{prop:linearfixpppad}
$\textsf{Linear-FIXP} = \textsf{PPAD}$.
\end{theorem}
Consequently, the solutions of instances in $\textsf{Linear-FIXP}$ are always rationals of polynomial size.

\begin{remark}\label{rem:1}
One detail that we have omitted so far in our definition of $\textsf{FIXP}$ (for simplicity) is that formally $\textsf{FIXP}$ and $\textsf{Linear-FIXP}$ are defined to also include their closures under polynomial time separable linear reductions and polynomial time reductions respectively.\footnote{The term ``\emph{separable linear}'' refers to the property that solutions of the reduced instance correspond to solutions of the original instance through a polynomial-time computable \emph{linear transformation} of a subset of the coordinates of the reduced instance's solution. The authors of \cite{etessami2010complexity} require separable linearity for \textsf{FIXP}-reductions, but not for \textsf{Linear-FIXP} reductions. This is due to the fact that the fixed points of functions belonging to \textsf{FIXP} may be irrational, and linearly separable reductions preserve the approximation properties of approximate fixed points. For \textsf{Linear-FIXP}, the fixed points are always rational, and one is generally interested in exact computation for this class. Thus, the traditional notion of polynomial-time reducibility is appropriate here.}. We remark furthermore that in the initial definition of $\textsf{FIXP}$, which was given in \cite{etessami2010complexity}, the division operator and the operators $\sqrt[k]{\cdot}$ for $k \in \mathbb{N}$ were included in the arithmetic base. However, in the same paper the authors show that these can be dropped without altering the class. Similarly, the original definition of $\textsf{FIXP}$ includes problems in which $D_I$ is an arbitrary polytope with a polynomially computable set of constraints, but restricting this to $[0,1]^n$ turns out to not change the class. 
\end{remark}

An informal understanding of how the hardness of $\textsf{FIXP}$ compares to $\textsf{PPAD}$ (or $\textsf{Linear-FIXP}$) is as follows. $\textsf{PPAD}$ captures a type of computational hardness stemming from an essentially combinatorial source, as $\textsf{PPAD}$ was originally defined as a class that captures the hardness computing a solution to a graph-theoretical problem, where the solution is guaranteed to exist via a non-constructive parity argument. The class $\textsf{FIXP}$ introduces on top of that a type of numerical hardness that emerges from the introduction of multiplication and division operations: These operations give rise to irrationality in the exact solutions to these problems, and may independently also require the computation of rational numbers of very high precision or very high magnitude (which is enabled by the presence of the multiplication operation, by which one can perform successive squaring inside the associated algebraic circuit).









\section{FIXP-Completeness of CDS-Clearing}\label{sec:fixp}
Our first main result shows that \problem\ and its strong approximation variant are $\textsf{FIXP}_{(a)}$ complete.
What is remarkable about our proof for this, is that our hardness reduction does not start from any existing $\textsf{FIXP}$-hard problem. Rather, we show that we can take an arbitrary algebraic circuit and encode it in a direct way in the form of a financial system. Hence, our polynomial time hardness reduction is implicitly defined to work from to any arbitrary fixed point problem in $\textsf{FIXP}$. The reduction is constructed by devising various financial network gadgets which enforce that certain banks in the system have recovery rates that are the result of applying one of the operators in $\textsf{FIXP}$'s arithmetic base to the recovery rates of two other banks in the system: In other words, we can design our financial systems such that the interrelation between the recovery rates mimics a computation through an arbitrary algebraic circuit.
\begin{theorem}\label{thm:fixp}
\problem\ is $\textsf{FIXP}$-complete, and its strong approximation version is $\textsf{FIXP}_a$-complete.
\end{theorem}
\begin{proof}
Containment of \problem\ in $\textsf{FIXP}$ is immediate: The clearing vectors for an instance $I = (N,e,c) \in \problem$ are the fixed points of the function $f_I$ defined coordinate-wise by
\begin{equation*}
    f_I(r)_i = \frac{a_i(r)}{\max\{l_i(r),a_i(r)\}}
\end{equation*}
as in (\ref{eq:f}). The functions $l_i(r)$ and $a_i(r)$ are defined in (\ref{eq:liabilities}) and (\ref{eq:assets}), from which it is clear that $f_I$ can be computed using a polynomial size algebraic circuit with only of $\{\max, +, *\}$, and rational constants. Note that non-degeneracy of $I$ prevents division by $0$, so that the output of the circuit is well-defined for every $x \in [0,1]^n$.~\footnote{Although not explicitly stated in our definitions, we can always assume that in our networks there are no isolated vertices with no assets and no liabilities.} This shows that \problem\ is in $\textsf{FIXP}$ and that its strong approximation version is in $\textsf{FIXP}_a$.

For the $\textsf{FIXP}$-hardness of the problem, let $\Pi$ be an arbitrary problem in $\textsf{FIXP}$. We describe a polynomial-time reduction from $\Pi$ to \problem. Let $I \in \Pi$ be an instance, let $F_I : [0,1]^n \rightarrow [0,1]^n$ be $I$'s associated fixed point function, and let $C_I$ be the algebraic circuit corresponding to $F_I$. As a pre-processing step, we convert $C_I$ to an equivalent alternative circuit $C_I'$ that satisfies that all the signals propagated by all gates in $C_I'$ and all the used rational constants in $C_I'$ are contained in the interval $[0,1]$. The transformed circuit $C_I'$ may contain two additional type of gates: Division gates and gates that computes the absolute value of the difference of two operands. We will refer to the latter type of gate as an \emph{absolute difference gate}. The circuit $C_I'$ will not contain any subtraction gates, and will not contain $\max$ and $\min$ gates either. 

The transformation procedure for $C_I$ follows the same approach of the transformation given in Theorem 18 of \cite{etessami2010complexity} where the 3-Player Nash equilibrium problem is proved $\textsf{FIXP}$-complete, and borrows some important ideas from there. Nonetheless, there are important differences in our transformation, starting with the fact that we use a different set of types of gates in our circuit. 

First, transforming $C_I$ into a circuit with only non-negative rational constants is trivial, since we can introduce subtraction gates combined with the rational constant $0$ (which we will get rid of later on in the transformation process). 

As a next step, we remove all $\min$ and $\max$ gates and replace them with addition, subtraction, multiplication, and absolute difference gates through the identities $\max\{a,b\} = (1/2)\cdot ((a+b) + |a-b|)$ and $\min\{a,b\} = (1/2)\cdot((a+b) - |a-b|)$. Let the resulting circuit be $C_I''$. 

The third step in the transformation is to move all the subtraction gates to the top of the circuit, right before the output gates, which results in exactly one subtraction gate being present for each output variable. 
This step is realised by representing each (potentially negative) signal $s$ in $C_I''$ that is propagated between two arithmetic gates, by two separate signals $s^+, s^- \geq 0$, which represent the positive and negative parts of $s$ respectively, and are guaranteed to have values such that $s^+ - s^- = s$: Each arithmetic gate $g_i$ in the circuit is replaced (sequentially, from the bottom of the circuit upward) by gates that operate on the separate positive and negative parts of the original input signals as follows. 

Suppose $g_i$ is a node outputting a constant rational value $c > 0$, then we may replace $g_i$ by two constant nodes, one which outputs $c$, the positive part of the resulting separated signal, and one which outputs $0$, representing the negative part of the signal. 

Suppose next that $g_i$ is an addition gate operating originally on the outputs $s$ and $t$ of two other gates. Then, we use the observation that for the two signals $s$ and $t$, separated into positive and negative parts $s^+, s^-, t^+, t^-$, it holds that $(s^+ - s^-) + (t^+ - t^-) = (s^+ + t^+) - (s^- + t^-)$. Therefore, $g_i$ can be replaced by two addition gates $g_i^+$ and $g_i^-$, where $g_i^+$ operates on signals $s^+$ and $t^+$, and $g_i^-$ operates on gates $s^-$ and $t^-$. 

If $g_i$ is a multiplication gate, note that for two inputs $s$ and $t$, separated into positive and negative parts $s^+,s^-,t^+,t^-$, it holds that $(s^+ - s^-)\cdot(t^+ - t^-) = (s^+t^+ + s^-t^-) - (s^+t^- + s^-t^+)$, so we may similarly replace each multiplication gate with four multiplication gates and two addition gates accordingly.

When $g_i$ is an absolute difference gate, we note that for two inputs $s$ and $t$, separated into positive and negative parts $s^+,s^-,t^+,t^-$, it holds that $|(s^+ - s^-) - (t^+ - t^-)| = |(s^+ + t^-) - (s^- + t^+)| - 0$ so that we can replace $g_i$ with two addition gates, one absolute difference gate, and one constant gate outputting $0$ (representing the negative part of the resulting signal).

Lastly, when $g_i$ is a subtraction gate in $C_I''$, we remove it and connect the gates that provide the inputs to $g_i$ appropriately to the subsequent gates that $g_i$ points to. The final subtraction gates introduced in the top layer of the circuit are directly connected to the output gates, and subtract the positive and negative parts of the final output signal, so as to generate the correct $n$ outputs. Since the function $F_I$ maps $[0,1]^n$ to itself, and the resulting circuit is equivalent to the original circuit $C_I$, we now have a circuit where all signals in the circuit are non-negative, and where we have introduced the absolute difference operation as an additional type of gate into the circuit. Next, we perform a straightforward step that eliminates the remaining $m$ subtraction gates altogether, by replacing them with absolute difference gates. This gives us an equivalent circuit, as we know that the result of the $n$ subtractions at the top of the circuit are positive for all possible inputs. We call this resulting circuit without subtraction gates $C_I'''$.

The last step is to transform the circuit $C_I'''$ into a circuit where all signals inside the circuit are in $[0,1]$ for all possible inputs to the circuit. To do so, in case there are rational constants exceeding $1$ present in the circuit, we first multiply all the rational constants by the inverse $c \in [0,1]$ of the largest rational constant in the circuit. We add the appropriate gates at the start of the circuit that multiply all input gates by $c$, and we add division gates at the end of the circuit that divide the final signal by $c$, before being propagated to the output gates. Furthermore, for each multiplication gate, we add a division gate that divides its result by $c$. This results in a circuit where all rational constant nodes are in $[0,1]$, and all internal signals of the original circuit are essentially scaled by a factor of $c$. The resulting circuit is thus equivalent to $C_I'''$.\footnote{Note that the division by $c$ that we introduce right after every multiplication gate is needed because multiplying two scaled signals $c\cdot s$ and $c\cdot s'$ results in $c^2 s s'$, hence dividing this result by $c$ results in the desired output signal $css'$.} 

Now that we have that all the rational constants are in $[0,1]$, our final step is transforming the circuit further in such a way that all signals inside the circuit are contained in $[0,1]$ as well.
Observe that the magnitude of any signal in the circuit is at most doubly exponential in the size of $C_I'''$ (this can be attained when a certain signal $s > 1$ goes through a chain of successive squaring operations, using multiplication gates). We thus perform the following final transformation to enforce that all signals stay in $[0,1]$: Let $d$ be a number, bounded by a polynomial in $|I|$, such that all signals in the circuit are strictly smaller than $2^{2^d}$ for every input in $[0,1]^n$. We add to the start of the circuit an auxiliary circuit $T$ that computes $t = 1/2^{2^d}$ by successively squaring the constant $1/2$ a number of $d$ times, and we use $T$ right after each input node and right after each rational constant node in order to multiply every input signal and every rational constant node by $t$. For each multiplication gate, we add a division gate such that its result is divided by $t$. Lastly, at the end of the circuit, we add a division gate just before each output gate, that divides the final signal by $t$. We now observe that each signal in the new circuit has effectively been multiplied by a factor of $t$, with exception of the signals directly propagated by the input gates, and directly entering the output gates, which still retain their original values.

This completes our pre-processing steps on the algebraic circuit, and we denote the resulting circuit by $C_I'$.
The circuit $C_I'$ has the desired properties that we are looking for: All signals in the circuit are in $[0,1]$ regardless of the input. Furthermore, the applied transformation ensures that $C_I'$ is equivalent to the original $C_I$ and thus represents $F_I$. The circuit can be obtained in a polynomial number of computation steps from $C_I$ and is thus polynomial-time computable from the instance $I$. The circuit $C_I'$ consists of $\{+,*,/\}$-gates, as well as absolute difference gates that compute the absolute value of the difference of two inputs. Lastly, there are rational constant nodes in the circuit where all such constants are in $[0,1]$.
For notational convenience, in the remainder of the proof we may treat $C_I'$ as the function $F_I$, hence we may write $C_I'(x) = y$ to denote $F_I(x)=y$. 

In the remainder of the proof, we define our reduction $\rho$ to \problem: We construct our instance $\rho(I)$ of \problem\ (i.e., a non-degenerate financial network) from the circuit $C_I'$. The instance $\rho(I)$ will have the property that its clearing vectors are in one-to-one correspondence with the fixed points of $C_I'$, and that banks $1,\ldots, n$ in our construction correspond to the input gates of $C_I'$. More precisely, our construction is such that for each fixed point $x$ of $C_I'$, in the corresponding clearing vector $r$ for $\rho(I)$ it holds that $(r_1,\ldots, r_n) = x$.

Our reduction works through a set financial system \emph{gadgets}, of which we prove that their recovery rates (under the clearing condition) must replicate the behaviour of each type of arithmetic operation that can occur in the circuit $C_I'$. This part of our hardness reduction shares similarities with the $\textsf{PPAD}$-hardness proof in \cite{schuldenzucker2017finding} who also design a set of gadgets to replicate the behaviour a certain circuit, although this circuit is of a different type, and our gadgets account for an entirely different set of operations.

Each of our gadgets has one or two \emph{input banks} that correspond to the input signals of one of the types of arithmetic gate, and there is an \emph{output bank} that corresponds to the output signal of the gate. For each of the gadgets, it holds that the output bank must have a recovery rate that equals the result of applying the respective arithmetic operation on the recovery rates of the input banks.

In our financial system, these gadgets are then connected together according to the structure of the circuit $C_I'$: Output banks of (copies of) gadgets are connected to input banks of other gadgets through a single unit-cost debt contract, which mimics the propagation of a signal between two gates of the arithmetic circuit. This results in a financial system whose behaviour replicates the behaviour of the arithmetic circuit. The first layer of the financial system consists of $n$ banks representing the $n$ input nodes of the circuit, and the last layer of the financial system has $n$ banks corresponding to the $n$ output nodes of the circuit. As a final step in our reduction, the $n$ output banks in the last layer are connected through a single unit-cost debt contract to the $n$ input banks. This last step enforces the recovery rates of the input banks (i.e. banks $1, \ldots, n$) are equal to the recovery rates of the last layer, under the clearing requirement. Consequently, any vector of clearing recovery rates for $\rho(I)$ must then correspond to a fixed point of $C_I'$, where the recovery rates of the first $n$ banks in the system equal those of the final $n$ banks, so that $C_I'(r_1, \ldots, r_n) = (r_1, \ldots, r_n)$, i.e., $(r_1, \ldots, r_n)$ is a fixed point of $C_I'$.

The above paragraphs comprise a high-level description of our reduction. We will now proceed to discuss the details. We start with defining our gadgets, using our graphical notation. As stated, our gadgets each have one or two input banks, and one output bank. Each gadget represents certain arithmetic operations, and gadgets can be combined with each other to form the arithmetic basis $\{+,*,/\}$ and the absolute difference operation defined above. 

We start with a straightforward addition gadget, named $g_+$, depicted in Figure \ref{fig:addgadget}. This gadget directly accounts for the addition operation in the arithmetic basis. In our figures, input banks are denoted by black arrows incoming from the left, and output banks correspond to black arrows pointing out of the bank. The black arrows represent connections to other gadgets, and these connections are always realised by a debt contract of unit cost, and are always from the output node of a gadget to an input node of another gadget. The assets that flow into an input gadget are thus always in $[0,1]$ and all gadgets are designed such that the incoming flow of assets must equal the clearing recovery rate of a bank. For example: For both input banks of the addition gadget, their liabilities are equal to $1$, given that their assets are equal to the incoming flow they receive, their recovery rate under any clearing vector (which must equal the ratio of assets and liabilities) is equal to the inflow along the black arc.

In the figures representing our gadgets, some of the nodes have been annotated with a formula in terms of the recovery rates of the input banks of the gadget, subject to the resulting values being in the interval $[0,1]$. This can be seen for example in the output node of our addition gadget in Figure \ref{fig:addgadget}. Such a formula represents the value that a clearing recovery rate must satisfy for the respective node. It is straightforward to verify that the given formulas are correct for each of our gadgets. 

Since all signals inside $C_I'$ are guaranteed to be in $[0,1]$ for all input vectors, our financial system gadgets can readily be used and connected to each other to construct a financial system that corresponds to $C_I'$, i.e., such that the clearing recovery rates of the input and output banks of each of the gadgets must correspond to each of the signals inside the circuit $C_I'$.

\begin{figure}[htbp!]
    \centering
\begin{tikzpicture}
[shorten >=1pt,node distance=2cm,initial text=]
\tikzstyle{every state}=[draw=black!50,very thick]
\tikzset{every state/.style={minimum size=0pt}}
\tikzstyle{accepting}=[accepting by arrow]
\node[state,initial] (1) {$r_1$};
\node[state,initial] (2)[below of=1] {$r_2$};
\node[state,accepting](3) [below right of=1,xshift=5mm,yshift=5mm] {$r_1 + r_2$};
\draw[blue,->,very thick] (1)--node[midway,black,yshift=2mm]{1}(3);
\draw[blue,->,very thick] (2)--node[midway,black,yshift=3mm]{1}(3);
\end{tikzpicture}
\caption{Addition gadget $g_+$.}
    \label{fig:addgadget}
\end{figure}
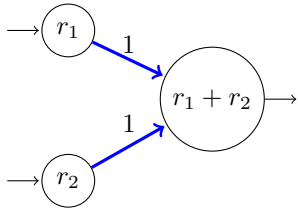

The remaining gadgets are displayed in Figures \ref{fig:duplicationgadget} to \ref{fig:absdifferencegadget}, in Appendix \ref{apx:gadgets}. 


Besides gadgets for the necessary arithmetic operations, our reduction employs an additional \emph{duplication gadget} $g_{\text{dup}}$ that can be used to connect the output of a particular gadget to the input of more than one subsequent gadget. This gadget is displayed in Figure \ref{fig:duplicationgadget}.
Gadget $g_{\text{dup}}$ furthermore has the convenient property that it can be used for multiplication by a rational $c \in [0,1]$: One of the output nodes has the recovery of the input node, while the other has this recovery rate multiplied by $c$. When choosing $c = 1$ the input recovery rate is effectively duplicated. 

The gadgets for multiplication and division, $g_{*}$ and $g_{/}$ are given in simplified form in Figures \ref{fig:multiplicationgadget} and \ref{fig:divisiongadget}. These gadgets work as intended, but the problem with them is that they do not satisfy the non-degeneracy condition that we assume instances of \problem\ to have.
It is possible to transform these into non-degenerate gadgets, but this process is technically detailed, in particular for the case of division: It requires replacing some of the divisions in the circuit $C_I'$ by taking successive square-roots followed by successive squaring operations, where proper care has to be taken to ensure that the results of all these operations stay within the interval $[0,1]$.
The non-degenerate versions of our gadgets (which interestingly includes a gadget that outputs square roots of input recovery rates), and the additional transformation of $C_I'$ that  is required for this, are described in Appendix \ref{apx:nondegenerate}.

The absolute difference gadget $g_{\text{abs}}$ is given in Figure \ref{fig:absdifferencegadget}. It uses $g_{\text{dup}}$ and $g_{+}$, as sub-gadgets in its definition, as well as the sub-gadget $g_{\text{pos}-}$ (given in Figure \ref{fig:subtractiongadget}) that takes two input recovery rates $r_1$ and $r_2$, and produces an output recovery rate of $\max\{r_1-r_2,0\}$.

Besides this set of gadgets that correspond to the arithmetic base of the circuit $C_I'$, in order to translate the circuit $C_I'$ appropriately into a financial system, we also need the ability to  generate rational constants as inputs to the gadgets (which correspond to the rational constant nodes of $C_I'$). This is straightforward: Any rational recovery rate $c \in [0,1]$ can be generated using a single node $i$ with a single outgoing debt contract of unit cost, that can be pointed towards an input node of any gadget copy.

We provide some further financial system gadgets in Appendix \ref{apx:furthergadgets}. These gadgets demonstrate the ability of financial systems to simulate the $\max$ and $\min$ operations. We included these for the interested reader, but those gadgets are otherwise not used in our reduction, since $C_I'$ does not have any $\min$ and $\max$ gates. 

With the set of gadgets we have now defined, we can proceed to interconnect copies of gadgets to create financial systems of which the clearing recovery rates of the input and output gadgets respects an arbitrary algebraic computations over the arithmetic basis $\{+,-,*,/\}$, the absolute difference operator, and rational constants. In particular, we can connect $n$ input nodes, say banks $\{1,\ldots, n\}$, to a network of gadgets that implements the algebraic circuit $C_I'$. We connect its $n$ outputs, say banks $(m-n+1, \ldots, m)$ (where $m$ is the number of banks in the resulting financial system) pairwise to the $n$ input nodes, and we define the resulting financial system to be $\rho(I)$. It is clear that $\rho(I)$ can be constructed in polynomial time from $C_I'$, and since $C_I'$ can be constructed in polynomial time from $I$, the financial system $\rho(I)$ takes polynomial time to compute.

Next, we show that clearing vectors of $\rho(I)$ correspond in a polynomial-time computable way (in this case: trivially) to fixed points of $C_I'$, thereby proving the correctness of the reduction and completing the proof. Let $r$ be a clearing vector of $\rho(I)$. We show that $(r_1, \ldots, r_n)$ is a fixed point of $I$. Due to the fact that $r$ is clearing, and by the definition of the gadgets, we know that the assets flowing into nodes $(1,\ldots, n)$ are respectively $(r_1, \ldots, r_n)$, and come from nodes $(m-n+1, \ldots, m)$. Hence, the recovery rates of nodes $(m-n+1, \ldots, m)$ are equal to $(r_1,\ldots,r_n)$. By construction of our network, (i.e., by correspondence of our gadgets and the way they are connected to each other, following the structure of the algebraic circuit $C_I'$), it must then hold that $C_I'(r_1, \ldots, r_n) = (r_1, \ldots, r_n)$. Hence, $(r_1, \ldots, r_n)$ is a fixed point of $I$. 

This completes the proof, as the above is sufficient to establish $\textsf{FIXP}$-completeness. $\textsf{FIXP}_a$ completeness of the strong approximation variant is immediate, since any strong $\epsilon$-approximation of the recovery rate vector of $R(I)$ corresponds in the same manner to a strong $\epsilon$-approximate fixed point of $C_I'$.

In addition, it is straightforward to see that fixed points of $C_I'$ correspond to recovery rates of $\rho(I)$, so that indeed $\rho(I)$ captures the complete set of fixed points of $C_I'$: Let $x$ be a fixed point of $C_I'$. Now construct a vector of recovery rates for $\rho(I)$ by setting $(r_1, \ldots, r_n) = x$ and compute the remaining recovery rates of the nodes inside the gadgets in the natural way, as described by the arithmetic operations that each of the gadgets represents. Since $x$ is a fixed point of $C_I'$, the recovery rates of banks $(m-n+1,\ldots, m)$ will be set to $x$, which causes the recovery rates of banks $(1,\ldots, n)$ (and thereby also the entire financial system) to satisfy the clearing condition.
\end{proof}

\section{A Sufficient Structural Condition for Irrational Solutions}\label{sec:irrationality}
In \cite{schuldenzucker2016clearing}, it was pointed out that financial systems with CDS contracts can encode polynomials. This claim was made informally and therefore its meaning is somewhat imprecise. However, it raises some interesting questions about which properties of financial systems generate irrational solutions. 

In this section we investigate the existence of irrational solutions in financial systems in more depth. Our starting point is the observation that irrational clearing recovery rates can only arise under certain structural conditions on the financial system. For example, if the system has no CDSes at all, a rational clearing vector is guaranteed to exist and can be found in polynomial time, by the algorithms in \cite{eisenberg2001systemic} (one of which comes down to solving a simple linear program). Another example applies when we consider the directed graph containing all the arcs from all debtors to their corresponding creditors, and include also the arcs from reference banks of CDSes to the corresponding debtors: If this graph is acyclic, then it can be shown (as we will see below) that every clearing vector must be rational. 

Which structural conditions must exactly hold in a financial system for irrational clearing vectors to potentially exist? We aim to characterise such structural conditions of the financial system graph independently of the relevant numerical values in these systems (i.e., the values of the external assets, and the notionals on the debt contracts and CDSes): Given a partially specified instance of a financial system where these numerical values are not specified, and thus only the debtors, creditors, and reference banks (in the case of CDSes) are provided to us for each contract in the instance, can we identify in which cases there are numerical coefficients (notionals $c$ and external assets $e$) for the instance such that under these coefficients, no rational solutions exist? In this section, we present a set of sufficient structural conditions that provides a partial answer to this question. 
In the subsequent section, we study the complementary goal of identifying under which conditions a rational solution must exist, and formulate a second set of structural conditions that are close to our former irrationality conditions.



We define a type of auxiliary coloured directed graph, associated to a financial system, for which we examine specific types of cycles,
and prove that the presence of such a cycle is a structurally sufficient condition for irrational solutions to arise, in the sense that we can then set the rational coefficients such that every clearing vector of the system is irrational.  


\subsection{Switched Cycles}
Assume an instance $I = (N,e,c)$ of \problem\ and let $G_I$ be its contract graph. We construct an auxiliary coloured directed graph $G_{I,\text{aux}}$ as follows: We include all the arcs of the contract graph in $G_{I,\text{aux}}$, which retain their blue and orange colours. Furthermore, for every pair $(i,j)$ such that there exists a CDS $(i,j,R) \in \mathcal{CDS}$ (for some $R \in N$), we ad a red-coloured arc $(i,j)$ to $G_{I,\text{aux}}$. Thus, in $G_{I,\text{aux}}$ there is at most one red, one orange, and one blue arc between every ordered pair of nodes. We refer to the resulting tricoloured directed graph as simply the \emph{auxiliary graph} of $I$. An example of a financial system and its auxiliary graph is given in Figure \ref{fig:auxgraphexample}. 
\begin{figure}[htb]
    \centering
\begin{tikzpicture}[shorten >=1pt,node distance=2cm,initial text=]
\tikzstyle{every state}=[draw=black!50,very thick]
\tikzset{every state/.style={minimum size=0pt}}
\node[state] (1) {$1$}; 
\node[state] (2) [right of=1] {$2$};
\node[state] (3) [right of=2] {$3$};
\node[state] (4) [right of=3] {$4$};
\draw[orange,very thick,->,snake=snake] (2)--(1);
\draw[blue,very thick,->] (2)--(3);
\draw[blue,very thick,->] (3)--(4);
\node[state] (5) [below of=1] {$5$};
\node[state] (6) [below of=2] {$6$};
\node[state] (7) [below of=3] {$7$};
\node[state] (8) [below of=4] {$8$};
\draw[blue,very thick,->] (6)--(5);
\draw[blue,very thick,->] (7)--(6);
\draw[orange,very thick,->,snake=snake] (7)--(8);
\path [orange,->,draw,dashed,thick] (6) -- ($ (2) !.5! (1) $);
\path [orange,->,draw,dashed,thick] (3) -- ($ (7) !.5! (8) $);
\end{tikzpicture}
\qquad
\begin{tikzpicture}[shorten >=1pt,node distance=2cm,initial text=]
\tikzstyle{every state}=[draw=black!50,very thick]
\tikzset{every state/.style={minimum size=0pt}}
\node[state] (1) {$1$}; 
\node[state] (2) [right of=1] {$2$};
\node[state] (3) [right of=2] {$3$};
\node[state] (4) [right of=3] {$4$};
\draw[orange,very thick,->,snake=snake] (2)--(1);
\draw[blue,very thick,->] (2)--(3);
\draw[blue,very thick,->] (3)--(4);
\node[state] (5) [below of=1] {$5$};
\node[state] (6) [below of=2] {$6$};
\node[state] (7) [below of=3] {$7$};
\node[state] (8) [below of=4] {$8$};
\draw[blue,very thick,->] (6)--(5);
\draw[blue,very thick,->] (7)--(6);
\draw[orange,very thick,->,snake=snake] (7)--(8);
\draw[red,very thick,->] (6)--(2);
\draw[red,very thick,->] (3)--(7);
\path [orange,->,draw,dashed,thick] (6) -- ($ (2) !.5! (1) $);
\path [orange,->,draw,dashed,thick] (3) -- ($ (7) !.5! (8) $);
\end{tikzpicture}
\caption{A financial system $G_I$ represented as its contract graph, and the corresponding auxiliary graph $G_{I,\text{aux}}$. The dashed lines used in our original graphical notation are retained, for clarity. The coefficients along the arcs and on the nodes are omitted.}
    \label{fig:auxgraphexample}
\end{figure}
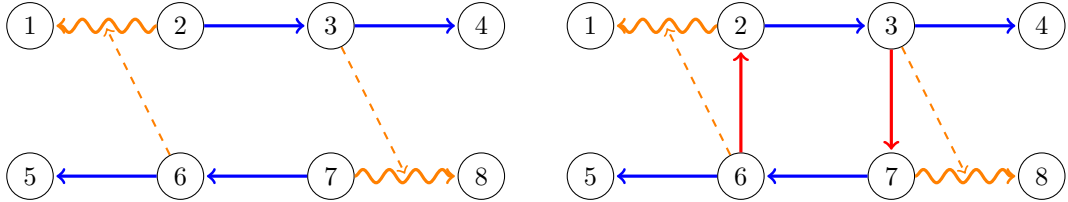


\begin{definition}\label{def:DAG}
A financial system is called \emph{acyclic} iff its auxiliary graph does not contain any directed cycle.
\end{definition}
The financial system of Figure \ref{fig:auxgraphexample} is not acyclic, since its auxiliary graph contains the cycle $(2,3,7,6)$. Acyclic financial systems turn out to be easy to analyse: the clearing recovery rate vector is always rational. 
\begin{proposition}
Every acyclic financial system has only rational solutions.
\end{proposition}
\begin{proof}
Assume an acyclic financial system $I = (N,e,c)$ and let $G$ and $G_{\text{aux}} = (N, A_{\text{aux}})$ be its contract graph and auxiliary graph respectively. Since $I$ is acyclic, $G_{\text{aux}}$ contains no cycles, meaning that for any $(i,j,R) \in \mathcal{CDS}$, there is no path from $i$ to $R$. Since $G_{\text{aux}}$ is a directed acyclic graph (DAG), we can rearrange its nodes in topological order in polynomial time (see e.g. \cite{kleinberg2006algorithm}). 
Let $T$ be the topological ordering of $G_{aux}$. Without loss of generality, we assume that each node $i \in N$ is also the $i$'th node in $T$ so that outgoing arcs of $i$ point to higher-numbered nodes, and incoming arcs of $i$ point to lower-numbered nodes. 

We observe that the recovery rate $r_i$ of any node $i \in N$ is determined by only the recovery rates of banks $j < i$, under the clearing condition, which holds because both the assets $a_i(r)$ and liabilities $l_i(r)$ under any clearing vector do not depend on any of the recovery rates $r_{i+1}, \ldots, r_n$ (following (\ref{eq:liabilities}) and (\ref{eq:assets})). This means that we can straightforwardly compute a clearing vector of recovery rates by iterating over the banks in topological order $T$, and using (\ref{eq:liabilities}) and (\ref{eq:assets}) to compute $r_i$ from $r_1, \ldots r_{i-1}$ in each iteration $i$. 

In each iteration, the computation of the numerator and denominator of the recovery rate involves summations of terms containing multiplications, subtractions, and additions of recovery rates with other recovery rates and with rational constants. The numerator and denominator are thus rationals if all preceding recovery rates are rationals as well. Since the recovery rate computed in the first iteration is rational, by induction all computed recovery rates are rational. 
%
%
%
\end{proof}
 
This initial insight leads us to focus on specifying cyclic structures that might be capable of generating irrationality. The classical results of \cite{eisenberg2001systemic} show that financial systems without credit default swaps always have rational (and polynomial time computable) recovery rates. Thus, for irrational solutions to emerge, we know that the presence of just \emph{any} cycle is not sufficient, and that CDSes must be involved in some way. 

\begin{definition}[Switched off/switched on nodes]
We define a node $i \in N$ as \emph{switched off} iff its number of incoming red arcs equals 1, and $i$ does not have outgoing blue arcs (this implies that all CDSes of which $i$ is the debtor share the same reference bank). We define a node $i \in N$ as \emph{switched on} iff one of the following holds:
\begin{itemize}
    \item its number of incoming red arcs exceeds 1;
    \item its number of incoming red arcs equals 1 and there is at least one outgoing blue arc.
\end{itemize} 
\end{definition}  
We note that \emph{switched on} and \emph{switched off} nodes are not complements of each other: A node that is not a debtor in any CDS is neither switched on nor switched off.  

Figure \ref{fig:switches} below illustrates these notions.
\begin{figure}[htbp!]
    \centering
\begin{tikzpicture}[shorten >=1pt,node distance=2cm,initial text=]
\tikzstyle{every state}=[draw=black!50,very thick]
\tikzset{every state/.style={minimum size=0pt}}
\node[state] (1) {}; 
\node[state] (2) [right of=1] {off};
\node[state] (3) [above right of=2]{};
\node[state] (4) [below right of=2]{};
\draw[orange,very thick,snake=snake,->](2)--(3);
\draw[orange,very thick,snake=snake,->](2)--(4);
\draw[red,very thick,->](1)--(2);
\path[orange,->,draw,dashed,thick] (1) -- ($ (2) !.5! (3) $);
\path[orange,->,draw,dashed,thick] (1) -- ($ (2) !.5! (4) $);
\end{tikzpicture}
\qquad
\begin{tikzpicture}[shorten >=1pt,node distance=2cm,initial text=]
\tikzstyle{every state}=[draw=black!50,very thick]
\tikzset{every state/.style={minimum size=0pt}}
\node[state] (1) {}; 
\node[state] (2) [right of=1] {on};
\node[state] (3) [above right of=2] {};
\node[state] (4) [below right of=2] {};
\draw[orange,very thick,snake=snake,->](2)--(3);
\draw[blue,very thick,->](2)--(4);
\draw[red,very thick,->](1)--(2);
\path[orange,->,draw,dashed,thick] (1) -- ($ (2) !.5! (3) $);
\end{tikzpicture}
\qquad
\begin{tikzpicture}[shorten >=1pt,node distance=2cm,initial text=]
\tikzstyle{every state}=[draw=black!50,very thick]
\tikzset{every state/.style={minimum size=0pt}}
\node[state] (1) {}; 
\node[state] (2) [right of=1] {on};
\node[state] (3) [above right of=2] {};
\node[state] (4) [below right of=2] {};
\node[state] (5)[below of=1]{};
\draw[red,very thick,->](5)--(2);
\draw[orange,very thick,snake=snake,->](2)--(3);
\draw[orange,very thick,->,snake=snake](2)--(4);
\draw[red,very thick,->](1)--(2);
\path[orange,->,draw,dashed,thick] (1) -- ($ (2) !.5! (3) $);
\path[orange,->,draw,dashed,thick] (5) -- ($ (2) !.5! (4) $);
\end{tikzpicture}
\caption{One switched off and two switched on nodes.}
    \label{fig:switches}
\end{figure}
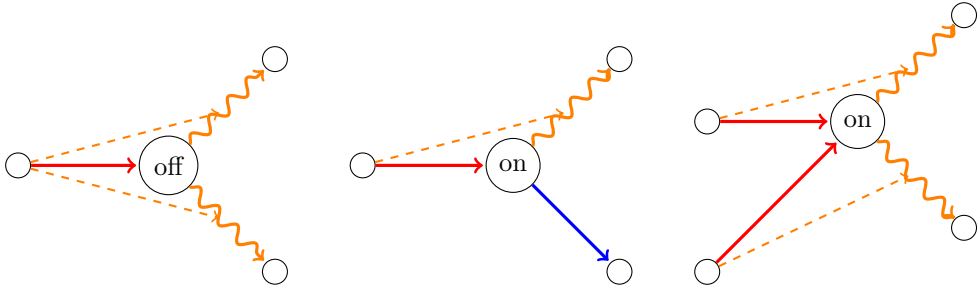

\begin{definition}[Switched Cycles]
We call a cycle \emph{red} iff it has at least one red arc. 
A cycle $C$ is called \emph{weakly switched} iff it is red, and for at least one red arc $(u,v)$ in $C$, $v$ is switched on. 
A cycle is called \emph{strongly switched} iff it is red, and for each red arc $(u,v)$ in $C$, $v$ is switched on.
\end{definition}


The notion of switched cycles is central to the occurrence of irrational solutions in financial systems. We will show that when a non-degenerate financial system $I$ has a strongly switched cycle, and satisfies a certain additional technical condition holds (to be specified later), there exist coefficients for the financial system under which all clearing vectors of $I$ are irrational.
Conversely, in the subsequent section we show that in the absence of any \emph{weakly} switched cycle, all clearing vectors of $I$ are guaranteed to be rational regardless of its coefficients, and the computational problem of computing an exact fixed point in such instances lies in a complexity class that is closely related to $\textsf{PPAD}$.
  
   
\subsection{Rewriting Rules for Strongly Switched Cycles}
In this section, we first present a framework for formulating strongly switched cycles, and we present various ``rewriting rules'' that capture an equivalence relation between strongly switched cycles in terms of their recovery rates. We begin by defining a set of primitive financial systems, represented in the form of their auxiliary graph, which we call \emph{fragments}. These fragments lack coefficients: They specify only some of the contracts present among a set of nodes, but do not contain any specification of the notionals on these contracts, and do not contain a specification of the external assets of any of the banks. Each of these fragments has a designated start and end node, and each of these fragments has the property that there is a unique path from the start to the end node. Furthermore, each fragment has the property that all nodes are directly connected to the unique path from start to end node.

We then define a binary operation that concatenates copies of fragments with each other, by identifying the start node of one copy of a fragment with the end node of a copy of another fragment. This results in auxiliary graphs of financial systems that are obtainable by ``stringing'' together fragments. We refer to graphs obtainable through this concatenation operation as \emph{fragment strings}. A fragment string represents a path (starting at the start node of the first fragment and ending at the end node of the last fragment) along with some further neighbouring nodes that are directly attached to the main path through arcs.

In turn, the start and end nodes of the paths that can be formed this way can be connected together: The end node of the last fragment in a fragment string can be identified with the start node of the first fragment in the fragment string, and this will create a (coefficientless) auxiliary graph of some financial system that contains a single cycle, along with some neighbouring nodes with arcs that are attached to the nodes in the main cycle. We call such a graph a \emph{fragment cycle}.

We will explain below in detail how these fragments are defined and how to form cycles with them. Subsequently, we will equip each fragment with particular choices of coefficients (i.e., notionals and external assets) that we will later on show to generate unique irrational solutions when composed into fragment cycles.

We study the resulting set of fragment cycle graphs, constructible through attaching such fragments to each other, with the following goal in mind:
Given an arbitrary financial system $I$, we will show that if a certain subset of the graphs constructible from our set of fragments occurs as a subgraph in $G_{I,\text{aux}}$, then there exists a choice of rational coefficients for $I$ such that all clearing vectors of $I$ are irrational. This will yield us our intended structural characterisation of irrational financial systems. Therefore, the fragments we will define, and the graphs obtainable from stringing them together, should be interpreted as subgraphs of a larger financial system under consideration, and we want to determine whether this larger financial system is susceptible to having irrational clearing vectors.


\subsubsection{Fragment Strings}
We denote by $\mathcal{G}$ the set of all fragments that we will use. All fragments of $\mathcal{G}$ are defined in Figure \ref{fig:base}, presented in our tricoloured graphical notation. Each fragment has designated start and end nodes, which are indicated by short incoming and outgoing black arrows, respectively. It can be verified that for each fragment indeed there is a single path from input to output node, as we claimed previously. Furthermore, as also claimed above, each node of any given fragment is either on this path, or directly connected to it by a single (incoming or outgoing) arc. Note that there are many fragments in $\mathcal{G}$ that are highly similar, and the names of our fragments are chosen such that highly similar fragments differ only in their superscripts: For example, the fragments $g_1^a, g_1^b, g_1^c$, and $g_1^d$ all have in common that the main path from the input node to the output node consist of a red arc followed by an orange arc that corresponds to the same CDS as the red arc, and these four fragments only differ in the colours of two arcs attached to the main path (along with one or two external reference banks, denoted by $c, c_1$, or $c_2$). The two sets of fragments $\{g_2^a,g_2^b,g_2^c,g_2^d\}$ and $\{g_3^a,g_3^d\}$ are related in a similar fashion. 

We define a binary merging operation on ordered pairs of fragments $(a,b)$, where every pair $(a,b)$ is mapped to a graph obtained by taking disjoint copies of $a$ and $b$, and connecting the two copies together by identifying the end node of $a$ with the start node of $b$. We define the new start node and end node of the resulting system to be the start node of the copy of $a$ and the end node of the copy of $b$, respectively. We denote the result of the merge operation on fragments $a$ and $b$ symbolically by the notation $ab$. A \emph{fragment string} is a fragment obtainable from fragments in $\mathcal{G}$ using any number of sequential applications of the merge operation. We let $\mathcal{GS}$ be the set of fragment strings (i.e., the closure of $\mathcal{G}$ under the \emph{merge} operation).

\begin{figure}[htbp!]
    \centering
\scalebox{0.7}{\begin{tikzpicture}
[shorten >=1pt,node distance=2cm,initial text=]
\tikzstyle{every state}=[draw=black!50,very thick]
\tikzset{every state/.style={minimum size=0pt}}
\tikzstyle{accepting}=[accepting by arrow]
\node[state,initial] (1) {$1$};
\node[state] (2)[above of=1] {$2$};
\node[state] (3)[left of=2]{$3$};
\node[state] (4)[right of=1]{$4$};
\node[state,accepting] (5)[right of=2]{$5$};
\draw[blue,->,very thick] (2)--(3);
\draw[orange,->,very thick,snake=snake] (2)--(5);
\draw[red,->,very thick] (1)--(2);
\draw[blue,->,very thick] (1)--(4);
\path [orange,-,draw,dashed,thick] (1) -- ($ (2) !.5! (5) $);
\node[black,left of=1,xshift=1cm,yshift=1cm]{$g_1^a$};
\end{tikzpicture}
\quad
\begin{tikzpicture}
[shorten >=1pt,node distance=2cm,initial text=]
\tikzstyle{every state}=[draw=black!50,very thick]
\tikzset{every state/.style={minimum size=0pt}}
\tikzstyle{accepting}=[accepting by arrow]
\node[state,initial] (1) {$1$};
\node[state] (2)[above of=1] {$2$};
\node[state] (3)[left of=2]{$3$};
\node[state] (4)[right of=1]{$4$};
\node[state,accepting] (5)[right of=2]{$5$};
\node[state] (6)[above right of=3,xshift=-5mm, yshift=-5mm]{$c$};
\draw[orange,->,very thick,snake=snake] (2)--(3);
\draw[orange,->,very thick,snake=snake] (2)--(5);
\draw[red,->,very thick] (1)--(2);
\draw[blue,->,very thick] (1)--(4);
\path [orange,-,draw,dashed,thick] (1) -- ($ (2) !.5! (5) $);
\path [orange,-,draw,dashed,thick] (6) -- ($ (2) !.5! (3) $);
\draw[red,->,very thick] (6)--(2);
\node[black,left of=1,xshift=1cm,yshift=1cm]{$g_1^b$};
\end{tikzpicture}
\quad
\begin{tikzpicture}
[shorten >=1pt,node distance=2cm,initial text=]
\tikzstyle{every state}=[draw=black!50,very thick]
\tikzset{every state/.style={minimum size=0pt}}
\tikzstyle{accepting}=[accepting by arrow]
\node[state,initial] (1) {$1$};
\node[state] (2)[above of=1] {$2$};
\node[state] (3)[left of=2]{$3$};
\node[state] (4)[right of=1]{$4$};
\node[state,accepting] (5)[right of=2]{$5$};
\node[state] (6)[above right of=1,xshift=-5mm,yshift=-5mm]{$c$};
\draw[blue,->,very thick] (2)--(3);
\draw[orange,->,very thick,snake=snake] (2)--(5);
\draw[red,->,very thick] (1)--(2);
\draw[orange,->,very thick,snake=snake] (1)--(4);
\path [orange,-,draw,dashed,thick] (1) -- ($ (2) !.5! (5) $);
\path [orange,-,draw,dashed,thick] (6) -- ($ (1) !.5! (4) $);
\draw[red,->,very thick] (6)--(1);
\node[black,left of=1,xshift=1cm,yshift=1cm]{$g_1^c$};
\end{tikzpicture}
\quad
\begin{tikzpicture}
[shorten >=1pt,node distance=2cm,initial text=]
\tikzstyle{every state}=[draw=black!50,very thick]
\tikzset{every state/.style={minimum size=0pt}}
\tikzstyle{accepting}=[accepting by arrow]
\node[state,initial] (1) {$1$};
\node[state] (2)[above of=1] {$2$};
\node[state] (3)[left of=2]{$3$};
\node[state] (4)[right of=1]{$4$};
\node[state,accepting] (5)[right of=2]{$5$};
\node[state] (6)[above right of=1,xshift=-5mm,yshift=-5mm]{$c_1$};
\node[state] (7)[above right of=3,xshift=-5mm,yshift=-5mm]{$c_2$};
\draw[orange,->,very thick,snake=snake] (2)--(3);
\draw[orange,->,very thick,snake=snake] (2)--(5);
\draw[red,->,very thick] (1)--(2);
\draw[orange,->,very thick,snake=snake] (1)--(4);
\path [orange,-,draw,dashed,thick] (1) -- ($ (2) !.5! (5) $);
\path [orange,-,draw,dashed,thick] (6) -- ($ (1) !.5! (4) $);
\path [orange,-,draw,dashed,thick] (7) -- ($ (2) !.5! (3) $);
\draw[red,->,very thick] (6)--(1);
\draw[red,->,very thick] (7)--(2);
\node[black,left of=1,xshift=1cm,yshift=1cm]{$g_1^d$};
\end{tikzpicture}}

\vspace{1cm}

\scalebox{0.7}{\begin{tikzpicture}
[shorten >=1pt,node distance=2cm,initial text=]
\tikzstyle{every state}=[draw=black!50,very thick]
\tikzset{every state/.style={minimum size=0pt}}
\tikzstyle{accepting}=[accepting by arrow]
\node[state,initial] (1) {$1$};
\node[state] (2)[above of=1]{$2$};
\node[state] (3)[left of=2]{$3$};
\node[state] (4)[right of=1]{$4$};
\node[state,accepting] (5)[right of=2]{$5$};
\draw[blue,->,very thick] (2)--(5);
\draw[orange,->,very thick,snake=snake] (2)--(3);
\draw[red,->,very thick] (1)--(2);
\draw[blue,->,very thick] (1)--(4);
\path[orange,-,draw,dashed,thick] (1) -- ($ (2) !.5! (3) $);
\node[black,left of=1,xshift=1cm,yshift=1cm]{$g_2^a$};
\end{tikzpicture}
\quad
\begin{tikzpicture}
[shorten >=1pt,node distance=2cm,initial text=]
\tikzstyle{every state}=[draw=black!50,very thick]
\tikzset{every state/.style={minimum size=0pt}}
\tikzstyle{accepting}=[accepting by arrow]
\node[state,initial] (1) {$1$};
\node[state] (2)[above of=1]{$2$};
\node[state] (3)[left of=2]{$3$};
\node[state] (4)[right of=1]{$4$};
\node[state,accepting] (5)[right of=2]{$5$};
\node[state] (6)[above right of=2,xshift=-5mm,yshift=-5mm]{$c$};
\draw[orange,->,very thick,snake=snake] (2)--(5);
\draw[orange,->,very thick,snake=snake] (2)--(3);
\draw[red,->,very thick] (1)--(2);
\draw[red,->,very thick] (6)--(2);
\draw[blue,->,very thick] (1)--(4);
\path[orange,-,draw,dashed,thick] (1) -- ($ (2) !.5! (3) $);
\path[orange,-,draw,dashed,thick] (6) -- ($ (2) !.5! (5) $);
\node[black,left of=1,xshift=1cm,yshift=1cm]{$g_2^b$};
\end{tikzpicture}
\quad
\begin{tikzpicture}
[shorten >=1pt,node distance=2cm,initial text=]
\tikzstyle{every state}=[draw=black!50,very thick]
\tikzset{every state/.style={minimum size=0pt}}
\tikzstyle{accepting}=[accepting by arrow]
\node[state,initial] (1) {$1$};
\node[state] (2)[above of=1]{$2$};
\node[state] (3)[left of=2]{$3$};
\node[state] (4)[right of=1]{$4$};
\node[state,accepting] (5)[right of=2]{$5$};
\node[state] (6)[above right of=1,xshift=-5mm,yshift=-5mm]{$c$};
\draw[blue,->,very thick] (2)--(5);
\draw[orange,->,very thick,snake=snake] (2)--(3);
\draw[red,->,very thick] (1)--(2);
\draw[orange,->,very thick,snake=snake] (1)--(4);
\draw[red,->,very thick] (6)--(1);
\path[orange,-,draw,dashed,thick] (1) -- ($ (2) !.5! (3) $);
\path[orange,-,draw,dashed,thick] (6) -- ($ (1) !.5! (4) $);
\node[black,left of=1,xshift=1cm,yshift=1cm]{$g_2^c$};
\end{tikzpicture}
\quad
\begin{tikzpicture}
[shorten >=1pt,node distance=2cm,initial text=]
\tikzstyle{every state}=[draw=black!50,very thick]
\tikzset{every state/.style={minimum size=0pt}}
\tikzstyle{accepting}=[accepting by arrow]
\node[state,initial] (1) {$1$};
\node[state] (2)[above of=1]{$2$};
\node[state] (3)[left of=2]{$3$};
\node[state] (4)[right of=1]{$4$};
\node[state,accepting] (5)[right of=2]{$5$};
\node[state] (6)[above right of=1,xshift=-5mm,yshift=-5mm]{$c_1$};
\node[state] (7)[above right of=2,xshift=-5mm,yshift=-5mm]{$c_2$};
\draw[orange,->,very thick,snake=snake] (2)--(5);
\draw[orange,->,very thick,snake=snake] (2)--(3);
\draw[red,->,very thick] (1)--(2);
\draw[orange,->,very thick,snake=snake] (1)--(4);
\draw[red,->,very thick] (6)--(1);
\draw[red,->,very thick] (7)--(2);
\path[orange,-,draw,dashed,thick] (1) -- ($ (2) !.5! (3) $);
\path[orange,-,draw,dashed,thick] (6) -- ($ (1) !.5! (4) $);
\path[orange,-,draw,dashed,thick] (7) -- ($ (2) !.5! (5) $);
\node[black,left of=1,xshift=1cm,yshift=1cm]{$g_2^d$};
\end{tikzpicture}}

\vspace{1cm}

\begin{tikzpicture}
[shorten >=1pt,node distance=2cm,initial text=]
\tikzstyle{every state}=[draw=black!50,very thick]
\tikzset{every state/.style={minimum size=0pt}}
\tikzstyle{accepting}=[accepting by arrow]
\node[state,initial] (1) {$1$};
\node[state,accepting above] (4)[above of=1]{$4$};
\node[state] (3)[left of=4]{$3$};
\node[state] (2)[right of=1]{$2$};
\draw[orange,->,very thick,snake=snake] (4)--(3);
\draw[red,->,very thick] (1)--(4);
\draw[blue,->,very thick] (1)--(2);
\path[orange,-,draw,dashed,thick] (1) -- ($ (4) !.5! (3) $);
\node[black,left of=1,xshift=1cm,yshift=1cm]{$g_3^a$};
\end{tikzpicture}
\quad
\begin{tikzpicture}
[shorten >=1pt,node distance=2cm,initial text=]
\tikzstyle{every state}=[draw=black!50,very thick]
\tikzset{every state/.style={minimum size=0pt}}
\tikzstyle{accepting}=[accepting by arrow]
\node[state,initial] (1) {$1$};
\node[state,accepting above] (4)[above of=1]{$4$};
\node[state] (3)[left of=4]{$3$};
\node[state] (2)[right of=1]{$2$};
\node[state] (6)[above right of=1,xshift=-5mm,yshift=-5mm]{$c$};
\draw[orange,->,very thick,snake=snake] (4)--(3);
\draw[red,->,very thick] (1)--(4);
\draw[orange,->,very thick,snake=snake] (1)--(2);
\draw[red,->,very thick] (6)--(1);
\path[orange,-,draw,dashed,thick] (1) -- ($ (4) !.5! (3) $);
\path[orange,-,draw,dashed,thick] (6) -- ($ (1) !.5! (2) $);
\node[black,left of=1,xshift=1cm,yshift=1cm]{$g_3^b$};
\end{tikzpicture}

\vspace{1cm}

\begin{tikzpicture}
[shorten >=1pt,node distance=2cm,initial text=]
\tikzstyle{every state}=[draw=black!50,very thick]
\tikzset{every state/.style={minimum size=0pt}}
\tikzstyle{accepting}=[accepting by arrow]
\node[state,initial] (1) {$1$};
\node[state,accepting] (2)[right of=1] {$2$};
\draw[blue,very thick,->](1)--(2);
\node[black,above of=1,xshift=1.1cm,yshift=-1.5cm]{$d_1$};
\end{tikzpicture}
\qquad
\begin{tikzpicture}
[shorten >=1pt,node distance=2cm,initial text=]
\tikzstyle{every state}=[draw=black!50,very thick]
\tikzset{every state/.style={minimum size=0pt}}
\tikzstyle{accepting}=[accepting by arrow]
\node[state,initial] (1) {$1$};
\node[state,accepting] (2)[right of=1] {$2$};
\node[state] (3)[above right of=1,xshift=-5mm]{$c$};
\draw[red,very thick,->](3)--(1);
\draw[orange,very thick,snake=snake,->](1)--(2);
\path[orange,-,draw,dashed,thick] (3) -- ($ (1) !.5! (2) $);
\node[black,left of=1,xshift=1.5cm,yshift=0.5cm]{$d_2$};
\end{tikzpicture}
\caption{The fragments in $\mathcal{G}$. Each fragment is labeled with a name that we will use to refer to the individual fragments. We have $\mathcal{G} = \{g_1^a,g_1^b,g_1^c,g_1^d,g_2^a,g_2^b,g_2^c,g_2^d,g_3^a,g_3^b,d_1,d_2\}$ }
    \label{fig:base}
\end{figure}
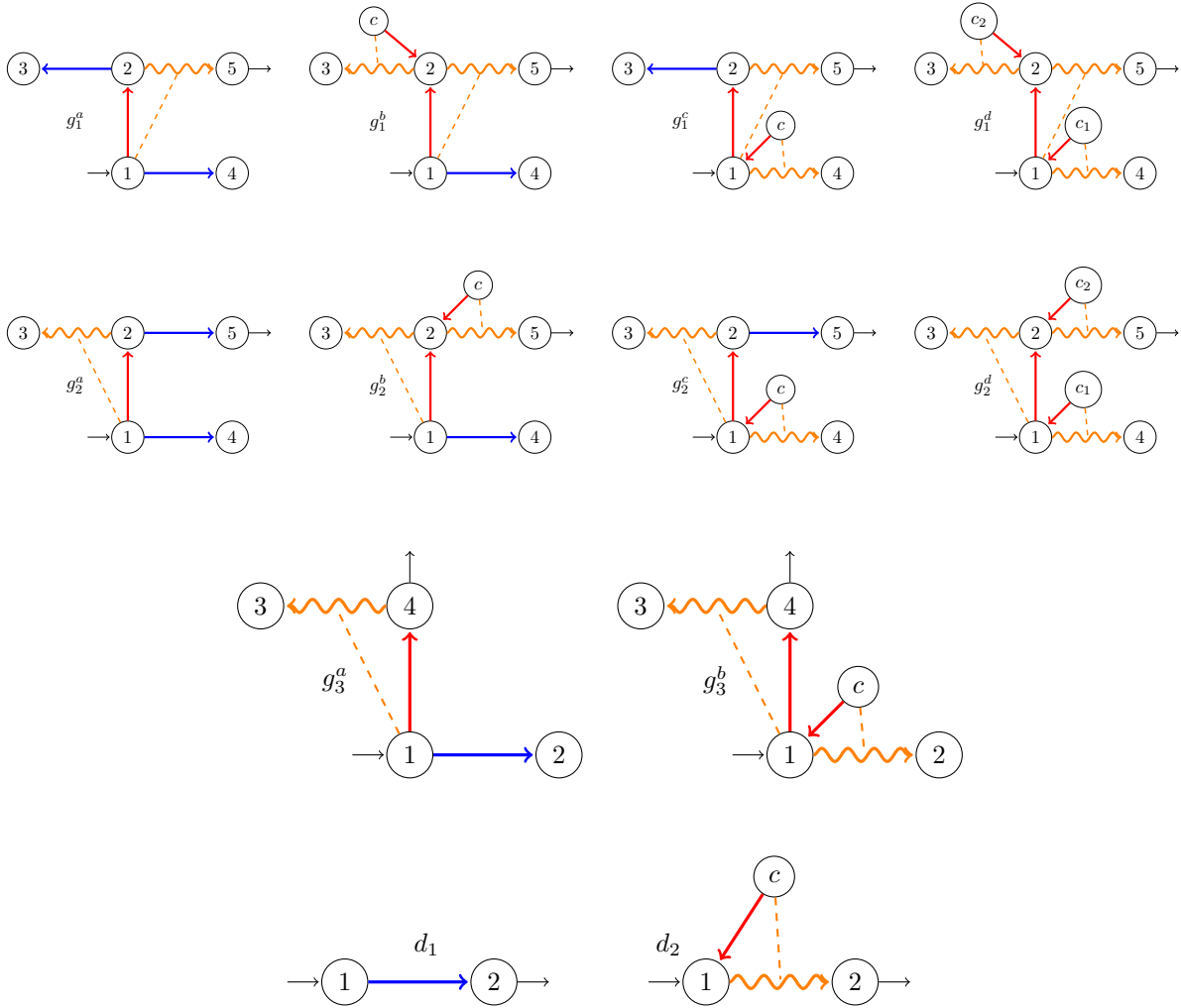

We may turn any fragment string into a \emph{fragment cycle} by identifying (i.e., connecting) its start node with its end node. We use the following symbolic notation for fragment cycles: For a fragment string $x_1x_2\cdots x_{k-1}x_k$ of $k$ fragments, the corresponding fragment cycle will be written symbolically by marking the first fragment $x_1$ as $\dot{x_1}$ and appending $\dot{x_1}$ to the end of the string, i.e., we write $\dot{x_1}x_2\cdots x_{k-1}x_k\dot{x_1}$ to denote the fragment cycle corresponding to fragment string $x_1x_2\cdots x_{k-1}x_k$.

\begin{definition}
We define $\mathcal{GC}$ to be the set of all fragment cycles, i.e., the graphs $\dot{x}g_s\dot{x}$, where $x \in \mathcal{G}$ and $g_s \in \mathcal{GS}$. 
\end{definition}



\subsubsection{Arithmetic Fragment Strings}
A fragment in $\mathcal{G}$ can correspond to many concrete financial systems: The blue and yellow arcs in our fragments represent debt contracts, but their notionals are unspecified. Similarly, the nodes represent banks, but a specification of their external assets is not given. A fragment equipped with coefficients specifying these notionals and external assets will be referred to as an \emph{arithmetic fragment}. Figure \ref{fig:arithgad_example} presents the set of arithmetic versions ofthe fragments in $\mathcal{G}$ that we will be using. We use the notational convention that $x'$ or $x''$ is an arithmetic version of a fragment $x \in \mathcal{G}$. For some of these arithmetic fragments, we have annotated the end nodes with red labels that indicate the assets of the end node under any clearing vector as a function of the recovery rate $r$ of the start node of the fragment.
Note that if the external assets of a bank are set to $0$, our convention is to omit the blue label at the respective node. Furthermore, all nodes labeled with $c, c_1$, and $c_2$ are assumed to have a recovery rate of $0$, which is achieved by setting the external assets of such banks to $0$ and setting the coefficients in the financial system (in which the fragment is embedded) such that $c$ has a strictly positive liability.

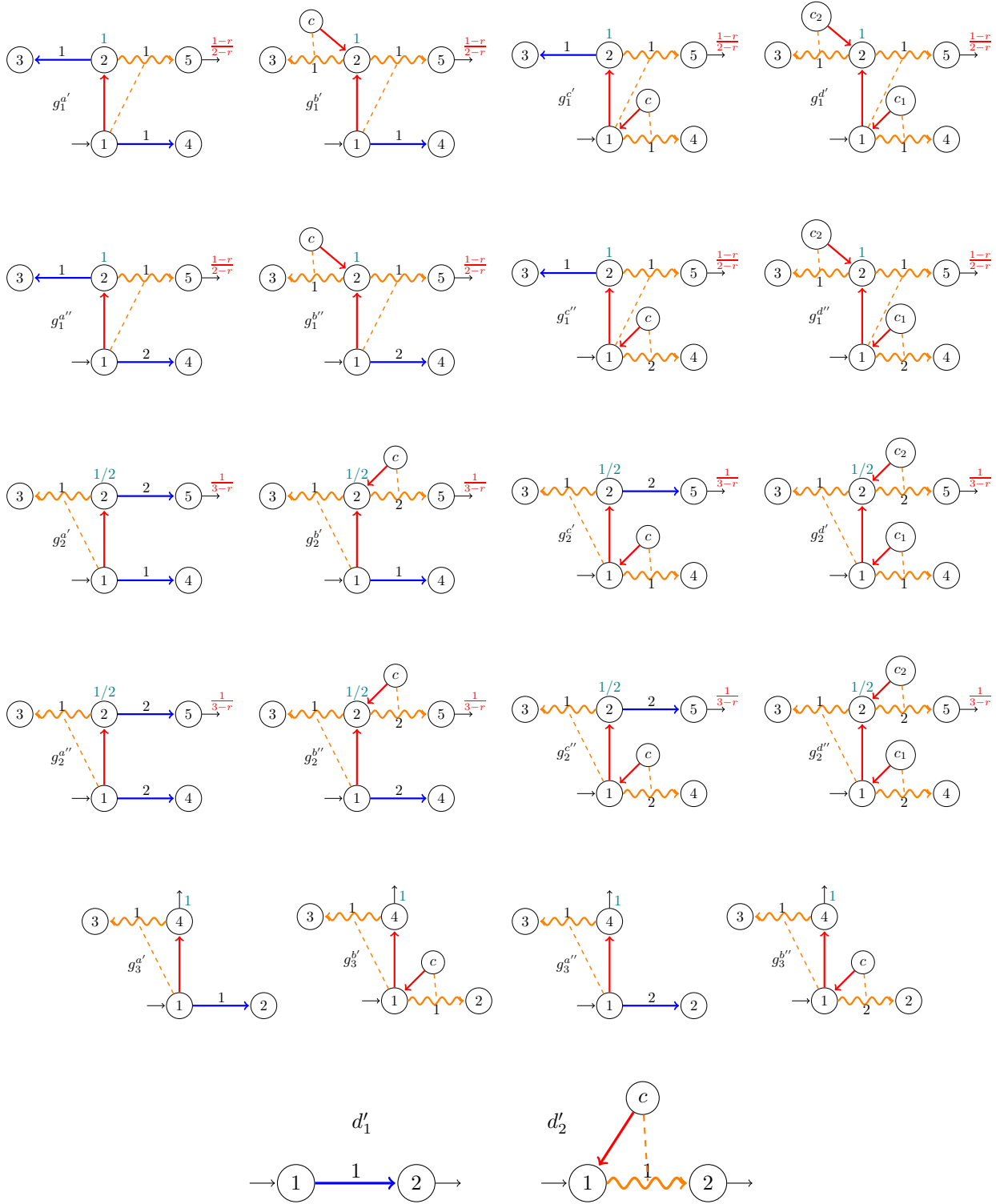
\begin{figure}[htbp!]
    \centering
\scalebox{0.7}{\begin{tikzpicture}
[shorten >=1pt,node distance=2cm,initial text=]
\tikzstyle{every state}=[draw=black!50,very thick]
\tikzset{every state/.style={minimum size=0pt}}
\tikzstyle{accepting}=[accepting by arrow]
\node[state,initial] (1) {$1$};
\node[state] (2)[above of=1] {$2$};
\node[state] (3)[left of=2]{$3$};
\node[state] (4)[right of=1]{$4$};
\node[state,accepting] (5)[right of=2]{$5$};
\draw[blue,->,very thick] (2)--node[midway,black,yshift=2mm]{1}(3);
\draw[orange,->,very thick,snake=snake] (2)--node[midway,black,yshift=2mm]{1}(5);
\draw[red,->,very thick] (1)--(2);
\draw[blue,->,very thick] (1)--node[midway,black,yshift=2mm]{1}(4);
\path [orange,-,draw,dashed,thick] (1) -- ($ (2) !.5! (5) $);
\node[teal,above of=2,yshift=-1.5cm]{1};
\node[red,right of=5,xshift=-1.2cm,yshift=3mm]{$\frac{1-r}{2-r}$};
\node[black,left of=1,xshift=1cm,yshift=1cm]{$g_1^{a'}$};
\end{tikzpicture}
\quad
\begin{tikzpicture}
[shorten >=1pt,node distance=2cm,initial text=]
\tikzstyle{every state}=[draw=black!50,very thick]
\tikzset{every state/.style={minimum size=0pt}}
\tikzstyle{accepting}=[accepting by arrow]
\node[state,initial] (1) {$1$};
\node[state] (2)[above of=1] {$2$};
\node[state] (3)[left of=2]{$3$};
\node[state] (4)[right of=1]{$4$};
\node[state,accepting] (5)[right of=2]{$5$};
\node[state] (6)[above right of=3,xshift=-5mm, yshift=-5mm]{$c$};
\draw[orange,->,very thick,snake=snake] (2)--node[midway,black,yshift=-2mm]{1}(3);
\draw[orange,->,very thick,snake=snake] (2)--node[midway,black,yshift=2mm]{1}(5);
\draw[red,->,very thick] (1)--(2);
\draw[blue,->,very thick] (1)--node[midway,black,yshift=2mm]{1}(4);
\path [orange,-,draw,dashed,thick] (1) -- ($ (2) !.5! (5) $);
\path [orange,-,draw,dashed,thick] (6) -- ($ (2) !.5! (3) $);
\draw[red,->,very thick] (6)--(2);
\node[teal,above of=2,yshift=-1.5cm]{1};
\node[red,right of=5,xshift=-1.2cm,yshift=3mm]{$\frac{1-r}{2-r}$};
\node[black,left of=1,xshift=1cm,yshift=1cm]{$g_1^{b'}$};
\end{tikzpicture}
\quad
\begin{tikzpicture}
[shorten >=1pt,node distance=2cm,initial text=]
\tikzstyle{every state}=[draw=black!50,very thick]
\tikzset{every state/.style={minimum size=0pt}}
\tikzstyle{accepting}=[accepting by arrow]
\node[state,initial] (1) {$1$};
\node[state] (2)[above of=1] {$2$};
\node[state] (3)[left of=2]{$3$};
\node[state] (4)[right of=1]{$4$};
\node[state,accepting] (5)[right of=2]{$5$};
\node[state] (6)[above right of=1,xshift=-5mm,yshift=-5mm]{$c$};
\draw[blue,->,very thick] (2)--node[midway,black,yshift=2mm]{1}(3);
\draw[orange,->,very thick,snake=snake] (2)--node[midway,black,yshift=2mm]{1}(5);
\draw[red,->,very thick] (1)--(2);
\draw[orange,->,very thick,snake=snake] (1)--node[midway,black,yshift=-2mm]{1}(4);
\path [orange,-,draw,dashed,thick] (1) -- ($ (2) !.5! (5) $);
\path [orange,-,draw,dashed,thick] (6) -- ($ (1) !.5! (4) $);
\draw[red,->,very thick] (6)--(1);
\node[teal,above of=2,yshift=-1.5cm]{1};
\node[red,right of=5,xshift=-1.2cm,yshift=3mm]{$\frac{1-r}{2-r}$};
\node[black,left of=1,xshift=1cm,yshift=1cm]{$g_1^{c'}$};
\end{tikzpicture}
\quad
\begin{tikzpicture}
[shorten >=1pt,node distance=2cm,initial text=]
\tikzstyle{every state}=[draw=black!50,very thick]
\tikzset{every state/.style={minimum size=0pt}}
\tikzstyle{accepting}=[accepting by arrow]
\node[state,initial] (1) {$1$};
\node[state] (2)[above of=1] {$2$};
\node[state] (3)[left of=2]{$3$};
\node[state] (4)[right of=1]{$4$};
\node[state,accepting] (5)[right of=2]{$5$};
\node[state] (6)[above right of=1,xshift=-5mm,yshift=-5mm]{$c_1$};
\node[state] (7)[above right of=3,xshift=-5mm,yshift=-5mm]{$c_2$};
\draw[orange,->,very thick,snake=snake] (2)--node[midway,black,yshift=-2mm]{1}(3);
\draw[orange,->,very thick,snake=snake] (2)--node[midway,black,yshift=2mm]{1}(5);
\draw[red,->,very thick] (1)--(2);
\draw[orange,->,very thick,snake=snake] (1)--node[midway,black,yshift=-2mm]{1}(4);
\path [orange,-,draw,dashed,thick] (1) -- ($ (2) !.5! (5) $);
\path [orange,-,draw,dashed,thick] (6) -- ($ (1) !.5! (4) $);
\path [orange,-,draw,dashed,thick] (7) -- ($ (2) !.5! (3) $);
\draw[red,->,very thick] (6)--(1);
\draw[red,->,very thick] (7)--(2);
\node[teal,above of=2,yshift=-1.5cm]{1};
\node[red,right of=5,xshift=-1.2cm,yshift=3mm]{$\frac{1-r}{2-r}$};
\node[black,left of=1,xshift=1cm,yshift=1cm]{$g_1^{d'}$};
\end{tikzpicture}}

\vspace{1cm}

\scalebox{0.7}{\begin{tikzpicture}
[shorten >=1pt,node distance=2cm,initial text=]
\tikzstyle{every state}=[draw=black!50,very thick]
\tikzset{every state/.style={minimum size=0pt}}
\tikzstyle{accepting}=[accepting by arrow]
\node[state,initial] (1) {$1$};
\node[state] (2)[above of=1] {$2$};
\node[state] (3)[left of=2]{$3$};
\node[state] (4)[right of=1]{$4$};
\node[state,accepting] (5)[right of=2]{$5$};
\draw[blue,->,very thick] (2)--node[midway,black,yshift=2mm]{1}(3);
\draw[orange,->,very thick,snake=snake] (2)--node[midway,black,yshift=2mm]{1}(5);
\draw[red,->,very thick] (1)--(2);
\draw[blue,->,very thick] (1)--node[midway,black,yshift=2mm]{2}(4);
\path [orange,-,draw,dashed,thick] (1) -- ($ (2) !.5! (5) $);
\node[teal,above of=2,yshift=-1.5cm]{1};
\node[red,right of=5,xshift=-1.2cm,yshift=3mm]{$\frac{1-r}{2-r}$};
\node[black,left of=1,xshift=1cm,yshift=1cm]{$g_1^{a''}$};
\end{tikzpicture}
\quad
\begin{tikzpicture}
[shorten >=1pt,node distance=2cm,initial text=]
\tikzstyle{every state}=[draw=black!50,very thick]
\tikzset{every state/.style={minimum size=0pt}}
\tikzstyle{accepting}=[accepting by arrow]
\node[state,initial] (1) {$1$};
\node[state] (2)[above of=1] {$2$};
\node[state] (3)[left of=2]{$3$};
\node[state] (4)[right of=1]{$4$};
\node[state,accepting] (5)[right of=2]{$5$};
\node[state] (6)[above right of=3,xshift=-5mm, yshift=-5mm]{$c$};
\draw[orange,->,very thick,snake=snake] (2)--node[midway,black,yshift=-2mm]{1}(3);
\draw[orange,->,very thick,snake=snake] (2)--node[midway,black,yshift=2mm]{1}(5);
\draw[red,->,very thick] (1)--(2);
\draw[blue,->,very thick] (1)--node[midway,black,yshift=2mm]{2}(4);
\path [orange,-,draw,dashed,thick] (1) -- ($ (2) !.5! (5) $);
\path [orange,-,draw,dashed,thick] (6) -- ($ (2) !.5! (3) $);
\draw[red,->,very thick] (6)--(2);
\node[teal,above of=2,yshift=-1.5cm]{1};
\node[red,right of=5,xshift=-1.2cm,yshift=3mm]{$\frac{1-r}{2-r}$};
\node[black,left of=1,xshift=1cm,yshift=1cm]{$g_1^{b''}$};
\end{tikzpicture}
\quad
\begin{tikzpicture}
[shorten >=1pt,node distance=2cm,initial text=]
\tikzstyle{every state}=[draw=black!50,very thick]
\tikzset{every state/.style={minimum size=0pt}}
\tikzstyle{accepting}=[accepting by arrow]
\node[state,initial] (1) {$1$};
\node[state] (2)[above of=1] {$2$};
\node[state] (3)[left of=2]{$3$};
\node[state] (4)[right of=1]{$4$};
\node[state,accepting] (5)[right of=2]{$5$};
\node[state] (6)[above right of=1,xshift=-5mm,yshift=-5mm]{$c$};
\draw[blue,->,very thick] (2)--node[midway,black,yshift=2mm]{1}(3);
\draw[orange,->,very thick,snake=snake] (2)--node[midway,black,yshift=2mm]{1}(5);
\draw[red,->,very thick] (1)--(2);
\draw[orange,->,very thick,snake=snake] (1)--node[midway,black,yshift=-2mm]{2}(4);
\path [orange,-,draw,dashed,thick] (1) -- ($ (2) !.5! (5) $);
\path [orange,-,draw,dashed,thick] (6) -- ($ (1) !.5! (4) $);
\draw[red,->,very thick] (6)--(1);
\node[teal,above of=2,yshift=-1.5cm]{1};
\node[red,right of=5,xshift=-1.2cm,yshift=3mm]{$\frac{1-r}{2-r}$};
\node[black,left of=1,xshift=1cm,yshift=1cm]{$g_1^{c''}$};
\end{tikzpicture}
\quad
\begin{tikzpicture}
[shorten >=1pt,node distance=2cm,initial text=]
\tikzstyle{every state}=[draw=black!50,very thick]
\tikzset{every state/.style={minimum size=0pt}}
\tikzstyle{accepting}=[accepting by arrow]
\node[state,initial] (1) {$1$};
\node[state] (2)[above of=1] {$2$};
\node[state] (3)[left of=2]{$3$};
\node[state] (4)[right of=1]{$4$};
\node[state,accepting] (5)[right of=2]{$5$};
\node[state] (6)[above right of=1,xshift=-5mm,yshift=-5mm]{$c_1$};
\node[state] (7)[above right of=3,xshift=-5mm,yshift=-5mm]{$c_2$};
\draw[orange,->,very thick,snake=snake] (2)--node[midway,black,yshift=-2mm]{1}(3);
\draw[orange,->,very thick,snake=snake] (2)--node[midway,black,yshift=2mm]{1}(5);
\draw[red,->,very thick] (1)--(2);
\draw[orange,->,very thick,snake=snake] (1)--node[midway,black,yshift=-2mm]{2}(4);
\path [orange,-,draw,dashed,thick] (1) -- ($ (2) !.5! (5) $);
\path [orange,-,draw,dashed,thick] (6) -- ($ (1) !.5! (4) $);
\path [orange,-,draw,dashed,thick] (7) -- ($ (2) !.5! (3) $);
\draw[red,->,very thick] (6)--(1);
\draw[red,->,very thick] (7)--(2);
\node[teal,above of=2,yshift=-1.5cm]{1};
\node[red,right of=5,xshift=-1.2cm,yshift=3mm]{$\frac{1-r}{2-r}$};
\node[black,left of=1,xshift=1cm,yshift=1cm]{$g_1^{d''}$};
\end{tikzpicture}}

\vspace{1cm}

\scalebox{0.7}{\begin{tikzpicture}
[shorten >=1pt,node distance=2cm,initial text=]
\tikzstyle{every state}=[draw=black!50,very thick]
\tikzset{every state/.style={minimum size=0pt}}
\tikzstyle{accepting}=[accepting by arrow]
\node[state,initial] (1) {$1$};
\node[state] (2)[above of=1]{$2$};
\node[state] (3)[left of=2]{$3$};
\node[state] (4)[right of=1]{$4$};
\node[state,accepting] (5)[right of=2]{$5$};
\draw[blue,->,very thick] (2)--node[midway,black,yshift=2mm]{2}(5);
\draw[orange,->,very thick,snake=snake] (2)--node[midway,black,yshift=2mm]{1}(3);
\draw[red,->,very thick] (1)--(2);
\draw[blue,->,very thick] (1)--node[midway,black,yshift=2mm]{1}(4);
\path[orange,-,draw,dashed,thick] (1) -- ($ (2) !.5! (3) $);
\node[teal,above of=2,yshift=-1.5cm]{1/2};
\node[red,right of=5,xshift=-1.2cm,yshift=3mm]{$\frac{1}{3-r}$};
\node[black,left of=1,xshift=1cm,yshift=1cm]{$g_2^{a'}$};
\end{tikzpicture}
\quad
\begin{tikzpicture}
[shorten >=1pt,node distance=2cm,initial text=]
\tikzstyle{every state}=[draw=black!50,very thick]
\tikzset{every state/.style={minimum size=0pt}}
\tikzstyle{accepting}=[accepting by arrow]
\node[state,initial] (1) {$1$};
\node[state] (2)[above of=1]{$2$};
\node[state] (3)[left of=2]{$3$};
\node[state] (4)[right of=1]{$4$};
\node[state,accepting] (5)[right of=2]{$5$};
\node[state] (6)[above right of=2,xshift=-5mm,yshift=-5mm]{$c$};
\draw[orange,->,very thick,snake=snake] (2)--node[midway,black,yshift=-2mm]{2}(5);
\draw[orange,->,very thick,snake=snake] (2)--node[midway,black,yshift=2mm]{1}(3);
\draw[red,->,very thick] (1)--(2);
\draw[red,->,very thick] (6)--(2);
\draw[blue,->,very thick] (1)--node[midway,black,yshift=2mm]{1}(4);
\path[orange,-,draw,dashed,thick] (1) -- ($ (2) !.5! (3) $);
\path[orange,-,draw,dashed,thick] (6) -- ($ (2) !.5! (5) $);
\node[teal,above of=2,yshift=-1.5cm]{1/2};
\node[red,right of=5,xshift=-1.2cm,yshift=3mm]{$\frac{1}{3-r}$};
\node[black,left of=1,xshift=1cm,yshift=1cm]{$g_2^{b'}$};
\end{tikzpicture}
\quad
\begin{tikzpicture}
[shorten >=1pt,node distance=2cm,initial text=]
\tikzstyle{every state}=[draw=black!50,very thick]
\tikzset{every state/.style={minimum size=0pt}}
\tikzstyle{accepting}=[accepting by arrow]
\node[state,initial] (1) {$1$};
\node[state] (2)[above of=1]{$2$};
\node[state] (3)[left of=2]{$3$};
\node[state] (4)[right of=1]{$4$};
\node[state,accepting] (5)[right of=2]{$5$};
\node[state] (6)[above right of=1,xshift=-5mm,yshift=-5mm]{$c$};
\draw[blue,->,very thick] (2)--node[midway,black,yshift=2mm]{2}(5);
\draw[orange,->,very thick,snake=snake] (2)--node[midway,black,yshift=2mm]{1}(3);
\draw[red,->,very thick] (1)--(2);
\draw[orange,->,very thick,snake=snake] (1)--node[midway,black,yshift=-2mm]{1}(4);
\draw[red,->,very thick] (6)--(1);
\path[orange,-,draw,dashed,thick] (1) -- ($ (2) !.5! (3) $);
\path[orange,-,draw,dashed,thick] (6) -- ($ (1) !.5! (4) $);
\node[teal,above of=2,yshift=-1.5cm]{1/2};
\node[red,right of=5,xshift=-1.2cm,yshift=3mm]{$\frac{1}{3-r}$};
\node[black,left of=1,xshift=1cm,yshift=1cm]{$g_2^{c'}$};
\end{tikzpicture}
\quad
\begin{tikzpicture}
[shorten >=1pt,node distance=2cm,initial text=]
\tikzstyle{every state}=[draw=black!50,very thick]
\tikzset{every state/.style={minimum size=0pt}}
\tikzstyle{accepting}=[accepting by arrow]
\node[state,initial] (1) {$1$};
\node[state] (2)[above of=1]{$2$};
\node[state] (3)[left of=2]{$3$};
\node[state] (4)[right of=1]{$4$};
\node[state,accepting] (5)[right of=2]{$5$};
\node[state] (6)[above right of=1,xshift=-5mm,yshift=-5mm]{$c_1$};
\node[state] (7)[above right of=2,xshift=-5mm,yshift=-5mm]{$c_2$};
\draw[orange,->,very thick,snake=snake] (2)--node[midway,black,yshift=-2mm]{2}(5);
\draw[orange,->,very thick,snake=snake] (2)--node[midway,black,yshift=2mm]{1}(3);
\draw[red,->,very thick] (1)--(2);
\draw[orange,->,very thick,snake=snake] (1)--node[midway,black,yshift=-2mm]{1}(4);
\draw[red,->,very thick] (6)--(1);
\draw[red,->,very thick] (7)--(2);
\path[orange,-,draw,dashed,thick] (1) -- ($ (2) !.5! (3) $);
\path[orange,-,draw,dashed,thick] (6) -- ($ (1) !.5! (4) $);
\path[orange,-,draw,dashed,thick] (7) -- ($ (2) !.5! (5) $);
\node[teal,above of=2,yshift=-1.5cm]{1/2};
\node[red,right of=5,xshift=-1.2cm,yshift=3mm]{$\frac{1}{3-r}$};
\node[black,left of=1,xshift=1cm,yshift=1cm]{$g_2^{d'}$};
\end{tikzpicture}}

\vspace{1cm}

\scalebox{0.7}{\begin{tikzpicture}
[shorten >=1pt,node distance=2cm,initial text=]
\tikzstyle{every state}=[draw=black!50,very thick]
\tikzset{every state/.style={minimum size=0pt}}
\tikzstyle{accepting}=[accepting by arrow]
\node[state,initial] (1) {$1$};
\node[state] (2)[above of=1]{$2$};
\node[state] (3)[left of=2]{$3$};
\node[state] (4)[right of=1]{$4$};
\node[state,accepting] (5)[right of=2]{$5$};
\draw[blue,->,very thick] (2)--node[midway,black,yshift=2mm]{2}(5);
\draw[orange,->,very thick,snake=snake] (2)--node[midway,black,yshift=2mm]{1}(3);
\draw[red,->,very thick] (1)--(2);
\draw[blue,->,very thick] (1)--node[midway,black,yshift=2mm]{2}(4);
\path[orange,-,draw,dashed,thick] (1) -- ($ (2) !.5! (3) $);
\node[teal,above of=2,yshift=-1.5cm]{1/2};
\node[red,right of=5,xshift=-1.2cm,yshift=3mm]{$\frac{1}{3-r}$};
\node[black,left of=1,xshift=1cm,yshift=1cm]{$g_2^{a''}$};
\end{tikzpicture}
\quad
\begin{tikzpicture}
[shorten >=1pt,node distance=2cm,initial text=]
\tikzstyle{every state}=[draw=black!50,very thick]
\tikzset{every state/.style={minimum size=0pt}}
\tikzstyle{accepting}=[accepting by arrow]
\node[state,initial] (1) {$1$};
\node[state] (2)[above of=1]{$2$};
\node[state] (3)[left of=2]{$3$};
\node[state] (4)[right of=1]{$4$};
\node[state,accepting] (5)[right of=2]{$5$};
\node[state] (6)[above right of=2,xshift=-5mm,yshift=-5mm]{$c$};
\draw[orange,->,very thick,snake=snake] (2)--node[midway,black,yshift=-2mm]{2}(5);
\draw[orange,->,very thick,snake=snake] (2)--node[midway,black,yshift=2mm]{1}(3);
\draw[red,->,very thick] (1)--(2);
\draw[red,->,very thick] (6)--(2);
\draw[blue,->,very thick] (1)--node[midway,black,yshift=2mm]{2}(4);
\path[orange,-,draw,dashed,thick] (1) -- ($ (2) !.5! (3) $);
\path[orange,-,draw,dashed,thick] (6) -- ($ (2) !.5! (5) $);
\node[teal,above of=2,yshift=-1.5cm]{1/2};
\node[red,right of=5,xshift=-1.2cm,yshift=3mm]{$\frac{1}{3-r}$};
\node[black,left of=1,xshift=1cm,yshift=1cm]{$g_2^{b''}$};
\end{tikzpicture}
\quad
\begin{tikzpicture}
[shorten >=1pt,node distance=2cm,initial text=]
\tikzstyle{every state}=[draw=black!50,very thick]
\tikzset{every state/.style={minimum size=0pt}}
\tikzstyle{accepting}=[accepting by arrow]
\node[state,initial] (1) {$1$};
\node[state] (2)[above of=1]{$2$};
\node[state] (3)[left of=2]{$3$};
\node[state] (4)[right of=1]{$4$};
\node[state,accepting] (5)[right of=2]{$5$};
\node[state] (6)[above right of=1,xshift=-5mm,yshift=-5mm]{$c$};
\draw[blue,->,very thick] (2)--node[midway,black,yshift=2mm]{2}(5);
\draw[orange,->,very thick,snake=snake] (2)--node[midway,black,yshift=2mm]{1}(3);
\draw[red,->,very thick] (1)--(2);
\draw[orange,->,very thick,snake=snake] (1)--node[midway,black,yshift=-2mm]{2}(4);
\draw[red,->,very thick] (6)--(1);
\path[orange,-,draw,dashed,thick] (1) -- ($ (2) !.5! (3) $);
\path[orange,-,draw,dashed,thick] (6) -- ($ (1) !.5! (4) $);
\node[teal,above of=2,yshift=-1.5cm]{1/2};
\node[red,right of=5,xshift=-1.2cm,yshift=3mm]{$\frac{1}{3-r}$};
\node[black,left of=1,xshift=1cm,yshift=1cm]{$g_2^{c''}$};
\end{tikzpicture}
\quad
\begin{tikzpicture}
[shorten >=1pt,node distance=2cm,initial text=]
\tikzstyle{every state}=[draw=black!50,very thick]
\tikzset{every state/.style={minimum size=0pt}}
\tikzstyle{accepting}=[accepting by arrow]
\node[state,initial] (1) {$1$};
\node[state] (2)[above of=1]{$2$};
\node[state] (3)[left of=2]{$3$};
\node[state] (4)[right of=1]{$4$};
\node[state,accepting] (5)[right of=2]{$5$};
\node[state] (6)[above right of=1,xshift=-5mm,yshift=-5mm]{$c_1$};
\node[state] (7)[above right of=2,xshift=-5mm,yshift=-5mm]{$c_2$};
\draw[orange,->,very thick,snake=snake] (2)--node[midway,black,yshift=-2mm]{2}(5);
\draw[orange,->,very thick,snake=snake] (2)--node[midway,black,yshift=2mm]{1}(3);
\draw[red,->,very thick] (1)--(2);
\draw[orange,->,very thick,snake=snake] (1)--node[midway,black,yshift=-2mm]{2}(4);
\draw[red,->,very thick] (6)--(1);
\draw[red,->,very thick] (7)--(2);
\path[orange,-,draw,dashed,thick] (1) -- ($ (2) !.5! (3) $);
\path[orange,-,draw,dashed,thick] (6) -- ($ (1) !.5! (4) $);
\path[orange,-,draw,dashed,thick] (7) -- ($ (2) !.5! (5) $);
\node[teal,above of=2,yshift=-1.5cm]{1/2};
\node[red,right of=5,xshift=-1.2cm,yshift=3mm]{$\frac{1}{3-r}$};
\node[black,left of=1,xshift=1cm,yshift=1cm]{$g_2^{d''}$};
\end{tikzpicture}}

\vspace{1cm}

\scalebox{0.7}{\begin{tikzpicture}
[shorten >=1pt,node distance=2cm,initial text=]
\tikzstyle{every state}=[draw=black!50,very thick]
\tikzset{every state/.style={minimum size=0pt}}
\tikzstyle{accepting}=[accepting by arrow]
\node[state,initial] (1) {$1$};
\node[state,accepting above] (4)[above of=1]{$4$};
\node[state] (3)[left of=4]{$3$};
\node[state] (2)[right of=1]{$2$};
\draw[orange,->,very thick,snake=snake] (4)--node[midway,black,yshift=2mm]{1}(3);
\draw[red,->,very thick] (1)--(4);
\draw[blue,->,very thick] (1)--node[midway,black,yshift=2mm]{1}(2);
\path[orange,-,draw,dashed,thick] (1) -- ($ (4) !.5! (3) $);
\node[teal,above of=4,xshift=2mm,yshift=-1.5cm]{1};
\node[black,left of=1,xshift=1cm,yshift=1cm]{$g_3^{a'}$};
\end{tikzpicture}
\quad
\begin{tikzpicture}
[shorten >=1pt,node distance=2cm,initial text=]
\tikzstyle{every state}=[draw=black!50,very thick]
\tikzset{every state/.style={minimum size=0pt}}
\tikzstyle{accepting}=[accepting by arrow]
\node[state,initial] (1) {$1$};
\node[state,accepting above] (4)[above of=1]{$4$};
\node[state] (3)[left of=4]{$3$};
\node[state] (2)[right of=1]{$2$};
\node[state] (6)[above right of=1,xshift=-5mm,yshift=-5mm]{$c$};
\draw[orange,->,very thick,snake=snake] (4)--node[midway,black,yshift=2mm]{1}(3);
\draw[red,->,very thick] (1)--(4);
\draw[orange,->,very thick,snake=snake] (1)--node[midway,black,yshift=-2mm]{1}(2);
\draw[red,->,very thick] (6)--(1);
\path[orange,-,draw,dashed,thick] (1) -- ($ (4) !.5! (3) $);
\path[orange,-,draw,dashed,thick] (6) -- ($ (1) !.5! (2) $);
\node[teal,above of=4,xshift=2mm,yshift=-1.5cm]{1};
\node[black,left of=1,xshift=1cm,yshift=1cm]{$g_3^{b'}$};
\end{tikzpicture}
\quad
\begin{tikzpicture}
[shorten >=1pt,node distance=2cm,initial text=]
\tikzstyle{every state}=[draw=black!50,very thick]
\tikzset{every state/.style={minimum size=0pt}}
\tikzstyle{accepting}=[accepting by arrow]
\node[state,initial] (1) {$1$};
\node[state,accepting above] (4)[above of=1]{$4$};
\node[state] (3)[left of=4]{$3$};
\node[state] (2)[right of=1]{$2$};
\draw[orange,->,very thick,snake=snake] (4)--node[midway,black,yshift=2mm]{1}(3);
\draw[red,->,very thick] (1)--(4);
\draw[blue,->,very thick] (1)--node[midway,black,yshift=2mm]{2}(2);
\path[orange,-,draw,dashed,thick] (1) -- ($ (4) !.5! (3) $);
\node[teal,above of=4,xshift=2mm,yshift=-1.5cm]{1};
\node[black,left of=1,xshift=1cm,yshift=1cm]{$g_3^{a''}$};
\end{tikzpicture}
\quad
\begin{tikzpicture}
[shorten >=1pt,node distance=2cm,initial text=]
\tikzstyle{every state}=[draw=black!50,very thick]
\tikzset{every state/.style={minimum size=0pt}}
\tikzstyle{accepting}=[accepting by arrow]
\node[state,initial] (1) {$1$};
\node[state,accepting above] (4)[above of=1]{$4$};
\node[state] (3)[left of=4]{$3$};
\node[state] (2)[right of=1]{$2$};
\node[state] (6)[above right of=1,xshift=-5mm,yshift=-5mm]{$c$};
\draw[orange,->,very thick,snake=snake] (4)--node[midway,black,yshift=2mm]{1}(3);
\draw[red,->,very thick] (1)--(4);
\draw[orange,->,very thick,snake=snake] (1)--node[midway,black,yshift=-2mm]{2}(2);
\draw[red,->,very thick] (6)--(1);
\path[orange,-,draw,dashed,thick] (1) -- ($ (4) !.5! (3) $);
\path[orange,-,draw,dashed,thick] (6) -- ($ (1) !.5! (2) $);
\node[teal,above of=4,xshift=2mm,yshift=-1.5cm]{1};
\node[black,left of=1,xshift=1cm,yshift=1cm]{$g_3^{b''}$};
\end{tikzpicture}}

\vspace{1cm}

\begin{tikzpicture}
[shorten >=1pt,node distance=2cm,initial text=]
\tikzstyle{every state}=[draw=black!50,very thick]
\tikzset{every state/.style={minimum size=0pt}}
\tikzstyle{accepting}=[accepting by arrow]
\node[state,initial] (1) {$1$};
\node[state,accepting] (2)[right of=1] {$2$};
\draw[blue,very thick,->](1)--node[midway,black,yshift=2mm]{1}(2);
\node[black,above of=1,xshift=1.1cm,yshift=-1.0cm]{$d_1'$};
\end{tikzpicture}
\qquad
\begin{tikzpicture}
[shorten >=1pt,node distance=2cm,initial text=]
\tikzstyle{every state}=[draw=black!50,very thick]
\tikzset{every state/.style={minimum size=0pt}}
\tikzstyle{accepting}=[accepting by arrow]
\node[state,initial] (1) {$1$};
\node[state,accepting] (2)[right of=1] {$2$};
\node[state] (3)[above right of=1,xshift=-5mm]{$c$};
\draw[red,very thick,->](3)--(1);
\draw[orange,very thick,snake=snake,->](1)--node[midway,black,yshift=2mm]{1}(2);
\path[orange,-,draw,dashed,thick] (3) -- ($ (1) !.5! (2) $);
\node[black,left of=1,xshift=1.5cm,yshift=1cm]{$d_2'$};
\end{tikzpicture}
\caption{Our full set of arithmetic fragments. When the external assets of a bank are set to $0$, our notational convention is to omit the blue label at the respective node. All nodes labeled with $c, c_1$, and $c_2$, are assumed to have a recovery rate of $0$, which is achieved by setting the external assets of $c$ to $0$ and setting the coefficients in the financial system (in which the fragment is embedded) such that $c$ has a strictly positive liability.}
    \label{fig:arithgad_example}
\end{figure}




Many of the arithmetic fragments in Figure \ref{fig:arithgad_example} are highly similar, and this has been reflected in their names: For $i \in [3]$ and $j \in \{a,b,c,d\}$, the arithmetic fragment $g_i^{j'}$ differs from $g_i^{j''}$ in only a single contract's notional. 
Similarly to non-arithmetic fragment strings, we may string together arithmetic versions of fragment strings, and we may apply our symbolic notation to denote strings and cycles of arithmetic fragments.

Next, we study the recovery rates of strings and cycles of arithmetic fragments. We start with the following observations.
\begin{observation}\label{obs:1}
Let $x_1'$ and $x_2'$ be any two consecutive arithmetic fragments in a string or cycle $C$ of arithmetic fragments. Let $r$ be the recovery rate of the start node of $x_1'$ under a clearing vector of $C$ under the assumption that all nodes labeled with $c, c_1,$ and $c_2$ have recovery rate $0$. 
\begin{itemize}
    \item If $x_1' \in \{g_i^{j'}, g_i^{j''}\ :\ i \in [2], j \in \{a,b,c,d\}\}$ and $x_2' \in \{g_i^{j'}\ :\ i \in [2], j \in \{a,b,c,d\}\} \cup \{d_1',d_2'\}$, then the recovery rate of the end node of $x_1'$, is $(1-r)/(2-r)$ or $1/(3-r)$, as indicated by the red labels in Figure \ref{fig:arithgad_example}.
    \item If $x_1' \in \{g_3^{j'}, g_3^{j''}\ :\ j \in \{a,b\}\}$ and $x_2' \in \{g_i^{j''} : i \in [3], j \in \{a,b,c,d\}\}$, then the recovery rate of the end node of $x_1'$ is $1/(3-r)$.
\end{itemize}
\end{observation}


Next we provide a notion of \emph{equivalence} among fragment strings.
\begin{definition}
Let $x^1_s$, $x^2_s$ be two arithmetic fragment strings. 
We say that $x^1_s$ and $x^2_s$ are \emph{equivalent} iff the recovery rate of the end node of $x^1_s$ equals the recovery rate of the end node of $x^2_s$ for all possible choices $r \in [0,1]$ of the recovery rate of the input node of $x_1^s$ and $x_2^s$ respectively, under the assumption that all nodes labeled with $c$, $c_1$, and $c_2$ have a recovery rate of $0$.
\end{definition}

Using the notion of equivalence, we provide a set of \emph{rewriting rules} for the symbolic formulations of our arithmetic fragment strings and cycles. We can use such rules to rewrite a given fragment string or cycle into an equivalent one that is simpler to analyse with respect to the output recovery rate. 
By repeatedly applying such rewriting rules, we enable ourselves to produce a family of equivalent reformulations of a given arithmetic fragment string or cycle, such that the clearing recovery rate vector (and in particular the output rate of the rewritten fragment string) can be proved to be irrational. 


The rewriting rules will all be stated under the previously made assumption that nodes labeled with $c$, $c_1$, and $c_2$ in the  fragments all have recovery rate $0$.

\begin{description}
\item[Rule 0:] In any arithmetic fragment string or cycle, we may replace an occurrence of a fragment $g_i^{j'}$, where $i \in [3]$ and $j \in \{a,b,c,d\}$, with the fragment $g_i^{a'}$. Similarly, we may replace an occurrence of a fragment $g_i^{j''}$ with the fragment $g_i^{a''}$. It is straightforward to see that the recovery rate of the end node in the replaced fragment has not changed as a function of the recovery rate of the start node, and therefore the resulting fragment string is equivalent to the original.

\item[Rule 1:] In any arithmetic fragment string or cycle, we may replace an occurrence of a fragment $g_2^{a'}$ (respectively $g_2^{a''}$) by $g_1^{a'} g_1^{a'}$ respectively $g_1^{a''}g_1^{a'}$ if the fragment $g_2^{a'}$ (or respectively $g_2^{a''}$) is followed by one of the fragments in $\{g_1^{a'}, g_2^{a'}, g_3^{a'},d_1',d_2'\}$. 
This rewriting rule is valid because the recovery rate of the end node of $g_2^{a'}$ and $g_2^{a''}$ is equal to $1/(3-r)$, and 
the recovery rate of $g_1^{a'}g_1^{a'}$ is given by 
\begin{equation*}
    \frac{1-\frac{1-r}{2-r}}{2-\frac{1-r}{2-r}} = \frac{1}{3-r}.
\end{equation*}

\item[Rule 2:] In any arithmetic fragment string or cycle, we may replace an occurrence of a consecutive pair of fragments $g_3'g_i^{a''}$, where $g'_3 \in \{g_3^{a'}, g_3^{a''}\}$, and $i \in [3]$, by the fragments $g_2^{a'}g_i^{a'}$. By Observation \ref{obs:1}, the recovery rates of the end nodes of $g_3'$ and $g_2^{a'}$ are identical under this substitution, under any clearing vector, so that the two fragment strings are equivalent. 
\item[Rule 3:] In any arithmetic fragment string or cycle, we may remove an occurrence of $d'_1$ or $d'_2$. This substitution is straightforward from the fact that both $d'_1$ and $d'_2$ just transfer the recovery rate from the start to the end node. 
\end{description}

\begin{example}
Given the fragment cycle $\dot{g_1^a}g_2^b d_1d_2\dot{g_1^a}$, we may choose the coefficients in order to obtain the arithmetic fragment cycle $\dot{g_1^{a'}}g_2^{b'}d'_1d'_2\dot{g_1^{a'}}$. The following two fragment cycles are then equivalent.
\begin{enumerate}
    \item  $\dot{g_1^{a'}}g_2^{a'}d'_1d'_2\dot{g_1^{a}}$ (by applying Rule 0),
    \item  $\dot{g_1^{a'}}g_2^{a'}\dot{g_1^{a'}}$ (by applying Rule 3),
    \item $\dot{g_1^{a'}}g_1^{a'}g_1^{a'}\dot{g_1^{a'}}$ (by applying Rule 1).
\end{enumerate}
\end{example}

\subsection{Irrationality of Strongly Switched Cycles}

Consider any instance of a financial system $I$. Let $G_{I,\text{aux}}$ be its auxiliary graph, and suppose that this system has a strongly switched cycle. Then, this cycle is composed entirely of our fragments in Figure \ref{fig:base}. This is formalised as follows. We will use the notation $V(G)$ and $E(G)$ to denote the set of vertices and the set of arcs of a directed graph $G$, respectively.

\begin{definition}\label{def:agreement}
Let $G'$ be a fragment cycle, and let $C'$ be the unique directed cycle in $G'$. The fragment cycle $G'$ is said to \emph{agree with a cycle $C$ of $G_{I,\text{aux}}$} iff there is a mapping $g : V(G') \rightarrow V(G_{I,\text{aux}})$ with the following properties:
\begin{itemize}
    \item For all $(v,w) \in E(G')$, the arc $(g(v),g(w))$ is in $E(G_{I,\text{aux}})$ and has the same color as $(v,w)$. 
    \item $g$ restricted to the domain $V(C')$ defines a bijection between $V(C')$ and $V(C)$.  
    \item For each CDS $(i,j,R)$ in $G'$, $(g(i),g(j),g(R))$ is a CDS of $G$.  
\end{itemize}
Note that the above points imply that $g$ restricted to $V(C')$ defines an arc-color-preserving isomorphism between $C'$ and $C$. However, this isomorphism property does not necessarily extend to node sets larger than $C'$: nodes in $V(G')$ that are not in $V(C')$ may be mapped by $g$ to the same vertex of $G_{I,\text{aux}}$.

Furthermore, we define the fragment cycle $G'$ to \emph{simply} agree with a cycle $C$ of $G_{I,\text{aux}}$, if $G'$ agrees with cycle $C$ of $G_{I,\text{aux}}$ through a mapping $g$ for which it additionally holds that
\begin{itemize}
    \item all nodes outside $C'$ are mapped to vertices outside $C$,
    \item For every pair of nodes $\{u,v\} \subseteq V(G')$, where $u$ is a node labeled with $c$, $c_1$, or $c_2$ (in Figure \ref{fig:arithgad_example}) and $v$ is labeled with a number (in Figure \ref{fig:arithgad_example}), $g(u) \not= g(v)$, and
    \item for every node $u$ of $G'$ labeled with $c$, $c_1$, or $c_2$, $g(u)$ has an outgoing arc pointing towards a node not in $C'$.  
\end{itemize} 
\end{definition}

The notion of \emph{simple} agreement defined above is a somewhat technical one. Informally stated, it is a mild condition that requires that the neighbouring nodes of $C'$ are sufficiently ``independent'' from each other and from the cycle $C$, under the mapping $g$. This brings us to the definition of a \emph{simple} strongly switched cycle.~\footnote{When compared with the informal statement of our second main theorem in the introduction, this notion makes more precise what we claimed about the off-cycle paths between the nodes in the cycles.}

\begin{definition}
A cycle $C$ of $G_{I,\text{aux}}$ is a \emph{simple} strongly switched cycle iff $C$ is strongly switched, and for each red arc $(u,v)$ of $C$ there are non-red arcs $(u, u')$ and $(v,v')$ such that $u',v' \not\in C$. Furthermore, if $(u,u')$ or $(v,v')$ is orange, then the reference bank $R$ of the corresponding CDS is not in $C$ and $R$ has an outgoing non-red arc pointing to a node not in $C$.
\end{definition}

From the definition of our fragments $\mathcal{G}$, it is straightforward to see that our fragments can represent any strongly switched cycle: If $G_{I,\text{aux}}$ has a strongly switched cycle $C$, then there is a fragment cycle $G'$ consisting of fragments in $\mathcal{G}$ such that $G'$ agrees with $C$ of $G_{I,\text{aux}}$. Similarly, if $G_{I,\text{aux}}$ has a simple strongly switched cycle $C$, then there is a fragment cycle $G'$ consisting of fragments in $\mathcal{G}$ such that $G'$ simply agrees with $C$ of $G_{I,\text{aux}}$. All nodes of $C$ that are switched on correspond to the $2$-labeled nodes of a $g_2^j$ or $g_1^j$ fragment, for some $j \in \{a,b,c,d\}$. 

Next, we prove two lemmas that show that we can set the coefficients in any strongly switched fragment cycle such that the fragment cycle admits only irrational clearing recovery rates. We start by considering formulations consisting only of $g_1^j$ fragments.

\begin{lemma}\label{cla:g1}
For all fragment cycles $C \in \mathcal{CGS}$ consisting of only fragments in $\{g_1^j\ :\ j \in \{a,b,c,d\}\}$, there exist coefficients such that the clearing recovery rate vector of $C$ is irrational (under the assumption that all nodes labeled with $c$,$c_1$, and $c_2$ have a recovery rate of $0$).
\end{lemma}

\begin{proof}
Consider a fragment cycle consisting exclusively of only fragments in $\{g_1^j\ :\ j \in \{a,b,c,d\}\}$. For all $j \in \{a,b,c,d\}$, fix the coefficients of all $g_1^j$ fragments in the cycle to obtain the arithmetic version $g_1^{j'}$. Use rewriting Rule 0 to replace all $g_1^{j'}$ occurrences by $g_1^{a'}$. The resulting arithmetic fragment cycle consists of a number of consecutive copies of $g_1^{a'}$, say $k$ of them. Consider now any clearing vector for the fragment cycle. denoting the $r \in [0,1]$. We prove by induction that the end node of the $i$th fragment has recovery rate equal to $(f_i-rf_{i-2})/(f_{i+2}-rf_i)$, where $f_i$ is the $i$th Fibonacci number, with $f_0 = 0$.

As pointed out in Observation \ref{obs:1}, the recovery rate of the end node of the first fragment equals $\frac{1-r}{2-r}$. Repeating this argument once, we obtain that the recovery rate of the end node of the second fragment equals 
\begin{equation*}
\frac{1-\frac{1-r}{2-r}}{2-\frac{1-r}{2-r}} = \frac{1}{3-r} = \frac{f_2 - f_0r}{f_{4}-f_2r},
\end{equation*}
proving the base case.


Next, assume that the claim holds for the $i$th fragment in the fragment cycle. Taking the recovery rate of $(f_{i}-rf_{i-2})/(f_{i+2}-rf_i)$ as the recovery rate of the start node of fragment $i+1$, we obtain that the recovery rate of its end node is
\begin{eqnarray*}
    \frac{f_{i+2}-rf_i -f_i + rf_{i-2}}{2f_{i+2}-2rf_i -f_i + rf_{i-2}} = \frac{f_{i+1} - rf_{i-1}}{f_i + 2f_{i+1} -rf_{i-2} - 2rf_{i-1}} & = & \frac{f_{i+1} - rf_{i-1}}{f_{i+2} + f_{i+1} - r(f_{i} + f_{i-1})} = \frac{f_{i+1} - rf_{i-1}}{f_{i+3} - rf_{i+1}},
\end{eqnarray*}
which establishes the inductive step.



We know that the end node of the last fragment in the fragment cycle has a recovery rate that coincides with the recovery rate $r$ of the start node of the first fragment. Therefore, in a clearing vector of recovery rates, it holds that $r = \frac{f_n - rf_{k-2}}{f_{k+2}-rf_k}$ which is equivalent to solving the equation $r^2f_k - (f_{k+2}+f_{k-2})r + f_k = 0$. Since $f_{k+2}+f_{k-2} = f_{k+1} + f_k + f_{k-2} = 2f_k +f_{k-1} +f_{k-2} = 3f_k$, computing the recovery rate of the initial node 1 comes down to solving the quadratic equation $r^2-3r+1 = 0$. Solving this equation we obtain that the only solution in $[0,1]$ is $r = (3-\sqrt{5})/2$ which is irrational, thus the clearing recovery rate vector of the strongly switched arithmetic fragment cycle is irrational and is unique.
\end{proof}

The next lemma extends the above to a larger class of arithmetic fragments.
\begin{lemma}\label{lem:2}
For all fragment cycles composed of fragments $\mathcal{G}$ in which every occurrence of a fragment in $\{g_3^j\ :\ j \in \{a,b\}\}$ is followed by a fragment in $\{g_i^j\ :\ i \in [2], j \in \{a,b,c,d\}\}$, there exist coefficients such that the clearing recovery rate vector of $C$ is irrational (under the assumption that all nodes labeled with $c$, $c_1$, and $c_2$ have a recovery rate of $0$).
\end{lemma}

\begin{proof}
Consider any fragment cycle with the property described in the claim. Fix the coefficients of all fragments as follows.
\begin{itemize}
    \item For a fragment $f = g_i^j, i \in [3], j \in \{a,b,c,d\}$ occurring in the cycle, if the fragment preceding it is in $\{g_3^j\ :\ j \in \{a,b\}\}$, turn $f$ into the arithmetic fragment $g_i^{j''}$.
    \item For a fragment $f = g_i^j, i \in [3], j \in \{a,b,c,d\}$ occurring in the cycle, if the fragment preceding it is not in $\{g_3^j\ :\ j \in \{a,b\}\}$, turn $f$ into the arithmetic fragment $g_i^{j'}$.
    \item Turn every occurrence of $d_1$ into $d_1'$, and turn every occurrence of $d_2$ into $d_2'$.
\end{itemize}

Given the resulting arithmetic fragment cycle, we apply rewriting Rule 0 to all fragments in order to obtain a fragment cycle consisting only of arithmetic fragments in $\{g_1^{a'},g_1^{a''},g_2^{a'},g_2^{a''},g_3^{a'},g_3^{a''},d_1',d_2'\}$. We then use Rule 3 to remove all occurrences of $d_1'$ and $d_2'$ from the cycle, followed by Rule $2$ to remove all occurrences of $g_3$ from the cycle, followed by Rule $1$ to remove all occurrences of $g_2$ from the cycle, resulting in an arithmetic fragment cycle consisting of only a sequence of copies of $g_1^a$. Claim \ref{cla:g1} now completes the proof.
\end{proof}

We may now use these last two claims to yield the main result of this section.
\begin{theorem}
Let $I$ be a non-degenerate financial system 
such that $G_{I,\text{aux}}$ has a simple strongly switched cycle. Then there exist rational coefficients for $I$ such that all clearing vectors of $I$ are irrational.
\end{theorem}
\begin{proof}
Let $C$ be a strongly switched cycle of $G_{I,\text{aux}}$ and let $G'$ be a fragment cycle that simply agrees with $C$ through a mapping $g$ satisfying the conditions stated in Definition \ref{def:agreement}. By Lemma \ref{lem:2}, there are notionals $c$ and external assets $e$ for $G'$ such that all clearing vectors of $G'$ are irrational, under the assumption that the nodes labeled with $c$, $c_1$, and $c_2$ have a recovery rate of $0$. In $G_{I,\text{aux}}$, we can now set the notionals and external assets on the vertices and arcs such that they agree with $c$ and $e$ through the mapping $g$: 
\begin{itemize}
    \item For each $v \in V(G')$, we set the external assets of $g(v)$ to $e_v$.
    \item For each $(v,w) \in E(G')$ such that $(v,w)$ is a blue arc, we set the notional on the contract $(g(v),g(w))$ to $c_{v,w}$.
    \item For each $(v,w) \in E(G')$ such that $(v,w)$ is an orange arc, we set the notional on the contract $(g(v),g(w))$ to $c_{v,w}^R$, where $R$ is the reference bank corresponding to the CDS arc $(v,w)$.
\end{itemize}
This assignment of coefficients is well-defined by the properties of $g$ stated in Definition \ref{def:agreement} (i.e., there are no two arcs or vertices that get assigned multiple conflicting coefficients this way).  
We set the remaining coefficients of $G_{I,\text{aux}}$ (i.e., the coefficients on the arcs and vertices outside the image of $g$) as follows:
\begin{itemize}
    \item We set the external assets to $0$ for every node $v \in V(G_{I,\text{aux}})$ that is not in the image of $g$.
    \item We set the notional to $1$ on every arc $(v,w) \in E(G_{I,\text{aux}})$ such that $g^{-1}(v)$ is a node labeled with $c$, and $w$ is not in the image of $g$.
    \item We set the notional to $0$ on every arc $(v,w) \in E(G_{I,\text{aux}})$ that does not satisfy the condition in the point above.
\end{itemize}
Note that by the simplicity property of the agreement between $G'$ and $C$ (see Definition \ref{def:agreement}), the second point in the above list ensures that for each node $v \in V(G')$ labeled with $c$, $c_1$ and $c_2$, the node $g(v)$ has a recovery rate of $0$. Furthermore, if we denote by $G''$ the subgraph of $G$ formed by the image of $g$ (i.e., $G''$ is the projection of $G'$ to $G_{I,\text{aux}}$ through $g$), we can see that the above setting of the coefficients outside of the image of $g$ ensures that no payments flow from $G''$ to any node outside $G''$ under any clearing vector. It then follows by Lemma \ref{lem:2} and the simple agreement properties, that under this setting of the coefficient of $G_{I,\text{aux}}$, every clearing vector is irrational (and in particular these irrational recovery rates emerge in the nodes of $G''$). This establishes our claim.
\end{proof}

\section{Financial Systems with Guaranteed Rational Solutions}\label{sec:rationality}

In the previous section, we identified a sufficient structural condition for the ability of a financial system to have irrational clearing vectors. In this section we investigate how close these conditions are to a characterisation, by attempting to answer the opposite question: Under which structural conditions are rational clearing vectors guaranteed to exist in a financial system? The answer to this relates again to the notion of switched cycles: We will show that if a given non-degenerate financial system does not possess any weakly switched cycle, then there must exist clearing vectors of the system that are rational. We investigate furthermore the computational complexity of finding a clearing vector in this case: Solutions can, informally stated, be computed by solving a linear number of \textsf{PPAD}-complete problems. This latter result is achieved through identifying a natural class of financial systems for which the problem of computing an exact fixed point is \textsf{PPAD}-complete.

The results in this section indicate that the structural conditions for irrationality formulated in the previous section do close in on a characterisation, although there is still a ``gray area'' left: For those instances of financial systems that do have weakly switched cycles, but do not have any simple strongly switched cycles, we are not yet able to determine by the structural interrelationships of the financial contracts whether these systems are likely to possess rational or irrational solutions. This forms an interesting remaining problem that we leave open.

The main result we will prove in this section is thus the following.
\begin{theorem}\label{thm:rationality}
Let $I$ be a non-degenerate financial system. If $G_{I,\text{aux}}$ does not have any weakly switched cycles, then all clearing vectors of $I$ are rational.
\end{theorem}

We start by showing that for a particular subclass of financial systems without weakly switched cycles, the clearing vector computation problem lies in \textsf{Linear-FIXP}, which is equal to \textsf{PPAD} (as per Theorem \ref{prop:linearfixpppad}), and thus the clearing vectors of such financial system must have polynomial size rational coefficients. 

\begin{definition}
An instance $I = (N,e,c)$ of a financial system is said to have the \emph{dedicated CDS debtor property} iff for every node $i \in N$ that is a debtor of at least one CDS of $I$, the following holds: There are no debt contracts (with a non-zero notional) in which $i$ is the debtor, and all CDSes (with a non-zero notional) for which $i$ is the debtor share the same reference bank.
\end{definition}

\begin{lemma}\label{lem:PPADdedicated}
(The exact computation version of) \problem\ restricted to non-degenerate financial systems with the dedicated CDS debtor property is \textsf{PPAD}-complete.
\end{lemma}
\begin{proof}
Let $I = (N,e,c)$ be a non-degenerate financial system with the dedicated CDS debtor property. 
 
Let $D \subseteq N$ be the set of all nodes that are a debtor in at least one CDS of $I$. We first show that containment in \textsf{Linear-FIXP} holds for the special case where for each node in $D$. Finally, we will show that the problem is \text{Linear-FIXP}-hard as well.

First, note that by non-degeneracy of the instance $I$, for all $i \in D$ the reference bank $R$ of which $i$ is the debtor in all its CDS contracts satisfies that $R \in N$. We will show that the following function $f'$ is in \textsf{Linear-FIXP}: The function $f' : \mathbb{R}^{|\mathcal{CDS}|} \times [0,1]^|N\setminus D| \rightarrow \mathbb{R}^{|\mathcal{CDS}|} \times [0,1]^|N\setminus D|$ maps a vector of payments $p$ (one payment for each of the CDS contracts of $I$) combined with a vector of recovery rates $r'$ for the nodes in $N \setminus D$, to vectors of the same dimensions. The fixed points of $f'$ are those points $(p,r')$ for which there is a clearing vector $r$ such that $r'$ coincides with $r$ on $N \setminus D$, and $p$ coincides with the payments through the CDS contracts under $r$. In other words, when comparing our original function $f$ of (\ref{eq:f}) to $f'$, the difference is that $f'$ does not take the recovery rates of $D$ as arguments, but instead takes the payments made by $D$ as arguments. This change of fixed-point function is made in order to ensure that $f'$ can be computed using only the operations $\{+,-,\min\}$ and multiplications by constants. Given a fixed point $(p,r')$ of $f'$, it is then easy to compute a corresponding fixed point of the original $f$, as we can compute the recovery rates of $D$ from their payments $p$, given the recovery rates $r'$ of the remaining nodes. This places the problem in \textsf{Linear-FIXP}, as the class is closed under polynomial-time reductions (see Remark \ref{rem:1}), and hence also in \textsf{PPAD} by Theorem \ref{prop:linearfixpppad}.

The function $f'$ is defined as follows. Let $(p,r')$ denote a vector of payments on the CDS contracts of $I$ together with the recovery rates of $N\setminus D$. For a node $i \in N\setminus D$, $f'$ defines the recovery rate
\begin{equation*}
f'_i(p,r') = \min\left\{1,\frac{a_i(p, r')}{l_i(p, r')}\right\}, \text{ where } l_i(p,r') = \sum_{j \in N} c_{i,j}^{\emptyset}, \text{ and } a_i(p,r') = \sum_{j \in N \setminus D} r_j' c_{j,i}^{\emptyset} + \sum_{j\in D} p_{j,i}. 
\end{equation*}
For a node $i \in D$, if there is a CDS contract with debtor $i$ and creditor $j$, the payment on this contract specified by $f'$ is defined as follows. Let $R \in N\setminus D$ be the unique reference bank of the CDSes in which $i$ is the debtor.
\begin{equation}\label{eq:fprime}
f'_{i,j}(p,r') = \min\left\{1,\frac{a_i(p,r')}{l_i(p,r')}\right\}(1-r_R')c_{i,j}^R = \min\left\{(1-r_R')c_{i,j}^R,(1-r_R')c_{i,j}^R\frac{a_i(p,r')}{l_i(p,r')}\right\},
\end{equation}
where $a_i(p,r')$ are the assets of bank $i$,
\begin{equation*}
    a_i(p,r') = e_i + \sum_{k \in N\setminus D} r'_k c_{k,i}^{\emptyset} + \sum_{k \in D} p_{k,i},
\end{equation*} 
and $l_i(p,r')$ denotes the liabilities of bank $i$, given by:
\begin{equation*}
l_i(p,r') = (1 - r'_R)\sum_{k \in N} c_{i,k}^R.
\end{equation*}
We can therefore rewrite (\ref{eq:fprime}) to
\begin{equation*}
f'_{i,j}(p,r')  = \min\left\{\left(1-r_R'\right)c_{i,j}^R,c_{i,j}^R\frac{e_i + \sum_{k \in N\setminus D} r'_k}{\sum_{k \in N} c_{i,k}^R}\right\},
\end{equation*}
which makes it clear that we can compute $f'_{i,j}(p)$ using $\{+,-,\min\}$ and multiplication-by-constant gates. This places the problem of computing a fixed point for $f'$ (and in turn, for $f$) in \textsf{Linear-FIXP} and \textsf{PPAD}.

For \textsf{PPAD}-hardness, we simply refer to the reduction in \cite{schuldenzucker2017finding}, where it is proved that the weak approximation version of \problem\ is \text{PPAD}-complete: Their proof reduces instances of a known \text{PPAD}-hard problem into a non-degenerate instance of the weak approximation version of \problem, and the latter instance turns out to actually satisfy the dedicated CDS debtor property (where two trivial modifications need to be made to the \emph{amplifier} and \emph{sum} gadgets in the proof in \cite{schuldenzucker2017finding}). This shows \text{PPAD}-hardness for computing a weakly approximate clearing vector in a financial system with the dedicated CDS debtor property. This establishes \textsf{PPAD}-hardness of the exact computation version of the problem as well, since weak approximation trivially reduces to exact computation.
\end{proof}

The above \textsf{PPAD}-completeness result (and more precisely the \textsf{PPAD}-membership part of the result), shows that non-degenerate instances with the dedicated CDS debtor property must have polynomial size rational solutions. We use this fact to prove Theorem \ref{thm:rationality}.

\begin{proof}[Proof of Theorem \ref{thm:rationality}]
Consider the graph $D$ that has as its nodes the strongly connected components (SCCs) of $G_{I,\text{aux}}$, and has an arc from a node $S$ to a node $T$ if and only if there exists an arc in $G_{I,\text{aux}}$ that runs from a node in $S$ to a node in $T$. It is clear that $D$ is a directed acyclic graph.

We will show that we can find a rational clearing vector for $G_{I,\text{aux}}$ by finding rational clearing vectors of the separate SCCs of the system. However, both the assets and the liabilities of the nodes in a given SCC might depend on the contracts from outside the SCC that point into the SCC. Similarly, the liabilities of the nodes in the SCC might depend on arcs pointing from the SCC to external nodes. We may overcome this problem by including the outward-pointing arcs of an SCC into the subinstances that we aim to solve for, and by iterating over the SCCs according to the topological order of $D$: That is, we first find clearing vectors to the set $\mathcal{S}_1$ of SCCs that have no incoming arc in $D$. For such SCCs, the assets and liabilities of the nodes are not influenced by external arcs pointing into the SCC. We subsequently find clearing rates for the set of SCCs $\mathcal{S}_2$ that succeed $\mathcal{S}_1$ in the topological order defined by $D$. In general, we define $\mathcal{S}_j$ inductively as the set of SCCs that directly succeed $\mathcal{S}_{j-1}$ in the topological order defined by $D$, and we iteratively find clearing rates to the set of SCCs $\mathcal{S}_j$, given the clearing rates computed for $\mathcal{S}_1,\ldots,\mathcal{S}_{j-1}$, until we have obtained a clearing vector covering all nodes in the system. A crucial observation that motivates this approach is that the absence of any weakly switched cycle of $G_{I,\text{aux}}$ causes all SCCs to satisfy the dedicated CDS debtor property, and that therefore the clearing vector computation problem considered in each iteration lies in \textsf{PPAD}. At each iteration, we are thus guaranteed that there are rational recovery rates, and finding them requires solving a \textsf{PPAD}-complete problem.

However, there are some details required to make this approach work, which we will address next.

For an SCC $S$ of $G_{I,\text{aux}}$, define $N'(S)$ as the tricoloured subgraph of $G_{I,\text{aux}}$ formed by the set of all arcs that are going out of the nodes of $S$ (i.e., $N'(S)$ consists of $S$ itself and the set of arcs that point from $S$ to nodes outside of $S$). The colouring of the arcs is defined to correspond directly to the colouring in $G_{I,\text{aux}}$. We treat $N'(S)$ as a vertex-labeled and arc-labeled graph where the labels represent the notionals $c$ on the contracts and external assets $e$, according to the specification of $I$, in the usual way.
 
Note that $N'(S)$ does not in general define a valid financial subsystem of $I$, as there may be orange arcs in $N'(S)$ that do not have a corresponding red arc. The reason is that in $G_{I,\text{aux}}$ there may be such red arcs that point into $S$ from outside $S$, and hence are not included in the subgraph $N'(S)$. Another problem is that $N'(S)$ may have red arcs for which the corresponding orange arc is outside $N'(S)$. We will next modify $N'(S)$ accordingly, where we will furthermore allow ourselves to overwrite the external assets of the nodes with other values, as well as the notionals on the aforementioned set of orange arcs. To that end, let $E_o(S)$ be the set of orange arcs in $N'(S)$ for which there is no corresponding red arc present in $N'(S)$ and let $E_r(S)$ be the set of red arcs in $N'(S)$ for which there is no corresponding orange arc present in $N'(S)$. We now define the subgraph $N(S,(e_i)_{i \in V(N(S))},(c_e)_{e \in E_o(S)})$ which is obtained from $N'(S)$ by recolouring every orange arc $e \in E_o(S)$ into a blue arc and replacing its notional by $c_e$, removing every red arc in $E_r(S)$, and by replacing the external assets of every node $i \in V(N(S))$ by $e_i$.

Under this definition of $N(S,e,c)$, for any choice of non-negative vectors $e = (e_i)_{i \in V(N(S))}$ and $c = (c_e)_{e \in E_o(S)}$, we have the property $N(S)$ represents a valid financial system (i.e., where all orange arcs correspond to a CDS that is entirely contained in $N(S,e,c)$). Also, $N(S,e,c)$, has the dedicated CDS debtor property, which follows directly from $G_{I,\text{aux}}$ not having any weakly switched cycles. Therefore, from Lemma \ref{lem:PPADdedicated} it follows that $N(S,e,c)$ has a rational, polynomially sized clearing recovery rate vector, and finding one is \textsf{PPAD}-complete.

The algorithm by which we can find a rational vector of recovery rates now works as follows:
\begin{enumerate}
    \item Compute the recovery rates $r_{\mathcal{S}_1}$ for all  vertices in the SCCs $\mathcal{S}_1$. This is done by finding the recovery rates of the financial systems $N(S,e,c)$, where $S \in\mathcal{S}_1$ and $e$ and $c$ are defined as in $G_{I,\text{aux}}$. 
    \item Iterating over $j$, compute the recovery rates $r_{\mathcal{S}_j}$ given the recovery rates $r_{\mathcal{S}_1},\ldots, r_{\mathcal{S}_{j-1}}$ for the SCCs preceding $\mathcal{S}_j$ (according to the topological order defined by $j$). This is done by finding the recovery rates of the financial systems $N(S,e',c')$, where $S \in \mathcal{S}_j$, and the coefficients $e'$ and $c'$ are defined as follows:
    \begin{itemize}
        \item The external assets $e_i'$, for $i \in S$, is the sum of the original external assets $e_i$ as defined in instance $I$, and all the payments that are made to $i$ by nodes in the SCCs preceding $S$, under the recovery rates $r_{\mathcal{S}_1},\ldots, r_{\mathcal{S}_{j-1}}$ computed so far. 
        \item The notional $c'_{i,j}$ on the blue arc $(i,j) \in  E_o(S)$ of $N(S,e',c')$ (that is coloured orange in  $G'_{I,\text{aux}}$), is set to $(1-(r_{\mathcal{S}_{j-1}})_R)c_{i,j}^R$. Here $R$ denotes the reference bank of the CDS of $I$ that the orange arc $(i,j)$ of $G_{I,\text{aux}}$ associates to. 
    \end{itemize}
\end{enumerate}
Let $r$ be the resulting vector of recovery rates for $G_{I,\text{aux}}$, i.e., $r$ is obtained by combining the recovery rates $r_{\mathcal{S}_1},r_{\mathcal{S}_2}, \ldots$ computed through the above procedure. It can be proved inductively that $r$ is a clearing vector, by showing that for all nodes $i$, it holds that $f_i(r) = r$:

For an SCC $S \in \mathcal{S}_1$, all coefficients and arc colours in $N(S,e,c)$ correspond to those in $G_{I,\text{aux}}$.
Therefore, for every node in $i \in V(S)$ it holds that $a_i(r)$ under $G_{I,\text{aux}}$ equals $a_i(r_{\mathcal{S}_1})$ under $N(S,e,c)$. Similarly $l_i(r)$ under $G_{I,\text{aux}}$ equals  $l_i(r_{\mathcal{S}_1})$ under $N(S,e,c)$. Combining this with the fact that $f_i(r_{\mathcal{S}_1}) = (r_{\mathcal{S}_1})_i$ under $N(S,e,c)$, we conclude that it also holds that $f_i(r) = r$ under $G_{I,\text{aux}}$.

Next, we prove that under $G_{I,\text{aux}}$ it also holds that $f_i(r) = r$ for nodes in SCCs of $\mathcal{S}_j$, under the induction hypothesis that $f_i(r) = r$ for all nodes in SCCs of $\mathcal{S}_1 \cup \mathcal{S}_2 \cup \cdots$. Consider an SCC $S \in \mathcal{S}_j$ and a node $i \in V(S)$. First, we will establish that the payment along any arc of $N(S,e',c')$ under $r_{\mathcal{S}_j}$, equals the payment along the corresponding arc of $G_{I,\text{aux}}$ under $r$: This holds because (i.) all blue arcs in $N(S,e',c')$ that are not in $E_o(S)$ and are coming into $i$ have the same notional as in $G_{I,\text{aux}}$, so that the payment made through such an arc is equal among both $N(S,e',c')$ and $G_{I,\text{aux}}$; (ii.) all orange arcs in $N(S,e',c')$ that associate to CDSes for which the reference bank is in $S$ similarly have the same notional as in $G_{I,\text{aux}}$, so that the payment made through such an arc is equal among both $N(S,e',c')$ and $G_{I,\text{aux}}$ as well; (iii.) all blue arcs $(i,i')$ in $N(S,e',c')$ that are in $E_o(S)$ are orange arcs under $G_{I,\text{aux}}$, however, the notional on this arc in $N(S,e',c')$ is set to $(1-(r_{\mathcal{S}_{j-1}})_R)c_{i,i'}^R = (1-r_R)c_{i,i'}^R$, where $R$ is the reference bank of the CDS of which the orange arc $(i,i')$ is part. Thus, the notional on the blue arc $(i,i')$ under $N(S,e',c')$ equals the liability on the orange arc $(i,i')$ under $G_{I,\text{aux}}$ and $r$. Therefore, the payment going through the blue arc $(i,i')$ of $N(S,e',c')$ under $r_{\mathcal{S}_j}$ equals the payment going through the orange arc $(i,i')$ of $G_{I,\text{aux}}$ under $r$. 

This correspondence among payments in the two financial systems $N(S,e',c')$ (with recovery rates $r_{\mathcal{S}_j}$ and $G_{I,\text{aux}}$ (with recovery rates $r$) implies that $a_i(r)$ under $G_{I,\text{aux}}$ equals $a_i(r_{\mathcal{S}_j}$ under $N(S,e',c')$. Similarly $l_i(r)$ under $G_{I,\text{aux}}$ equals $l_i(r_{\mathcal{S}_j}$ under $N(S,e',c')$. Combining this with the fact that $f_i(r_{\mathcal{S}_j}) = (r_{\mathcal{S}_j})_i$ under $N(S,e',c')$, we conclude that it also holds that $f_i(r) = r$ under $G_{I,\text{aux}}$.
\end{proof}

The procedure outlined in the proof of Theorem \ref{thm:rationality} requires solving a \textsf{PPAD}-complete problem in each iteration, and the number of such iterations is at most linear in the instance size. Since solving each of these problems in \textsf{PPAD} yields a rational solution of size polynomial in the input, one might be tempted to think that the procedure in its entirety is capable of finding a polynomial size rational solution for any financial system that has no weakly switched cycles. Unfortunately, the latter is not true: Observe that in each iteration of the procedure, the \textsf{PPAD}-complete problem instance that is solved, is actually constructed using the rational recovery rate vectors that are computed in the preceding iterations. The coefficients in the \textsf{PPAD}-complete problem instance that is to be solved in any given iteration, are thus polynomially sized in the output recovery rates of the previous iteration. Altogether, this means that the coefficient sizes potentially grow by a polynomial factor in each iteration, and that the final recovery rates output by the procedure are potentially of exponential size.  

Indeed, there are examples of financial systems without weakly switched cycles for which the rational clearing vector has recovery rates that require an exponential number of bits to write down. A simple example is obtained by taking some of the gadgets in the reduction used in our \textsf{FIXP}-completeness result (Theorem \ref{thm:fixp}). By taking a duplication gadget (Figure \ref{fig:duplicationgadget}) followed by a multiplication gadget (Figure \ref{fig:multiplicationgadget} or \ref{fig:nondegeneratemultiplicationgadget}) that is connected to the two output nodes of the duplication gadget. We may then take multiple copies of these, and chain them together to form an acyclic financial system. If we now give the first node in the chain (i.e., the input node of the first duplication gadget) some small amount of positive external assets $c < 1$, this acyclic financial system essentially performs a sequence of successive squaring operations on the number $c$, under the unique clearing vector. The resulting recovery rates on the output nodes of the multiplication gadgets are then doubly exponentially small in magnitude, with respect to the number of squaring repetitions. Thus, the resulting clearing recovery rates require a number of bits that is exponential in the size of the financial system. 

If one is willing to discard the complexity issues that arise from working with large-size rational numbers, it is possible to study the procedure in the proof of Theorem \ref{thm:rationality} in the Blum-Shub-Smale model of computation. Under this computational model, any real number takes one unit of space to store, regardless of its size. Moreover, standard arithmetic operations are assumed to take unit time.\footnote{For a formal and more accurate definition of the Blum-Shub-Smale model, see the book \cite{blum1998complexity}.} The proof of Theorem \ref{thm:rationality} then implies that when one has oracle access to \textsf{PPAD}, it is possible to find rational clearing vectors in polynomial time under this model of computation. The class of problems polynomial time solvable under the Blum-Shub-Smale model is commonly denoted by $\textsf{P}_{\mathbb{R}}$. Hence, we obtain the following corollary.
\begin{corollary}\label{cor:rational}
The exact computation version of \problem, restricted to instances without weakly switched cycles, is in the complexity class $\textsf{P}_{\mathbb{R}}^{\textsf{PPAD}}$.
\end{corollary}

\section{Conclusions}\label{sec:conclusions}
In this paper we study two questions of significance related to the systemic risk in financial networks with CDSes, a widely used and potentially disruptive class of financial derivatives. Firstly, we settle the computational complexity of computing strong approximations of each bank's exposure to systemic risk, arguably the only notion of approximation of interest to industry -- a conceptual point so far overlooked in the literature. We show that this problem is \textsf{FIXP}-complete. Secondly, we initiate the study of the rational fragment of \textsf{FIXP} by studying the conditions under which rational solutions for \problem\ exist.  Our results here are not conclusive in that there is a gap between our necessary and sufficient conditions, the cycles which involve both switched on and switched off nodes being not fully understood.~\footnote{We regard the \emph{simplicity} condition we made (i.e., about off-cycle paths between cycle nodes) as a technicality, which is less interesting and likely somewhat easier to deal with.} We conjecture that for any network with a weakly switched cycle there exist rational values for assets and liabilities that lead to irrational solutions (Appendix \ref{apx:Counterexample} contains one such example); however, our arguments and scheme cannot be easily generalised to those instances. We leave providing a full characterisation as an open problem.

There are further research directions that are suggested by our work. It would be interesting to study whether Corollary \ref{cor:rational}'s connection between \problem, \textsf{PPAD}, and $\textsf{P}_{\mathbb{R}}$ (i.e. polynomial time under the Blum-Shub-Smale model of computation \cite{blum1998complexity}) holds more generally for the entire rational subset of problems in \textsf{FIXP}. Furthermore, it is interesting to pursue finding polynomial-time \emph{constant} approximation algorithms of clearing recovery rate vectors: Also from an applied point of view, achieving a good approximation factor here (say with an additive approximation term of $1/100$)  might yield solutions that are useful in most practical circumstances and could be considered acceptable by financial institutions. We note here that a $1/2$-strong approximation is easy to compute (a recovery rate vector of only $1/2$s would indeed suffice). 


\bibliographystyle{plain}
\bibliography{bibl.bib}

\begin{thebibliography}{10}

\bibitem{EUBusiness}
Euro-parliament bans 'naked' credit default swaps. {{EUB}}usiness.
\newblock \url{https://www.eubusiness.com/news-eu/finance-economy-cds.dij}.

\bibitem{IMFBlogs}
Lasting effects: The global economic recovery 10 years after the crisis.
  {{IMFB}}logs.
\newblock
  \url{https://blogs.imf.org/2018/10/03/lasting-effects-the-global-economic-recovery-10-years-after-the-crisis/}.

\bibitem{acemoglu2015systemic}
Daron Acemoglu, Asuman Ozdaglar, and Alireza Tahbaz-Salehi.
\newblock Systemic risk and stability in financial networks.
\newblock {\em American Economic Review}, 105(2):564--608, 2015.

\bibitem{BertschingerHS20}
Nils Bertschinger, Martin Hoefer, and Daniel Schmand.
\newblock Strategic payments in financial networks.
\newblock In Thomas Vidick, editor, {\em 11th Innovations in Theoretical
  Computer Science Conference, {ITCS} 2020, January 12-14, 2020, Seattle,
  Washington, {USA}}, volume 151 of {\em LIPIcs}, pages 46:1--46:16. Schloss
  Dagstuhl - Leibniz-Zentrum f{\"{u}}r Informatik, 2020.

\bibitem{blum1998complexity}
Lenore Blum, Felipe Cucker, Michael Shub, and Steve Smale.
\newblock {\em Complexity and real computation}.
\newblock Springer Science \& Business Media, 1998.

\bibitem{cifuentes2005liquidity}
Rodrigo Cifuentes, Gianluigi Ferrucci, and Hyun~Song Shin.
\newblock Liquidity risk and contagion.
\newblock {\em Journal of the European Economic association}, 3(2-3):556--566,
  2005.

\bibitem{eisenberg2001systemic}
Larry Eisenberg and Thomas~H Noe.
\newblock Systemic risk in financial systems.
\newblock {\em Management Science}, 47(2):236--249, 2001.

\bibitem{elliott2014financial}
Matthew Elliott, Benjamin Golub, and Matthew~O Jackson.
\newblock Financial networks and contagion.
\newblock {\em American Economic Review}, 104(10):3115--53, 2014.

\bibitem{etessami2010complexity}
Kousha Etessami and Mihalis Yannakakis.
\newblock On the complexity of nash equilibria and other fixed points.
\newblock {\em SIAM Journal on Computing}, 39(6):2531--2597, 2010.

\bibitem{filos2021complexity}
Aris Filos-Ratsikas, Yiannis Giannakopoulos, Alexandros Hollender, Philip
  Lazos, and Diogo Po{\c{c}}as.
\newblock On the complexity of equilibrium computation in first-price auctions.
\newblock {\em arXiv preprint arXiv:2103.03238}, 2021.

\bibitem{filos2021fixp}
Aris Filos-Ratsikas, Kristoffer~Arnsfelt Hansen, Kasper H{\o}gh, and Alexandros
  Hollender.
\newblock Fixp-membership via convex optimization: Games, cakes, and markets.
\newblock {\em arXiv preprint arXiv:2111.06878}, 2021.

\bibitem{DBLP:conf/stoc/GargMVY17}
Jugal Garg, Ruta Mehta, Vijay~V. Vazirani, and Sadra Yazdanbod.
\newblock Settling the complexity of leontief and {PLC} exchange markets under
  exact and approximate equilibria.
\newblock In Hamed Hatami, Pierre McKenzie, and Valerie King, editors, {\em
  Proceedings of the 49th Annual {ACM} {SIGACT} Symposium on Theory of
  Computing, {STOC} 2017, Montreal, QC, Canada, June 19-23, 2017}, pages
  890--901. {ACM}, 2017.

\bibitem{glasserman2015likely}
Paul Glasserman and H~Peyton Young.
\newblock How likely is contagion in financial networks?
\newblock {\em Journal of Banking \& Finance}, 50:383--399, 2015.

\bibitem{goldberg2021hairy}
Paul~W Goldberg and Alexandros Hollender.
\newblock The hairy ball problem is ppad-complete.
\newblock {\em Journal of Computer and System Sciences}, 2021.

\bibitem{heise2012derivatives}
Sebastian Heise and Reimer K{\"u}hn.
\newblock Derivatives and credit contagion in interconnected networks.
\newblock {\em The European Physical Journal B}, 85(4):1--19, 2012.

\bibitem{hemenway2016sensitivity}
Brett Hemenway and Sanjeev Khanna.
\newblock Sensitivity and computational complexity in financial networks.
\newblock {\em Algorithmic Finance}, 5(3-4):95--110, 2016.

\bibitem{hu2012network}
Daning Hu, J~Leon Zhao, Zhimin Hua, and Michael~CS Wong.
\newblock Network-based modeling and analysis of systemic risk in banking
  systems.
\newblock {\em MIS quarterly}, pages 1269--1291, 2012.

\bibitem{kleinberg2006algorithm}
Jon Kleinberg and Eva Tardos.
\newblock {\em Algorithm design}.
\newblock Pearson Education India, 2006.

\bibitem{DBLP:books/daglib/0072413}
Christos~H. Papadimitriou.
\newblock {\em Computational complexity}.
\newblock Addison-Wesley, 1994.

\bibitem{papadimitriou}
Christos~H. Papadimitriou.
\newblock On the complexity of the parity argument and other inefficient proofs
  of existence.
\newblock {\em Journal of Computer and System Sciences}, 48(3):498--532, 1994.

\bibitem{papp2020default}
P{\'a}l~Andr{\'a}s Papp and Roger Wattenhofer.
\newblock Default ambiguity: finding the best solution to the clearing problem.
\newblock {\em arXiv preprint arXiv:2002.07741}, 2020.

\bibitem{PappW20a}
P{\'{a}}l~Andr{\'{a}}s Papp and Roger Wattenhofer.
\newblock Network-aware strategies in financial systems.
\newblock In Artur Czumaj, Anuj Dawar, and Emanuela Merelli, editors, {\em 47th
  International Colloquium on Automata, Languages, and Programming, {ICALP}
  2020, July 8-11, 2020, Saarbr{\"{u}}cken, Germany (Virtual Conference)},
  volume 168 of {\em LIPIcs}, pages 91:1--91:17. Schloss Dagstuhl -
  Leibniz-Zentrum f{\"{u}}r Informatik, 2020.

\bibitem{rogers2013failure}
Leonard~CG Rogers and Luitgard~AM Veraart.
\newblock Failure and rescue in an interbank network.
\newblock {\em Management Science}, 59(4):882--898, 2013.

\bibitem{schuldenzucker2016clearing}
Steffen Schuldenzucker, Sven Seuken, and Stefano Battiston.
\newblock Clearing payments in financial networks with credit default swaps.
\newblock In {\em Proceedings of the 2016 ACM Conference on Economics and
  Computation}, pages 759--759, 2016.

\bibitem{schuldenzucker2017finding}
Steffen Schuldenzucker, Sven Seuken, and Stefano Battiston.
\newblock Finding clearing payments in financial networks with credit default
  swaps is {PPAD}-complete.
\newblock {\em LIPIcs: Leibniz International Proceedings in Informatics}, 67,
  2017.

\end{thebibliography}

\appendix

\newpage

\section{Weak Fixed Points versus Exact Fixed Points}\label{apx:A}
\begin{example}\label{ex:3}

\begin{figure}[htbp!]
\centering
\begin{tikzpicture}[shorten >=1pt,node distance=2cm,initial text=]
\tikzstyle{every state}=[draw=black!50,very thick]
\tikzset{every state/.style={minimum size=0pt}}
\node[state] (1) {$1$}; 
\node[teal,left of=1,xshift=1.5cm]{1};
\node[state] (2) [right of=1] {$2$};
\node[state] (3) [right of=2] {$3$};
\draw[orange,very thick,->,snake=snake] (1)--node[midway,black,yshift=3mm]{1}(2);
\draw[blue,very thick,->] (2)--node[midway,black,yshift=3mm]{1/2}(3);
\node[state] (4) [below of=1] {$4$};
\node[teal,left of=4,xshift=1.5cm]{1};
\node[state] (5) [below of=2] {$5$};
\node[state] (6) [below of=3] {$6$};
\draw[blue,very thick,->] (5)--node[midway,black,yshift=3mm]{$4\epsilon$}(6);
\draw[orange,very thick,->,snake=snake] (4)--node[midway,black,yshift=3mm,xshift=-1mm]{1}(5);
\path [orange,->,draw,dashed,thick] (5) -- ($ (1) !.5! (2) $);
\path [orange,->,draw,dashed,thick] (2) -- ($ (4) !.5! (5) $);
\end{tikzpicture}
\caption{Financial system of Example \ref{ex:3}}
    \label{fig:3}
\end{figure}
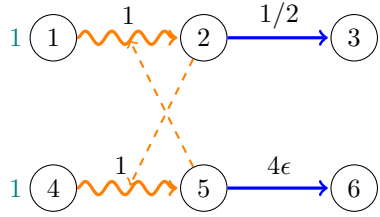

Consider the financial system depicted in Figure \ref{fig:3}.
Assume that $e_1 = 1$, $e_2 = 0$, $e_4 = 1$, $e_5 = 0$. Also $c_{1,2}^5 = c_{4,5}^2 = 1$, $c_{2,3} = 1/2$, $c_{5,6} = 4\epsilon$, for a tiny $\epsilon > 0$.

Take the point $r = (1,1,1,1,0,1)$. Obviously, $f_1(r) = f_4(r) = 1$, and $f_2(r) = \min\{1,(1 - r_5)/(1/2)\} = 1$, $f_5(r) = \min\{1,(1-r_2)(4\epsilon)\} = 0$. Since $r = f(r)$, the point
$r = (1,1,1,1,0,1) $ is an exact fixed point.       

Assume $r' = (1,1-2\epsilon,1,1,1/2+ \epsilon,1)$. It holds that $f_2(r') = \min\{1,(1-r'_5)/(1/2)\} = \min\{1,(1/2 - \epsilon)/(1/2)\} = 1-2\epsilon$. Moreover, $f_5(r') = \min\{1,(1-r'_2)/4\epsilon\} = \min\{1,(1-1+2\epsilon)/4\epsilon\} = 1/2$. Thus, $\|r' -f(r')\|_{\infty} \leq \epsilon$ which constitutes  $r' = (1,1-2\epsilon,1,1,1/2 + \epsilon,1)$ a weakly $\epsilon$-approximate fixed point.

When comparing $r'$ to $r$ though, we observe that the distance from the exact fixed point is $\|(1,1,1,1,0,1) - (1,1-2\epsilon,1,1,1/2+\epsilon,1)\|_{\infty} > 1/2$, proving the claim that  weakly $\epsilon$-approximate fixed points may potentially be very far from exact fixed points. On top of that, the information we get on node 2 in $r'$ is that $r'_2  = 1-2\epsilon < 1$, meaning that 2 is in default whereas actually 2 can fully pay its liabilities since $r_2 = 1$. This indicates that a weakly $\epsilon$-approximate fixed point may contain misleading information about whether a bank is in default or not, thus making such approximations unreliable while motivating the study of strong approximations.
\end{example}

\section{Financial System Gadgets Used in the Proof of Theorem \ref{thm:fixp}}\label{apx:gadgets}
In this section we present our gadgets, in Figures \ref{fig:duplicationgadget}--\ref{fig:absdifferencegadget}. For simplicity, (some of) our gadgets lead to degenerate instances; we deal with this in the next subsection.

\begin{figure}[htbp!]

    \centering
\begin{tikzpicture}
[shorten >=1pt,node distance=2cm,initial text=]
\tikzstyle{every state}=[draw=black!50,very thick]
\tikzset{every state/.style={minimum size=0pt}}
\tikzstyle{accepting}=[accepting by arrow]
\node[state,initial] (1) {$r$};
\node[state]         (2) [right of=1] {$1$};
\node[state,accepting]         (8) [right of=2]{$r$};
\draw[blue,very thick,->](2)--node[midway,black,yshift=2mm]{1}(8);
\draw[blue,->,very thick] (1)--node[midway,black,yshift=2mm]{1}(2);

\node[state] (3) [below right of=1]{$2$};
\node[teal,left of=3,xshift=1.5cm]{1};
\node[state] (4) [right of=3]{$3$};

\draw[orange,very thick,->,snake=snake] (3)--node[midway,black,xshift=-3mm,yshift=3mm]{1}(4);
\path[orange,-,draw,dashed,thick] (2) -- ($ (3) !.5! (4) $);

\node[state](7)[below right of=3]{$5$};
\node[teal,left of=7,xshift=1.5cm]{$c$};
\node[state](8)[right of=7]{$6$};
\node[state,accepting](10)[right of=8]{$c \cdot r$};
\draw[blue,very thick,->](8)--node[midway,black,yshift=2mm]{1}(10);
\node[state](9)[right of=4]{$4$};
\draw[blue,->,very thick] (4)--node[midway,black,yshift=2mm,xshift=-1mm]{1}(9);
\draw[orange,very thick,->,snake=snake](7)--node[midway,black,yshift=2mm,xshift=-4mm]{$c$}(8);
\path[orange,-,draw,dashed,thick] (4) -- ($ (7) !.5! (8) $);

\end{tikzpicture}
\caption{Gadget $g_{\text{dup}}$: This gadget outputs two values: $r$ and $cr$, where $c \in [0,1]$ is a rational constant. Choosing $c = 1$ yields a duplication gadget that takes the input recovery rate $r$ and outputs $r$ as the two recovery rates of the output nodes. We may also denote this gadget by $g_{c\cdot r}$, for $c \in [0,1]$ in case this gadget is used to multiply the input by a constant $c$.}
    \label{fig:duplicationgadget}
\end{figure}
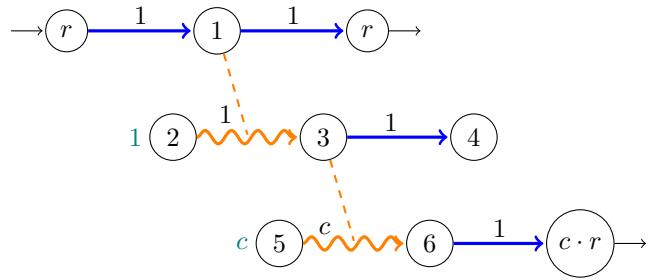

\begin{figure}[htbp!]

      \centering
\begin{tikzpicture}
[shorten >=1pt,node distance=2cm,initial text=]
\tikzstyle{every state}=[draw=black!50,very thick]
\tikzset{every state/.style={minimum size=0pt}}
\tikzstyle{accepting}=[accepting by arrow]

\node[state,initial] (1){$r_2$};
\node[state]         (2)[right of=1]{$1$};
\node[state]         (16)[right of=2]{$2$};
\draw[blue,very thick,->] (1)--node[midway,black,yshift=2mm]{1}(2);
\draw[blue,very thick,->] (2)--node[midway,black,yshift=2mm]{1}(16);

\node[state] (3)[below right of =1]{$3$};
\node[state] (4)[right of=3]{$4$};
\node[state] (5)[right of=4]{$5$};
\node[teal,left of =3,xshift=1.5cm]{1};
\draw[orange,very thick,->,snake=snake](3)--node[midway,black,xshift=1mm,yshift=3mm]{1}(4);
\draw[blue,very thick,->](4)--node[midway,black,yshift=2mm]{1}(5);
\path [orange,-,draw,dashed,thick] (2) -- ($ (3) !.5! (4) $);

\node[state,initial] (6)[below left of=3]{$r_1$};
\node[state]         (7)[right of=6]{$6$};
\node[state]         (8)[right of=7]{$7$};
\node[state]         (9)[right of=8]{$8$};
\draw[blue,very thick,->](6)--node[midway,black,yshift=2mm]{1}(7);
\draw[orange,very thick,->,snake=snake](7)--node[midway,black,xshift=-1mm,yshift=2mm]{1}(8);
\draw[blue,very thick,->](8)--node[midway,black,yshift=2mm]{1}(9);
\path[orange,-,draw,dashed,thick] (4) -- ($ (7) !.5! (8) $);

\node[state] (10) [below right of=6]{$9$};
\node[state] (11) [right of=10]{$10$};
\node[state] (12)[right of=11]{$11$};
\node[teal,left of=10,xshift=1.5cm]{1};
\draw[orange,very thick,->,snake=snake](10)--node[midway,black,xshift=2mm,yshift=3mm]{1}(11);
\draw[blue,very thick,->](11)--node[midway,black,yshift=2mm]{1}(12);
\path [orange,-,draw,dashed,thick] (7) -- ($ (10) !.5! (11) $);

\node[state] (13)[below right of=10]{$12$};
\node[state] (14)[right of=13]{$13$};
\node[state,accepting] (15)[right of=14]{$\frac{r_1}{r_2}$};
\draw[orange,very thick,->,snake=snake](13)--node[midway,black,xshift=1.5mm,yshift=3mm]{1}(14);
\draw[blue,very thick,->](14)--node[midway,black,yshift=2mm]{1}(15);
\path[orange,-,draw,dashed,thick] (11) -- ($ (13) !.5! (14) $);
\node[teal,left of =13,xshift=1.5cm]{1};

\end{tikzpicture}
\caption{Division gadget $g_{/}$. Note that this gadget does not satisfy non-degeneracy. We present it for the sake of keeping the exposition simple, but it is not formally used in the reduction of Theorem \ref{thm:fixp}, as this reduction should result in a non-degenerate financial system. 
In Appendix \ref{apx:nondegenerate}, we show that it is possible to change the basis and avoid using this gadget altogether.}
    \label{fig:divisiongadget}
\end{figure}
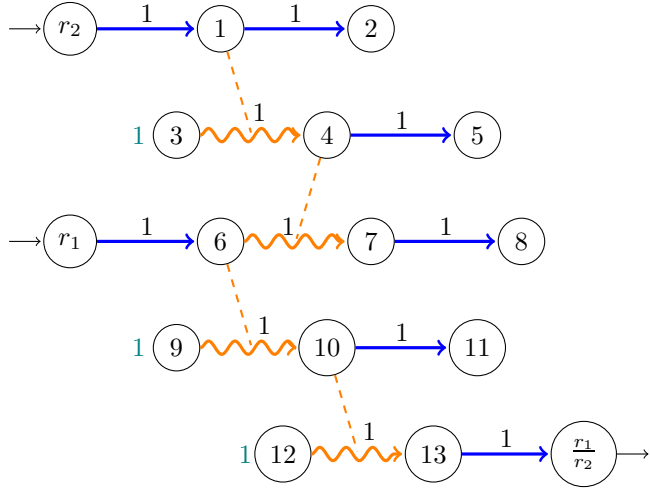

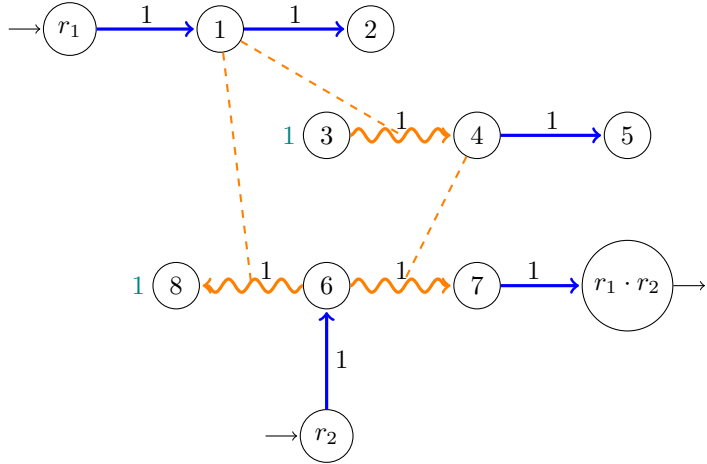
\begin{figure}[htbp!]
    \centering
\begin{tikzpicture}
[shorten >=1pt,node distance=2cm,initial text=]
\tikzstyle{every state}=[draw=black!50,very thick]
\tikzset{every state/.style={minimum size=0pt}}
\tikzstyle{accepting}=[accepting by arrow]
\node[state,initial] (1) {$r_1$};
\node[state] (2)[right of=1]{$1$};
\node[state] (11)[right of=2]{$2$};
\draw[blue,very thick,->](1)--node[midway,black,yshift=2mm]{1}(2);
\draw[blue,very thick,->](2)--node[midway,black,yshift=2mm]{1}(11);

\node[state] (3)[below right of=2]{$3$};
\node[state] (4)[right of=3]{$4$};
\node[state] (5)[right of=4]{$5$};
\node[teal,left of=3,xshift=1.5cm]{1};
\draw[blue,very thick,->](4)--node[midway,black,yshift=2mm]{1}(5);
\draw[orange,very thick,->,snake=snake](3)--node[midway,black,yshift=2mm]{1}(4);
\path[orange,-,draw,dashed,thick] (2) -- ($ (3) !.5! (4) $);

\node[state] (6)[below of=3]{$6$};
\node[state] (7)[right of=6]{$7$};
\node[state] (8)[left of=6]{$8$};
\node[state,accepting] (9)[right of=7]{$r_1 \cdot r_2$};
\draw[orange,very thick,->,snake=snake](6)--node[midway,black,yshift=2mm]{1}(7);
\draw[orange,very thick,->,snake=snake](6)--node[midway,black,xshift=2mm,yshift=2mm]{1}(8);
\draw[blue,very thick,->](7)--node[midway,black,xshift=-1mm,yshift=2mm]{1}(9);
\path[orange,-,draw,dashed,thick] (4) -- ($ (6) !.5! (7) $);
\path[orange,-,draw,dashed,thick] (2) -- ($ (6) !.5! (8) $);

\node[state,initial](10)[below of=6]{$r_2$};
\draw[blue,very thick,->](10)--node[midway,black,xshift=2mm]{1}(6);
\node[teal,left of =8,xshift=1.5cm]{1};
\end{tikzpicture}
\caption{Multiplication gadget $g_{*}$. Note that this gadget does not satisfy non-degeneracy. We present it for the sake of keeping the exposition simple, but it is not formally used in the reduction of Theorem \ref{thm:fixp}, as this reduction should result in a non-degenerate financial system. A more complex version of the gadget that is non-degenerate (i.e., the ``true'' version of the gadget) is given in Figure \ref{fig:nondegeneratemultiplicationgadget} in Appendix \ref{apx:nondegenerate}.}
    \label{fig:multiplicationgadget}
\end{figure}

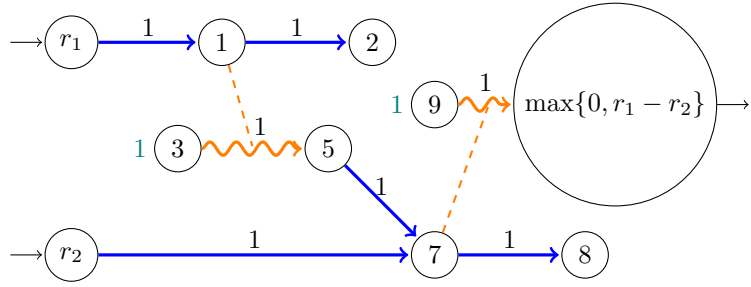
\begin{figure}[htbp!]

    \centering
\begin{tikzpicture}
[shorten >=1pt,node distance=2cm,initial text=]
\tikzstyle{every state}=[draw=black!50,very thick]
\tikzset{every state/.style={minimum size=0pt}}
\tikzstyle{accepting}=[accepting by arrow]
\node[state,initial] (1){$r_1$};
\node[state]         (2)[right of=1] {$1$};
\node[state]         (3)[right of=2]{$2$};
\draw[blue,very thick,->](1)--node[midway,black,yshift=2mm]{1}(2);
\draw[blue,very thick,->](2)--node[midway,black,yshift=2mm]{1}(3);

\node[state] (4)[below right of=1]{$3$};
\node[state] (5)[right of=4]{$5$};
\node[teal,left of=4,xshift=1.5cm]{1};
\draw[orange,-,very thick,->,snake=snake](4)--node[midway,black,yshift=3mm,xshift=1mm]{1}(5);
\path[orange,-,draw,dashed,thick] (2) -- ($ (4) !.5! (5) $);

\node[state,initial] (6)[below left of=4]{$r_2$};
\node[state] (7)[ below right of=5]{$7$};
\node[state](8)[right of=7]{$8$};
\draw[blue,->,very thick](7)--node[midway,black,yshift=2mm]{1}(8);
\draw[blue,->,very thick](6)--node[midway,black,yshift=2mm]{1}(7);
\draw[blue,very thick,->](5)--node[midway,black,yshift=2mm]{1}(7);

\node[state](9)[above of=7]{$9$};
\node[state,accepting](10)[right of=9,xshift=4mm]{$\max\{0,r_1 - r_2\}$};
\draw[orange,very thick,->,snake=snake](9)--node[midway,black,yshift=3mm]{1}(10);
\node[teal,left of=9,xshift=1.5cm]{1};
\path[orange,-,draw,dashed,thick] (7) -- ($ (9) !.3! (10) $);
\end{tikzpicture}
\caption{Positive subtraction gadget $g_{\text{pos}-}$. Generates the value $\max\{0,r_1-r_2\}$ from the input recovery rates $r_1$ and $r_2$.}
    \label{fig:subtractiongadget}
\end{figure}

\begin{figure}[htbp!]

    \centering
\begin{tikzpicture}
[shorten >=1pt,node distance=2cm,initial text=]
\tikzstyle{every state}=[draw=black!50,very thick]
\tikzset{every state/.style={minimum size=0pt}}
\tikzstyle{accepting}=[accepting by arrow]
\node[state] (1){$13$};
\node[state]         (2)[right of=1] {$14$};
\node[state]         (3)[right of=2]{$15$};
\draw[blue,very thick,->](1)--node[midway,black,yshift=2mm]{1}(2);
\draw[blue,very thick,->](2)--node[midway,black,yshift=2mm]{1}(3);

\node[state] (4)[below right of=1]{$16$};
\node[state] (5)[right of=4]{$17$};
\node[teal,left of=4,xshift=1.5cm]{1};
\draw[orange,-,very thick,->,snake=snake](4)--node[midway,black,yshift=3mm,xshift=2mm]{1}(5);
\path[orange,-,draw,dashed,thick] (2) -- ($ (4) !.5! (5) $);

\node[state] (6)[below left of=4]{$18$};
\node[state] (7)[ below right of=5]{$19$};
\node[state](8)[right of=7]{$20$};
\draw[blue,->,very thick](7)--node[midway,black,yshift=2mm]{1}(8);
\draw[blue,->,very thick](6)--node[midway,black,yshift=2mm]{1}(7);
\draw[blue,very thick,->](5)--node[midway,black,yshift=2mm]{1}(7);

\node[state](9)[below left of=7]{$21$};
\node[state](10)[below right of=8]{$22$};
\draw[orange,very thick,->,snake=snake](9)--node[midway,black,yshift=3mm]{1}(10);
\node[teal,left of=9,xshift=1.5cm]{1};
\path[orange,->,draw,dashed,thick] (7) -- ($ (9) !.5! (10) $);

\node[state] (11)[below of= 6,yshift=-2cm]{$23$};
\node[state] (12)[right of=11]{$24$};
\node[state] (13)[right of=12]{$25$};
\draw[blue,very thick,->](11)--node[midway,black,yshift=2mm]{1}(12);
\draw[blue,very thick,->](12)--node[midway,black,yshift=2mm]{1}(13);

\node[state] (14)[below right of=11]{$26$};
\node[state] (15)[right of=14]{$27$};
\node[teal,left of=14,xshift=1.5cm]{1};
\draw[orange,-,very thick,->,snake=snake](14)--(15);
\path[orange,-,draw,dashed,thick] (12) -- ($ (14) !.5! (15) $);

\node[state] (16)[below left of=14]{$28$};
\node[state] (17)[ below right of=15]{$29$};
\node[state](18)[right of=17]{$30$};
\draw[blue,->,very thick](17)--node[midway,black,yshift=2mm]{1}(18);
\draw[blue,->,very thick](16)--node[midway,black,yshift=2mm]{1}(17);
\draw[blue,very thick,->](15)--node[midway,black,yshift=2mm]{1}(17);

\node[state](19)[below left of=17]{$31$};
\node[state](20)[below right of=18]{$32$};
\draw[orange,very thick,->,snake=snake](19)--node[midway,black,yshift=3mm]{1}(20);
\node[teal,left of=19,xshift=1.5cm]{1};
\path[orange,->,draw,dashed,thick] (17) -- ($ (19) !.5! (20) $);

\node[state,accepting](21)[below right of=10,yshift=-2cm]{$\abs{r_1-r_2}$};
\draw[blue,very thick,->](10)--node[midway,black,xshift=1mm,yshift=2mm]{1}(21);
\draw[blue,very thick,->](20)--node[midway,black,xshift=-1mm,yshift=2mm]{1}(21);

\node[state,initial](22)[left of=1,xshift=-3cm]{$r_1$};
\node[state](23)[below left of=22]{$1$};
\node[state](24)[right of=23]{$2$};
\node[state](25)[right of=24]{$3$};
\node[state](26)[below of=23]{$4$};
\node[state](27)[right of=26]{$5$};
\node[state](28)[right of=27]{$6$};
\draw[blue,very thick,->](24)--node[midway,black,,yshift=2mm]{1}(25);
\draw[blue,very thick,->](27)--node[midway,black,yshift=2mm]{1}(28);
\draw[orange,very thick,->,snake=snake](23)--node[midway,black,xshift=-1mm,yshift=2mm]{1}(24);
\draw[orange,very thick,->,snake=snake](26)--node[midway,black,xshift=-2mm,yshift=2mm]{1}(27);
\path[orange,->,draw,dashed,thick] (22) -- ($ (23) !.5! (24) $);
\path[orange,->,draw,dashed,thick] (24) -- ($ (26) !.5! (27) $);

\node[state,initial](29)[left of=11,xshift=-3cm]{$r_2$};
\node[state](30)[below left of=29]{$7$};
\node[state](31)[right of=30]{$8$};
\node[state](32)[right of=31]{$9$};
\node[state](33)[below of=30]{$10$};
\node[state](34)[right of=33]{$11$};
\node[state](35)[right of=34]{$12$};
\draw[blue,very thick,->](31)--node[midway,black,yshift=2mm]{1}(32);
\draw[blue,very thick,->](34)--node[midway,black,yshift=2mm]{1}(35);
\draw[orange,very thick,->,snake=snake](30)--node[midway,black,xshift=-2mm,yshift=2mm]{1}(31);
\draw[orange,very thick,->,snake=snake](33)--node[midway,black,xshift=-1mm,yshift=2mm]{1}(34);
\path[orange,->,draw,dashed,thick] (29) -- ($ (30) !.5! (31) $);
\path[orange,->,draw,dashed,thick] (31) -- ($ (33) !.5! (34) $);

\draw[blue,very thick,->](22)--node[midway,black,yshift=2mm]{1}(1);
\draw[blue,very thick,->](29)--node[midway,black,yshift=2mm]{1}(6);
\draw[blue,very thick,->](35)--node[midway,black,yshift=3mm]{1}(11);
\draw[blue,very thick,->](28)--node[midway,black,xshift=2mm]{1}(16);
\node[teal,left of=23,xshift=1.5cm]{1};
\node[teal,left of=26,xshift=1.5cm]{1};
\node[teal,left of=30,xshift=1.5cm]{1};
\node[teal,left of=33,xshift=1.5cm]{1};

\end{tikzpicture}
\caption{Absolute difference gadget $g_{\text{abs}}$: The recovery rate of the output node is $\abs{r_1-r_2}$, where $r_1$ and $r_2$ are the recovery rates of the input nodes. The gadget is formed by first applying $g_{\text{dup}}$ on both its inputs. The positive difference gadget $g_{\text{pos}-}$ is then used on the two pairs of inputs to compute $\max\{r_1-r_2,0\}$ and $\max\{r_2-r_1,0\}$, after which these two maxima are added together using $g^+$, resulting in the desired output $|r_1 - r_2|$.}
    \label{fig:absdifferencegadget}
\end{figure}
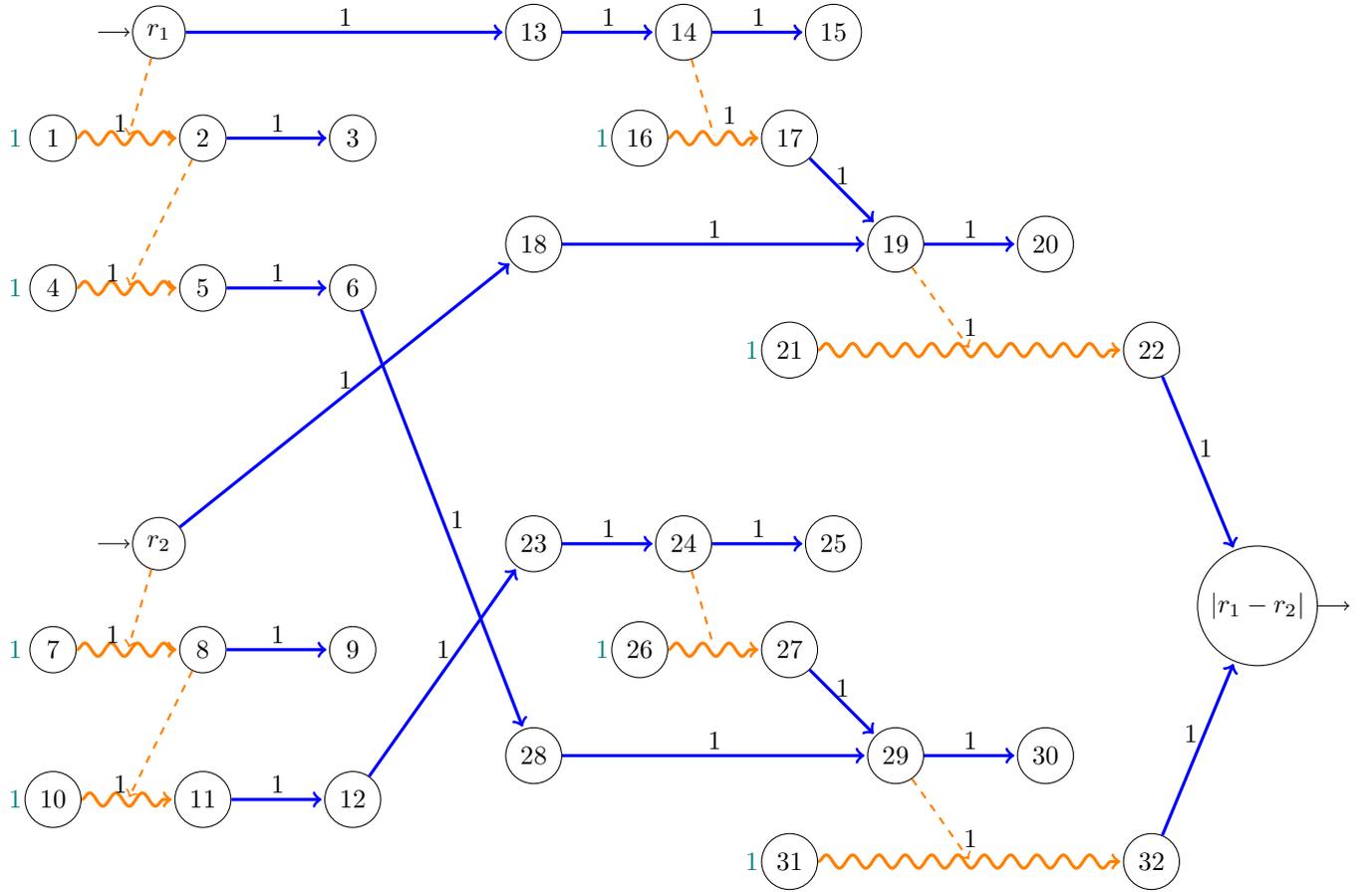

\newpage

\subsection{Non-degenerate Multiplication and Division Gadgets}\label{apx:nondegenerate}
This section describes how to adapt the hardness reduction of Theorem \ref{thm:fixp} in such a way that the financial system is non-degenerate. This is a rather technical detail and needs substantial work, yet it is needed for correctness of the reduction.

First, we present a non-degenerate version of our multiplication gadget in Figure \ref{fig:nondegeneratemultiplicationgadget}. In this gadget, the input recovery rates $r_1$ and $r_2$ are pre-processed by multiplying them both by $1/2$, using our constant multiplication gadget $g_{cr}$. The gadget then generates the recovery rate $(r_1(1+r_2))/4 \in [0,1]$, which can be post-processed using a positive subtraction gadget $g_{\text{pos}-}$ and a multiplication gadget $g_{*}$ to obtain $r_1 \cdot r_2$.

It is much more difficult to replace the degenerate gadget $g_{/}$. To do this, we will essentially get rid of division gates altogether in $C_I'$, and replace each of them by a set of alternative operations that achieve the same result. We thus make a few modifications to $C_I'$. The first modification is that we alter the sub-circuit $T$, which was used to generate the value $t = 1/2^{2^d}$, which was used to scale all the signals in the circuit so that each signal is in $[0,1]$ irrespective of the input vector. Here, $d$ is a polynomial time computable value that satisfies that every signal in the original circuit $C_I$ is at most $2^{2^d}$. We adapt $T$ such that the scaling factor it outputs is $1/(2^{1 + {2^d}})$ instead of $1/2^{2^d}$. This can be done by adding one additional multiplication gate at the end of $T$. Let the output of the modified $T$ be $t' = t/2$. After this modification, it holds that all signals inside the circuit are in $[0,1/2]$.

We now observe that division gates are used in a very limited way in $C_I'$ (with $T$ modified as above). 
\begin{itemize}
\item First, division is used for dividing by $c$, where $c$ is a given explicit constant in the original circuit $C_I$. For such divisions, we can simply use our constant multiplication gadget $g_{cr}$ to multiply by $1/c$, which is equivalent to dividing by $c$.
\item Secondly, division is used to divide certain outputs of gates by $t'$ (previously $t$), where $t' = 1/2^{1 + 2^d}$. This type of division happens in two cases. 
\begin{enumerate}
\item At every point in the circuit where two scaled signals $t'a$ and $t'b$ with $a,b \leq 2^{2^d}$ are multiplied with each other, resulting in the value $t't'ab$. This value is divided by $t'$ in order to generate the signal $t'ab$, i.e., a scaled version of the signal $ab$ in the original circuit. 
\item At the end of the circuit, where a scaled signal $t'a$ is divided by $t'$ to produce an output signal of the original circuit, which is in $[0,1]$.
\end{enumerate}
In the first case, let $x = t'ab$ and in the second case, let $x = a$, so that in both cases a number $t'x$ is divided by $t'$ to result in $x$, and in both cases it holds that $x \in [0,1]$. We replace the division $t'x/t'$ by a sequence of $d$ gates that compute the square root of its input, followed by a multiplication by $2$, followed by a sequence of $d$ successive squaring multiplication gates, followed by a final multiplication by $2$. This results in the correct value
\begin{equation*}
((t'x)^{1/2^d}\cdot 2)^{2^d} \cdot 2 = (t'^{1/2^d}x^{1/2^d}\cdot 2)^{2^d} \cdot 2 =
t'x \cdot 2^{2^d} \cdot 2 = 
\frac{1}{2}x \cdot 2 = x .
\end{equation*}
It is furthermore straightforward to verify that, due to the order in which we apply our arithmetic operations, all of the values throughout this computation stay in the interval $[0,1]$. We note that the adapted scaling factor $t' = t/2$ is needed because of the multiplication by $2$ that is executed after the successive square roots and before the successive squaring.
\end{itemize}

The resulting circuit has no division gates anymore, so we do not need a division gadget in our reduction. Instead, now we need a square root gadget. Fortunately, it is possible to design this gadget in a non-degenerate way,  although its construction is not very straightforward. It is presented in Figure \ref{fig:sqrtgadget}, and it combines the two subgadgets given in Figures \ref{fig:inversiongadget} and \ref{fig:constantsqrtgadget} which we provide separately for ease of understanding. The square root gadget may evidently output an irrational recovery rate, which is caused by the cyclic dependencies of the CDSes and debt contracts in the system.

\begin{figure}[htbp!]
    \centering
\begin{tikzpicture}
[shorten >=1pt,node distance=2cm,initial text=]
\tikzstyle{every state}=[draw=black!50,very thick]
\tikzset{every state/.style={minimum size=0pt}}
\tikzstyle{accepting}=[accepting by arrow]

\node[state] (1) {$r_1/2$};
\node[state] (2)[right of=1]{$1$};
\node[state] (11)[right of=2]{$2$};
\draw[blue,very thick,->](1)--node[midway,black,yshift=2mm]{1}(2);
\draw[blue,very thick,->](2)--node[midway,black,yshift=2mm]{1}(11);

\node[state,initial](12)[left of=1]{$g_{r_1/2}$};
\draw[blue,very thick,->](12)--(1);

\node[state] (3)[below right of=2]{$3$};
\node[state] (4)[right of=3]{$4$};
\node[state] (5)[right of=4]{$5$};
\node[teal,left of=3,xshift=1.5cm]{1};
\draw[blue,very thick,->](4)--node[midway,black,yshift=2mm]{1}(5);
\draw[orange,very thick,->,snake=snake](3)--node[midway,black,yshift=2mm]{1}(4);
\path[orange,-,draw,dashed,thick] (2) -- ($ (3) !.5! (4) $);

\node[state] (6)[below of=3]{$6$};
\node[teal,above of=6,yshift=-1.5cm]{1/2};
\node[state] (7)[right of=6]{$7$};
\node[state] (8)[left of=6]{$8$};
\node[state,accepting] (9)[right of=7]{$\frac{r_1 (1 +  r_2)}{4}$};
\draw[orange,very thick,->,snake=snake](6)--node[midway,black,yshift=2mm]{1}(7);
\draw[orange,very thick,->,snake=snake](6)--node[midway,black,xshift=2mm,yshift=2mm]{1}(8);
\draw[blue,very thick,->](7)--node[midway,black,xshift=-1mm,yshift=2mm]{1}(9);
\path[orange,-,draw,dashed,thick] (4) -- ($ (6) !.5! (7) $);
\path[orange,-,draw,dashed,thick] (2) -- ($ (6) !.5! (8) $);

\node[state](10)[below of=6]{$r_2/2$};
\draw[blue,very thick,->](10)--node[midway,black,xshift=2mm]{1}(6);
\node[teal,left of =8,xshift=1.5cm]{1};

\node[state,initial](13)[left of= 10]{$g_{r_2/2}$};
\draw[blue,very thick,->](13)--(10);

\end{tikzpicture}
\caption{Non-degenerate multiplication gadget $g_{*}'$.}
    \label{fig:nondegeneratemultiplicationgadget}
\end{figure}
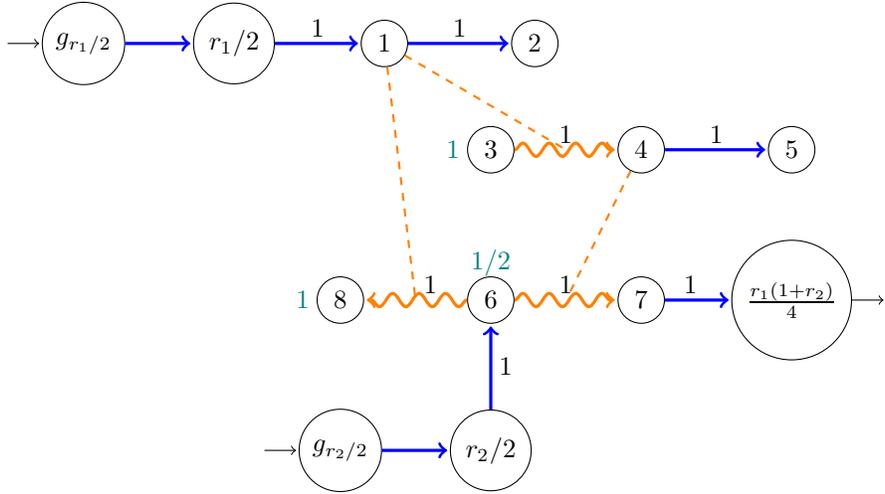

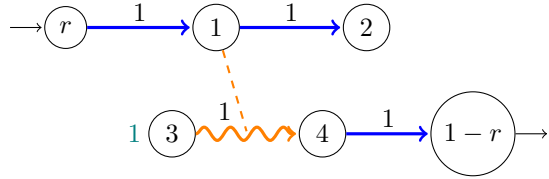
\begin{figure}[htbp!]
    \centering
\begin{tikzpicture}
[shorten >=1pt,node distance=2cm,initial text=]
\tikzstyle{every state}=[draw=black!50,very thick]
\tikzset{every state/.style={minimum size=0pt}}
\tikzstyle{accepting}=[accepting by arrow]
\node[state,initial] (1) {$r$};
\node[state]         (2) [right of=1] {$1$};
\node[state]         (3) [below right of=1] {$3$};
\node[state]         (6)[right of=2]{$2$};
\draw[blue,very thick,->](2)--node[midway,black,yshift=2mm]{1}(6);
\node[teal,left of=3,xshift=1.5cm]{1};
\node[state] (4) [right of=3] {$4$};
\node[state,accepting] (5)[right of=4]{$1-r$};
\draw[blue,->,very thick] (1)--node[midway,black,yshift=2mm]{1}(2);
\draw[orange,very thick,->,snake=snake] (3)--node[midway,black,xshift=-3mm,yshift=3mm]{1}(4);
\draw[blue,very thick,->](4)--node[midway,black,yshift=2mm]{1}(5);
\path [orange,-,draw,dashed,thick] (2) -- ($ (3) !.5! (4) $);
\end{tikzpicture}
\caption{Inversion gadget $g_{\text{inv}}$, computing $1-r$ from $r$.}
    \label{fig:inversiongadget}
\end{figure}

\begin{figure}[htbp!]
    \centering
\begin{tikzpicture}
[shorten >=1pt,node distance=2cm,initial text=]
\tikzstyle{every state}=[draw=black!50,very thick]
\tikzset{every state/.style={minimum size=0pt}}
\tikzstyle{accepting}=[accepting by arrow]
\node[state] (15)[right of=7][xshift=1mm]{$3$};
\node[state] (16)[below of=15]{$2$};
\node[state] (17)[below of=16]{$1$};
\node[state] (18)[below of=17]{$4$};
\draw[blue,very thick,->] (16)--node[midway,black,xshift=2mm]{1}(15);
\draw[blue,very thick,->] (17)--node[midway,black,xshift=2mm]{1}(16);
\draw[orange,very thick,->,snake=snake] (17)--node[midway,black,xshift=-2mm]{1}(18);
\node[state]  (19)[right of=15]{$9$};
\node[state]  (20)[below of=19]{$5$};
\node[state]  (21)[below of=20]{$6$};
\node[state]  (22)[below of=21]{$7$};
\draw[orange,very thick,->,snake=snake] (20)--node[midway,black,xshift=-3mm,yshift=1mm]{1}(19);
\draw[blue,very thick,->] (20)--node[midway,black,xshift=2mm]{1}(21);
\draw[blue,very thick,->] (21)--node[midway,black,xshift=2mm]{1}(22);
\path [orange,-,draw,dashed,thick] (16) -- ($ (20) !.5! (19) $);
\path [orange,-,draw,dashed,thick] (21) -- ($ (17) !.5! (18) $);
\node[state] (23)[right of=22]{$8$};
\draw[blue,very thick,->] (22)--node[midway,black,yshift=2mm]{1}(23);
\node[state] (24)[below of=18]{$10$};
\node[teal,left of=24,xshift=1.5cm]{1};
\node[state] (26)[right of=24]{$11$};
\node[state,accepting] (27)[right of=26]{$\sqrt c$};
\draw[orange,very thick,->,snake=snake] (24)--node[midway,black,yshift=3mm,xshift=-1mm]{1}(26);
\draw[blue,very thick,->] (26)--node[midway,black,yshift=3mm,xshift=-1mm]{1}(27);
\path [orange,-,draw,dashed,thick] (22) -- ($ (24) !.5! (26) $);
\node[teal,left of=17,xshift=1cm]{$1-c$};
\node[teal,right of=20,xshift=-1cm]{$1-c$};
\end{tikzpicture}
\caption{This gadget has no input nodes. It outputs the square root of a constant $c \in [0,1]$, where $1 - c$ is provided as the external assets of two banks in the gadget.}
\label{fig:constantsqrtgadget}
\end{figure}
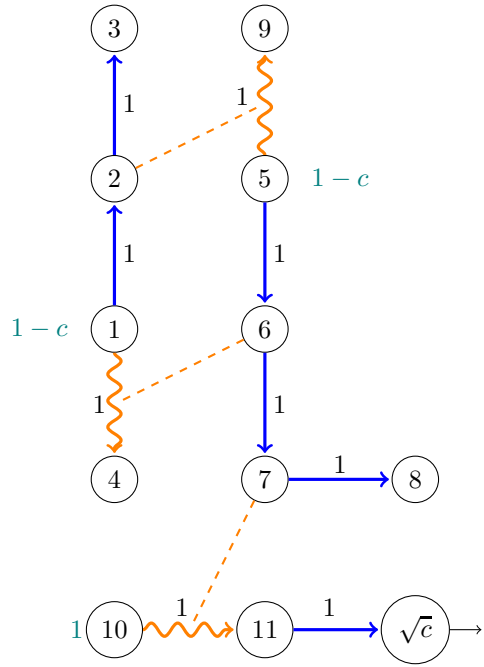

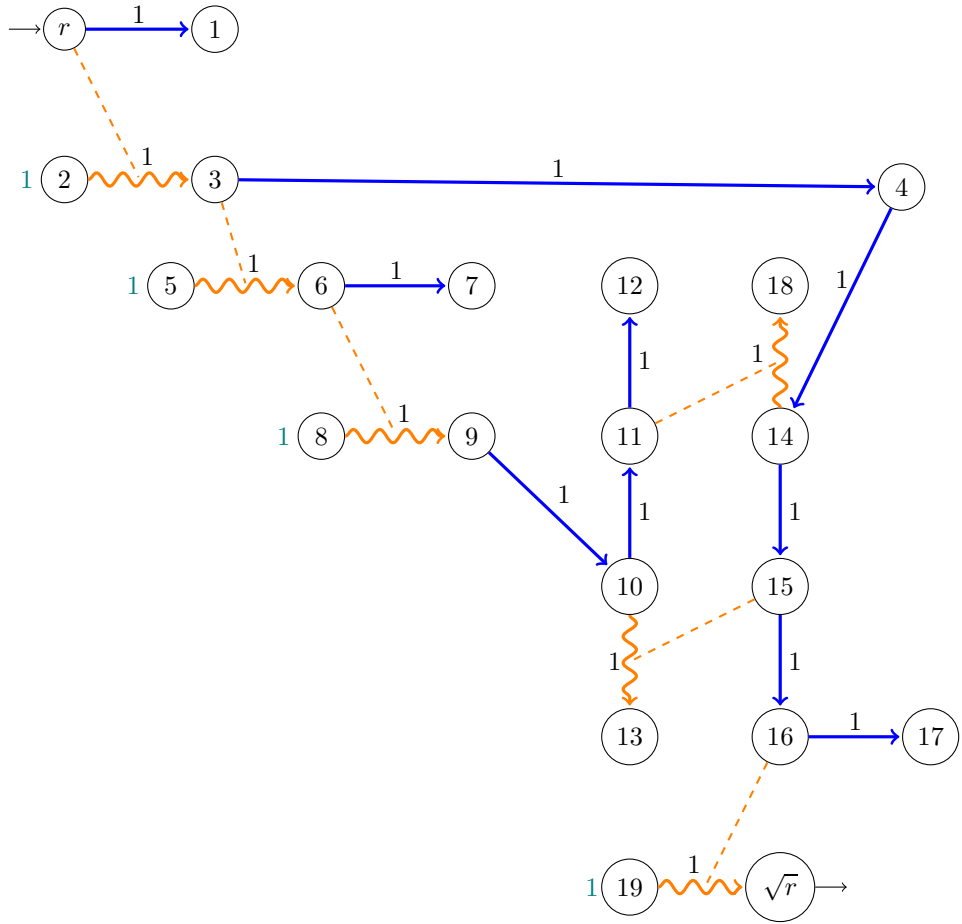
\begin{figure}[htbp!]
    \centering
\begin{tikzpicture}
[shorten >=1pt,node distance=2cm,initial text=]
\tikzstyle{every state}=[draw=black!50,very thick]
\tikzset{every state/.style={minimum size=0pt}}
\tikzstyle{accepting}=[accepting by arrow]

\node[state,initial] (1) {$r$};
\node[state]         (2) [right of=1] {$1$};
\draw[blue,very thick,->] (1)--node[midway,black,yshift=2mm]{1}(2);

\node[state]        (3) [below of=1]{$2$};
\node[teal,left of =3,xshift=1.5cm]{1};
\node[state]        (4) [right of=3]{$3$};
\draw [orange,very thick,->,snake=snake] (3)--node[midway,black,yshift=3mm,xshift=1mm]{1}(4);
\path [orange,-,draw,dashed,thick] (1) -- ($ (3) !.5! (4) $);

\node[state] (5)[below right of=3]{$5$};
\node[teal,left of=5,xshift=1.5cm]{1};
\node[state] (6)[right of=5]{$6$};
\node[state] (7)[right of=6]{$7$};
\draw[orange,very thick,->,snake=snake] (5)--node[midway,black,yshift=3mm,xshift=1mm]{1}(6);
\draw[blue,very thick,->](6)--node[midway,black,yshift=2mm]{1}(7);
\path [orange,-,draw,dashed,thick] (4) -- ($ (5) !.5! (6) $);

\node[state] (8)[below of=6]{$8$};
\node[teal,left of=8,xshift=1.5cm]{1};
\node[state]  (9)[right of=8]{$9$};
\draw[orange,very thick,->,snake=snake] (8)--node[midway,black,xshift=1mm,yshift=3mm]{1}(9);
\path[orange,-,draw,dashed,thick] (6) -- ($ (8) !.5! (9) $);

\node[state] (15)[right of=7][xshift=1mm]{$12$};
\node[state] (16)[below of=15]{$11$};
\node[state] (17)[below of=16]{$10$};
\node[state] (18)[below of=17]{$13$};
\draw[blue,very thick,->] (16)--node[midway,black,xshift=2mm]{1}(15);
\draw[blue,very thick,->] (17)--node[midway,black,xshift=2mm]{1}(16);
\draw[blue,very thick,->] (9)--node[midway,black,yshift=2mm,xshift=2mm]{1}(17);
\draw[orange,very thick,->,snake=snake] (17)--node[midway,black,xshift=-2mm]{1}(18);

\node[state]  (19)[right of=15]{$18$};
\node[state]  (20)[below of=19]{$14$};
\node[state]  (21)[below of=20]{$15$};
\node[state]  (22)[below of=21]{$16$};
\draw[orange,very thick,->,snake=snake] (20)--node[midway,black,xshift=-3mm,yshift=1mm]{1}(19);
\draw[blue,very thick,->] (20)--node[midway,black,xshift=2mm]{1}(21);
\draw[blue,very thick,->] (21)--node[midway,black,xshift=2mm]{1}(22);
\path [orange,-,draw,dashed,thick] (16) -- ($ (20) !.5! (19) $);
\path [orange,-,draw,dashed,thick] (21) -- ($ (17) !.5! (18) $);

\node[state] (23)[right of=22]{$17$};
\draw[blue,very thick,->] (22)--node[midway,black,yshift=2mm]{1}(23);

\node[state] (24)[below of=18]{$19$};
\node[teal,left of=24,xshift=1.5cm]{1};
\node[state,accepting] (26)[right of=24]{$\sqrt r$};

\draw[orange,very thick,->,snake=snake] (24)--node[midway,black,yshift=3mm,xshift=-1mm]{1}(26);
\path [orange,-,draw,dashed,thick] (22) -- ($ (24) !.5! (26) $);

\node[state] (8)[above right of=19][yshift=-1mm,xshift=2mm]{$4$};
\draw[blue,very thick,->] (4)--node[midway,black,yshift=2mm]{1}(8);
\draw[blue,very thick,->] (8)--node[midway,black,yshift=4mm]{1}(20);
\end{tikzpicture}
\caption{Square root gadget $g_{\sqrt{\cdot}}$, obtained by combining the duplication gadget $g_{\text{dup}}$, the inversion gadget $g_{\text{inv}}$, and the gadget of Figure \ref{fig:constantsqrtgadget}.}
    \label{fig:sqrtgadget}
\end{figure}

\section{Irrational Solutions for Weakly Switched Cycles}\label{apx:Counterexample}
Consider the instance in Figure \ref{fig:weaklyswitchedirrational}. This is an instance containing a weakly switched cycle that admits irrational solutions. Firstly the instance contains a red cycle which is $C = (1,2,3,R)$. Notice that $C$ contains one node that is switched on (which is Node 1) and one switched off node (Node 3). Thus, by definition, $C$ is weakly switched. Next we prove that through the coefficients shown in the figure, the instance has a unique irrational solution.

For Node 1, it holds that $r_1 = \frac{1}{3-r_R}$. Node 2 has a recovery equal to $r_2 = \frac{1-r_R}{3-r_R}$. Since node 3 has $e_3 = 1$ and a liability of at most 1, its recovery rate is equal to 1, so $r_3 = 1$. Now node $R$ has incoming payment from node 3 $a_R = 1-r_2 = 1-(1-r_R)/(3-r_R) = 2/(3-r_R)$ and a liability of $2$, so its recovery rate is $r_R = 1/(3-r_R)$. To compute the recovery rate of $R$ we must therefore solve the equation $r^2_R -3r_R + 1 = 0$, whose solution is $r_R = \frac{3-\sqrt5}{2}$. 

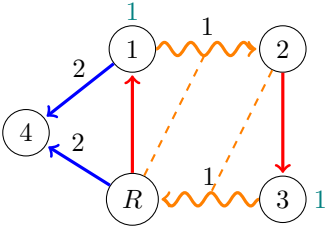
\begin{figure}[htbp!]
     \centering
\begin{tikzpicture}
[shorten >=1pt,node distance=2cm,initial text=]
\tikzstyle{every state}=[draw=black!50,very thick]
\tikzset{every state/.style={minimum size=0pt}}
\tikzstyle{accepting}=[accepting by arrow]

\node[state](1){$R$};
\node[state](2)[above of=1]{$1$};
\node[state](3)[right of=2]{$2$};
\node[state](4)[below of=3]{$3$};
\node[state](5)[below left of=2,yshift=3mm]{$4$};

\draw[blue,very thick,->](2)--node[midway,black,yshift=3mm]{2}(5);
\draw[blue,very thick,->](1)--node[midway,black,yshift=3mm]{2}(5);
\draw[orange,very thick,snake=snake,->](2)--node[midway,black,yshift=3mm]{1}(3);
\draw[orange,very thick,snake=snake,->](4)--node[midway,black,yshift=3mm]{1}(1);
\path [orange,-,draw,dashed,thick] (1) -- ($ (2) !.5! (3) $);
\path [orange,-,draw,dashed,thick] (3) -- ($ (4) !.5! (1) $);

\node[teal,above of=2,yshift=-1.5cm]{1};
\node[teal,right of=4,xshift=-1.5cm]{1};

\draw[red,very thick,->](1)--(2);
\draw[red,very thick,->](3)--(4);

\end{tikzpicture}
\caption{Example of an instance with a weakly (but not strongly) switched cycle having only irrational solutions.}
    \label{fig:weaklyswitchedirrational}
\end{figure}

\section{Further Financial System Gadgets}\label{apx:furthergadgets}
The gadgets presented in this section illustrate that financial systems are able to naturally capture the $\max$ and $\min$ operations. However, these gadgets remain unused in the reduction of Theorem~\ref{thm:fixp}.

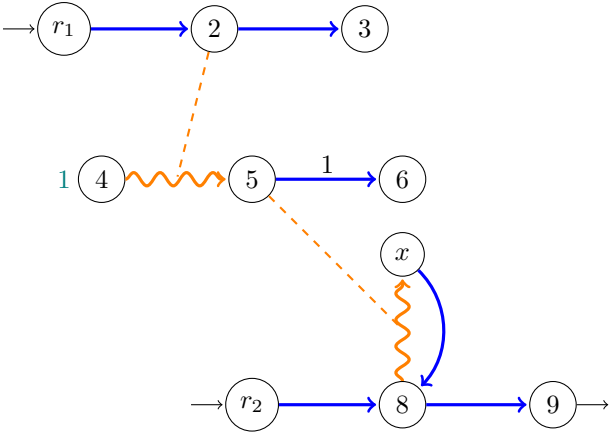
\begin{figure}[htbp!]
    \centering
\begin{tikzpicture}
[shorten >=1pt,node distance=2cm,bend angle=45,initial text=]
\tikzstyle{every state}=[draw=black!50,very thick]
\tikzset{every state/.style={minimum size=0pt}}
\tikzstyle{accepting}=[accepting by arrow]

\node[state,initial] (1) {$r_1$};
\node[state][right of =1](2){$2$};
\node[state][right of =2](3){$3$};
\draw[blue,very thick,->](1)--(2);
\draw[blue,very thick,->](2)--(3);

\node[state][below of =1,xshift=5mm](4){$4$};
\node[teal,left of=4,xshift=1.5cm]{1};
\node[state][right of =4](5){$5$};
\node[state][right of =5](6){$6$};
\draw[orange,very thick,->,snake=snake](4)--(5);
\draw[blue,very thick,->](5)--node[midway,black,yshift=2mm]{$1$}(6);

\node[state,initial][below of =5,yshift=-1cm](7){$r_2$};
\node[state][right of =7](8){$8$};
\node[state,accepting][right of =8](9){$9$};
\draw[blue,very thick,->](7)--(8);
\draw[blue,very thick,->](8)--(9);

\node[state][above of=8](10){$x$};
\draw[blue,very thick,->](10)[bend left] to (8);
\draw[orange,very thick,->,snake=snake](8) to(10);
\path [orange,-,draw,dashed,thick] (5) -- ($ (8) !.5! (10) $);
\path [orange,-,draw,dashed,thick] (2) -- ($ (4) !.5! (5) $);

\end{tikzpicture}
\caption{This is an alternative multiplication gadget where all contract notionals are assumed to be 1. This multiplication gadget was provided to us by an anonymous reviewer in the first draft of our paper. The gadget is non-degenerate, and is somewhat simpler than our non-degenerate multiplication gadget, at the trade-off of including a cycle in the gadget. We thank the reviewer for this contribution.}
    \label{fig:reviewer}
\end{figure}

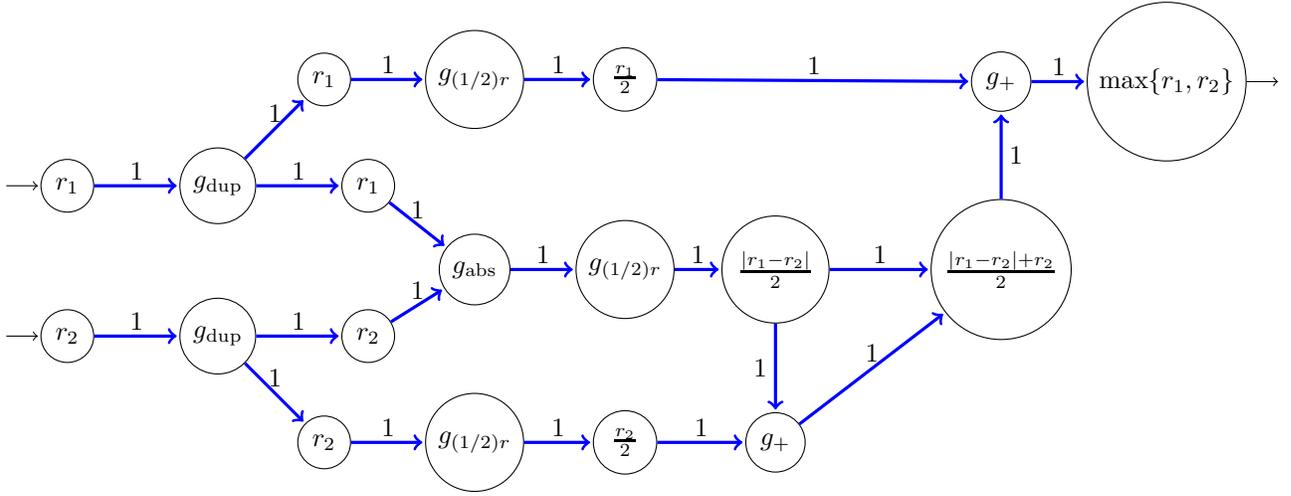
\begin{figure}[htbp!]

    \centering
\begin{tikzpicture}
[shorten >=1pt,node distance=2cm,initial text=]
\tikzstyle{every state}=[draw=black!50,very thick]
\tikzset{every state/.style={minimum size=0pt}}
\tikzstyle{accepting}=[accepting by arrow]

\node[state,initial] (1) {$r_1$};
\node[state,initial] (2)[below of=1]{$r_2$};
\node[state](3)[right of=1]{$g_{\text{dup}}$};
\node[state](4)[right of=2]{$g_{\text{dup}}$};
\draw[blue,very thick,->](1)--node[midway,black,yshift=2mm]{1}(3);
\draw[blue,very thick,->](2)--node[midway,black,yshift=2mm]{1}(4);

\node[state](5) [above right of=3]{$r_1$};
\node[state](6) [right of=3]{$r_1$};
\node[state](7) [below right of=4]{$r_2$};
\node[state](8) [right of=4]{$r_2$};
\draw[blue,very thick,->](3)--node[midway,black,yshift=2mm]{1}(5);
\draw[blue,very thick,->](3)--node[midway,black,yshift=2mm]{1}(6);
\draw[blue,very thick,->](4)--node[midway,black,yshift=2mm]{1}(7);
\draw[blue,very thick,->](4)--node[midway,black,yshift=2mm]{1}(8);

\node[state](9)[below right of=6,yshift=3mm]{$g_{\text{abs}}$};
\draw[blue,very thick,->](6)--node[midway,black,yshift=2mm]{1}(9);
\draw[blue,very thick,->](8)--node[midway,black,yshift=2mm]{1}(9);

\node[state](10)[right of=9]{$g_{(1/2)r}$};
\draw[blue,very thick,->](9)--node[midway,black,yshift=2mm]{1}(10);
\node[state](11)[right of=5]{$g_{(1/2)r}$};
\node[state](12)[right of=7]{$g_{(1/2)r}$};
\draw[blue,very thick,->](5)--node[midway,black,yshift=2mm]{1}(11);
\draw[blue,very thick,->](7)--node[midway,black,yshift=2mm]{1}(12);
\node[state](13)[right of=11]{$\frac{r_1}{2}$};
\node[state](14)[right of=12]{$\frac{r_2}{2}$};
\draw[blue,very thick,->](11)--node[midway,black,yshift=2mm]{1}(13);
\draw[blue,very thick,->](12)--node[midway,black,yshift=2mm]{1}(14);

\node[state](15)[right of=10]{$\frac{\abs{r_1-r_2}}{2}$};
\draw[blue,very thick,->](10)--node[midway,black,yshift=2mm]{1}(15);
\node[state](16)[below of=15,yshift=-3mm]{$g_+$};
\draw[blue,very thick,->](15)--node[midway,black,xshift=-2mm]{1}(16);
\draw[blue,very thick,->](14)--node[midway,black,yshift=2mm]{1}(16);
\node[state](17)[right of=15,xshift=1cm]{$\frac{\abs{r_1-r_2} + r_2}{2}$};
\draw[blue,very thick,->](16)--node[midway,black,yshift=2mm]{1}(17);
\draw[blue,very thick,->](15)--node[midway,black,yshift=2mm]{1}(17);
\node[state](18)[above of=17,yshift=5mm]{$g_+$};
\draw[blue,very thick,->](17)--node[midway,black,xshift=2mm]{1}(18);
\draw[blue,very thick,->](13)--node[midway,black,yshift=2mm]{1}(18);
\node[state,accepting](19)[right of=18,xshift=2mm]{$\max\{r_1,r_2\}$};
\draw[blue,very thick,->](18)--node[midway,black,yshift=2mm]{1}(19);
\end{tikzpicture}
\caption{Maximum gadget $g_{\max}$, computing $\max\{r_1,r_2\}$. This is a compact representation where the nodes labeled with a subscripted $g$ have to be replaced by copies of the respective gadgets, in order to obtain the full financial system defining the gadget.}
    \label{fig:maximumgadget}
\end{figure}

\begin{figure}[htbp!]

    \centering
\begin{tikzpicture}
[shorten >=1pt,node distance=2cm,initial text=]
\tikzstyle{every state}=[draw=black!50,very thick]
\tikzset{every state/.style={minimum size=0pt}}
\tikzstyle{accepting}=[accepting by arrow]

\node[state,initial] (1) {$r_1$};
\node[state,initial] (2)[below of=1]{$r_2$};
\node[state](3)[right of=1]{$g_{\text{dup}}$};
\node[state](4)[right of=2]{$g_{\text{dup}}$};
\draw[blue,very thick,->](1)--node[midway,black,yshift=2mm]{1}(3);
\draw[blue,very thick,->](2)--node[midway,black,yshift=2mm]{1}(4);

\node[state](5) [above right of=3]{$r_1$};
\node[state](6) [right of=3]{$r_1$};
\node[state](7) [below right of=4]{$r_2$};
\node[state](8) [right of=4]{$r_2$};
\draw[blue,very thick,->](3)--node[midway,black,yshift=2mm]{1}(5);
\draw[blue,very thick,->](3)--node[midway,black,yshift=2mm]{1}(6);
\draw[blue,very thick,->](4)--node[midway,black,yshift=2mm]{1}(7);
\draw[blue,very thick,->](4)--node[midway,black,yshift=2mm]{1}(8);

\node[state](9)[below right of=6,yshift=3mm]{$g_{\text{abs}}$};
\draw[blue,very thick,->](6)--node[midway,black,yshift=2mm]{1}(9);
\draw[blue,very thick,->](8)--node[midway,black,yshift=2mm]{1}(9);

\node[state](10)[right of=9]{$g_{(1/2)r}$};
\draw[blue,very thick,->](9)--node[midway,black,yshift=2mm]{1}(10);
\node[state](11)[right of=5]{$g_{(1/2)r}$};
\node[state](12)[right of=7]{$g_{(1/2)r}$};
\draw[blue,very thick,->](5)--node[midway,black,yshift=2mm]{1}(11);
\draw[blue,very thick,->](7)--node[midway,black,yshift=2mm]{1}(12);

\node[state](15)[right of=10]{$\frac{\abs{r_1-r_2}}{2}$};
\draw[blue,very thick,->](10)--node[midway,black,yshift=2mm]{1}(15);
\node[state](13)[above of=15,yshift=5mm]{$\frac{r_1}{2}$};
\node[state](14)[below  of=15,yshift=-3mm]{$\frac{r_2}{2}$};
\draw[blue,very thick,->](11)--node[midway,black,yshift=2mm]{1}(13);
\draw[blue,very thick,->](12)--node[midway,black,yshift=2mm]{1}(14);

\node[state](16)[right of=15,yshift=6mm,xshift=3mm]{$g_+$};
\draw[blue,very thick,->](13)--node[midway,black,yshift=2mm]{1}(16);
\draw[blue,very thick,->](14)--node[midway,black,xshift=-2mm,yshift=1mm]{1}(16);
\node[state](17)[right of=16]{$\frac{r_1+r_2}{2}$};
\draw[blue,very thick,->](16)--node[midway,black,yshift=2mm]{1}(17);
\node[state](18)[below of=17,yshift=-5mm]{$\max\{0,r_1-r_2\}$};
\draw[blue,very thick,->](17)--node[midway,black,xshift=2mm]{1}(18);
\draw[blue,very thick,->](15)--node[midway,black,yshift=2mm]{1}(18);
\node[state,accepting](19)[below of=18,yshift=-1cm]{$\min\{r_1,r_2\}$};
\draw[blue,very thick,->](18)--node[midway,black,xshift=2mm]{1}(19);
\end{tikzpicture}
\caption{Minimum gadget $g_{\min}$, computing $\min\{r_1,r_2\}$. This is a compact representation where the nodes labeled with a subscripted $g$ have to be replaced by copies of the respective gadgets, in order to obtain the full financial system defining the gadget.}
    \label{fig:minimumgadget}
\end{figure}
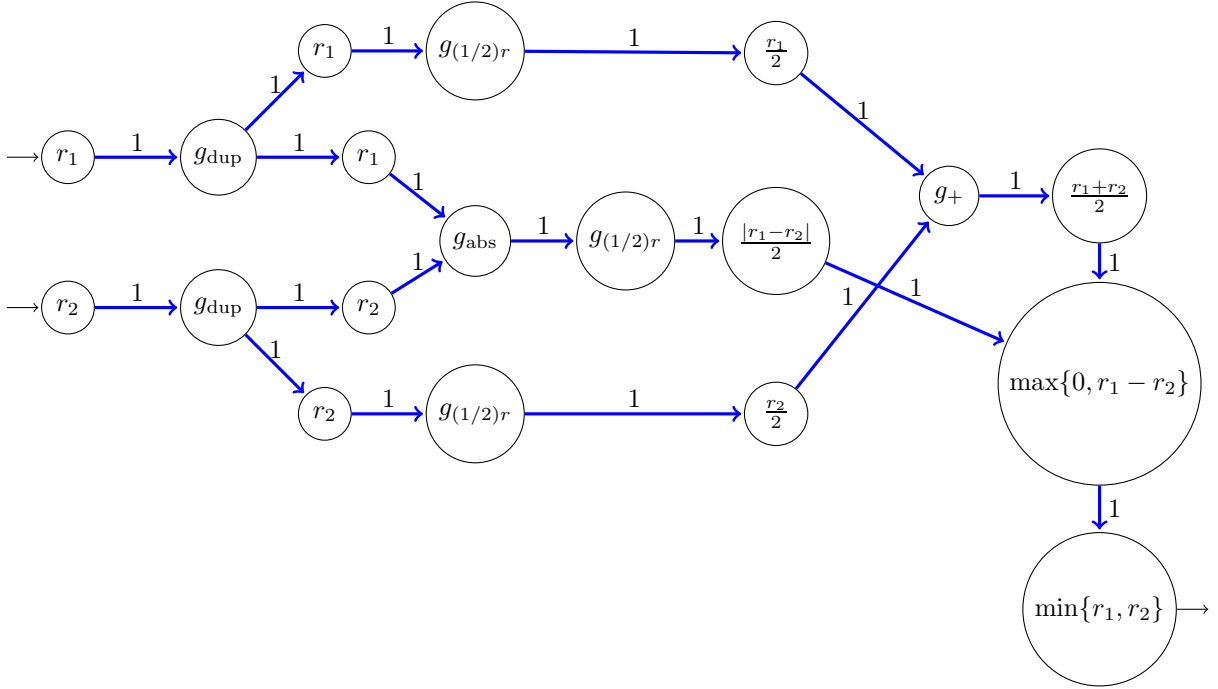

\end{document}